\documentclass{article}

\RequirePackage{amsthm,amsmath}
\RequirePackage[numbers]{natbib}
\RequirePackage[colorlinks,citecolor=blue,urlcolor=blue]{hyperref}
\RequirePackage{graphicx}
\usepackage{caption}
\usepackage{xr}
\usepackage{hyperref}
\usepackage{amssymb}
\usepackage{bbm}
\usepackage{bm}
\usepackage{caption}
\usepackage{listings}
\usepackage{bbm}
\usepackage{mathrsfs}
\usepackage{todonotes}
\usepackage{verbatim}
\usepackage{subcaption}
\usepackage{float}
\usepackage{hyperref}

\usepackage{arxiv}

\usepackage[utf8]{inputenc} % allow utf-8 input
\usepackage[T1]{fontenc}    % use 8-bit T1 fonts
\usepackage{hyperref}       % hyperlinks
\usepackage{url}            % simple URL typesetting
\usepackage{booktabs}       % professional-quality tables
\usepackage{amsfonts}       % blackboard math symbols
\usepackage{nicefrac}       % compact symbols for 1/2, etc.
\usepackage{microtype}      % microtypography
\usepackage{lipsum}		% Can be removed after putting your text content
\usepackage{graphicx}
\usepackage{natbib}
\usepackage{doi}

\numberwithin{equation}{section}
\theoremstyle{definition}
\newtheorem {definition}{Definition}[section]
\theoremstyle{remark}
\newtheorem {remark}{Remark}[section]

\theoremstyle{plain}
\newtheorem{theorem}{Theorem}[section]
\newtheorem{proposition}{Proposition}[section]

\title{Combining Stochastic Tendency and Distribution Overlap Towards Improved Nonparametric Effect Measures and Inference}

\author{ Jonas Beck \\
	Department of Artificial Intelligence   and Human Interfaces, \\ 
 Paris Lodron University of Salzburg, \\
 Hellbrunner Straße 34, 5020 Salzburg, Austria\\
	\texttt{jonas.beck@plus.ac.at} \\
	\And
	{Patrick B. Langthaler} \\
	Department of Artificial Intelligence and Human Interfaces, \\ 
 Paris Lodron University of Salzburg, \\
 Hellbrunner Straße 34, 5020 Salzburg, Austria\\
Department of Neurology, \\
Paracelsus Medical University of Salzburg, \\
Ignaz-Harrer-Straße 79, 5020 Salzburg, Austria\\
	\texttt{patrickbenjamin.langthaler@stud.plus.ac.at} \\
 \And
	{Arne C. Bathke} \\
	Department of Artificial Intelligence and Human Interfaces, \\ 
 Paris Lodron University of Salzburg, \\
 Hellbrunner Straße 34, 5020 Salzburg, Austria\\
	\texttt{arne.bathke@plus.ac.at} \\
}

\date{}

\renewcommand{\headeright}{}
\renewcommand{\undertitle}{}
\renewcommand{\shorttitle}{Combining Stochastics Tendency and Distribution Overlap}

\begin{document}
\maketitle

\begin{abstract}
	A fundamental functional in nonparametric statistics is the Mann-Whitney functional $\theta=P(X<Y)$ , which constitutes the basis for the most popular nonparametric procedures.  The functional $\theta$ measures a location or stochastic tendency effect between two distributions. A limitation of $\theta$ is its inability to capture scale differences. If differences of this nature are to be detected, specific tests for scale or omnibus tests need to be employed. However, the latter often suffer from low power, and they do not yield interpretable effect measures. In this manuscript, we extend $\theta$ by additionally incorporating the recently introduced distribution overlap index (nonparametric dispersion measure) $I_2$ that can be expressed in terms of the quantile process. We derive the joint asymptotic distribution of the respective estimators of $\theta$ and $I_2$ and construct confidence regions. Extending the Wilcoxon-Mann-Whitney test, we introduce a new test based on the joint use of these functionals. It results in much larger consistency regions while maintaining competitive power to the rank sum test for situations in which $\theta$  alone would suffice. Compared with classical omnibus tests, the simulated power is much improved. Additionally, the newly proposed inference method yields effect measures whose interpretation is surprisingly straightforward.
\end{abstract}

% keywords can be removed

\section{Introduction}
Statistical functionals are omnipresent in nonparametric statistics, especially the Mann- Whitney functional (or relative effect) $P(X<Y)$, which is widely being used for different statistical inference procedures, more recently also including probabilistic index models (PIM) (\cite{thas2012probabilistic}, \cite{de2015regression}) and confidence distributions \citep{beck_23}. Statistical theory for the Mann-Whitney parameter has a long history, going back to the first introduction by \cite{mann_whitney} and \cite{wilcoxon}, the development of asymptotic theory for the nonparametric Behrens-Fisher problem (in the case of no ties) by \cite{fligner_policello}, and the extension to the general case allowing for ties by \cite{brunner_munzel}.
The major limitation of all these relative effect based inference procedures is their inability of measure any kind of scale differences between distributions. Our new approach provides a solution to this limitation.

To this end, we incorporate the overlap index, which was recently developed by \cite{parkinson2018fast} to quantify the overlap of Hutchinsonian niches \citep{hutchinson1957population} that represent the ecological requirements for species survival in a particular habitat or their impact niche based on functional traits. These Hutchinsonian niches have various applications in ecology, animal behavior science, evolution, and other fields, and can be used to define the states of ecosystems and assess their shifts after environmental change. For more details and a short review on the concept of ecological niches we refer to \cite{swanson_15} and \cite{blonder2018new}. 

The overlap index is defined for two independent random variables (or distributions, $X \sim F$, $Y \sim G$) as 
\begin{align*}
I_2 &= \int_0^1 G \circ F^{-1}(1- \alpha/2)d \alpha - \int_0^1 G \circ F^{-1}(\alpha/2) d \alpha  \\
&= 2 (\int_{F^{-1}(1/2)}^{\infty} G dF - \int^{F^{-1}(1/2)}_{- \infty} G d F ).
\end{align*}
It was introduced to provide a nonparametric approach to estimating and inferring the overlap of hypervolumes without imposing a particular model or distribution family.
The niche overlap can be seen as a measure of how much of the mass assigned by $F$ is ``contained'' in the range of $G$.
Besides its obvious applicability in ecology and economics, we believe the overlap index constitutes a generally very informative nonparametric dispersion measure.

Indeed, we are convinced that simultaneous inference on both functionals provides highly  useful results for one of the most common statistical tasks: Comparing two independent samples and testing them for equality of the underlying data-generating distributions. 
Classical nonparametric omnibus test such as the Kolmogorov-Smirnov \citep{smirnov}, the Cramer-von Mises \citep{anderson}, or the Anderson-Darling test \citep{anderson_darling} are known to be limited by low power, especially for small and moderate sample sizes \citep{hollander2013nonparametric}. 
On the other hand, the popular asymptotic Wilcoxon-Mann-Whitney (WMW) test, 
although technically testing the null hypothesis of distribution equality, has a consistency region that is defined entirely by the 
relative effect  \citep{brunnerbathkekonnietschke} and thus by design cannot detect scale differences. 
Thus, while the two-sample problem is treated in most statistical textbooks (Section 6.9 in \citep{Lehmann2022}, Section 3.7 in \citep{van1996weak} and Chapter 6 in \citep{gibbons2020nonparametric}) and continues to receive research attention   
(\cite{JMLR:v13:gretton12a}, \cite{Ramdas2015OnWT} and \cite{clemenccon2023bipartite}), 
there is yet, to the best of our knowledge, no test whose  performance for small and moderate  sample sizes matches the strength of the WMW Test, while at the same time being consistent for a larger class of alternatives and yielding easily interpretable effect measures without having to assume restrictive model classes.

 The new \textit{almost omnibus} nonparametric inference approach  for two independent samples proposed in this paper includes a new test for equality of distributions that addresses this problem. The procedure is based on simultaneous inference for the two previously introduced nonparametric functionals, utilizing the asymptotic joint distribution of their estimators that is derived below. It has several advantages over existing methods. First, it is \textit{almost omnibus} in the sense that it can be used to test for a wide range of alternative hypotheses. Second, it is based on a combination of different \textit{nonparametric} methods, making it more robust and less sensitive to specific assumptions that underlie for example a \textit{semiparametric} location scale model as in the Lepage model \citep{lepage} which itself is an extension of the Ansari-Bradley test \citep{ansari_bradley}. 
Third, the test statistic corresponds to effect measures with a straightforward interpretation.
Finally, in simulations, the proposed procedure has demonstrated good power properties, even for small sample sizes.  We believe that this \textit{almost omnibus} inference helps to find a balance between the high power of the WMW test (especially for small samples) and the omnibus property of the Kolmogorov-Smirnov test.

In the following, we will first present a short review on the two nonparametric effect measures that form the foundation of our procedure. Then, we will look at the bivariate range of the functionals by deriving  the exact joint region of the relative effect and the overlap index. 
Next, we will prove the joint asymptotic normality of both functionals by utilizing their close relationship to the empirical ROC-process \citep{hsieh_turnbull}. 
When using this result for developing a test statistic for distribution equality, we realized that both coordinates of our bivariate test statistic have the same consistent variance estimator as the classical Mann-Whitney test. 
In order to derive simultaneous confidence intervals, we will use a resampling approach and compare different algorithms. 
As the theory of empirical processes assumes continuous distributions, we will provide an additional proof for the asymptotic distribution without restricting us to continuous distributions. 
To this end, we will use a rank-based approach incorporating ideas of \cite{brunner2017rank}. 
Finally, we will present results from an extensive simulation study, showing especially a rather good performance for small and medium size samples, and demonstrating the practical applicability by using a real-life data example from a current  epidemiological study.

\section{Nonparametric Effect Measures} \label{sec2}
If not mentioned otherwise, we will consider independent and absolutely continuous (in order to be able to use certain theorems below, and to allow for simplified notation) random variables $X \sim F$ and $Y \sim G$. 
Let us first define the relative effect or probabilistic index.
\begin{definition}
For two independent random variables $X \sim F$ and $Y \sim G$, we call
\begin{align}
\theta:= P(X <Y)= \int G dF = \int_0^1 G \circ F^{-1} (\alpha) d \alpha
\label{def:rte}
\end{align}
the \emph{(nonparametric) relative effect (probabilistic index, Mann-Whitney functional)} of $Y$ ($G$) with respect to $X$ ($F$).
\end{definition}
From this definition it follows that a relative effect 
$ \theta < 1/2$ indicates that the observations from $F$ tend to have larger values than those from $G$, and vice versa for $ \theta >1/2$. 
A value of $\theta=1/2$ may be interpreted as no tendency to larger or smaller values at all. It is clear that the relative effect equals $1/2$ if the distribution $F$ and $G$ are equal.
For distributions that are not absolutely continuous, the definition in (\ref{def:rte}) is adapted to $\theta=P(X<Y)+P(X=Y)/2$.

As we want to quantify not only the location (or tendency to larger values), but also the dispersion (or tendency to cover a larger niche) we also use the overlap index which was first introduced by \cite{parkinson2018fast}:
\begin{definition}
For two independent random variables $X \sim F$ and $Y \sim G$ we call
\begin{align}
I_2 &= \int_0^1 G \circ F^{-1}(1- \alpha/2)d \alpha - \int_0^1 G \circ F^{-1}(\alpha/2) d \alpha \label{def:niche}  \\
&= 2 (\int_{F^{-1}(1/2)}^{\infty} G dF - \int^{F^{-1}(1/2)}_{- \infty} G d F ) \nonumber
\end{align}
the \emph{(asymmetric) overlap index (niche overlap)} of $G$ with respect to $F$ (and $I_1$ analogously by switching the roles of $F$ and $G$).
\end{definition}
 The overlap index can be seen as a measure of how much of the mass assigned by $F$ is ``contained'' in the range of $G$.
 If the distributions are equal, the overlap index equals $1/2$, just as the relative effect. 
 \begin{remark} \label{alt_i2}
 We define the random variables $X^{(1)} \sim F_1$ and $X^{(2)} \sim F_2$ by conditioning $X$ on $X \leq F^{-1} (0.5)$ and $X > F^{-1} (0.5)$, respectively:
\begin{align*}
F_1(x)= 
\begin{cases}
2 F(x), \ &x < F^{-1} (0.5) \\
1,      &x \geq F^{-1} (0.5)
\end{cases}
\  \  \
F_2(x)= 
\begin{cases}
0, \ &x < F^{-1} (0.5) \\
2F(x)-1,      &x \geq F^{-1} (0.5)
\end{cases}
\end{align*}
$Y^{(1)} \sim G_1$ and $Y^{(2)} \sim G_2$ are analogously defined. 

Using this notation, the relative effect $\theta$ and the overlap index $I_2$ can be reformulated as
\begin{equation} \label{theta}
\theta= \frac{1}{4} [ P(X^{(1)} < Y^{(1)})+P(X^{(1)} < Y^{(2)})+P(X^{(2)} < Y^{(1)})+P(X^{(2)} < Y^{(2)})]
\end{equation}
and
\begin{equation} \label{I2}
I_2= \frac{1}{2} [ P(X^{(1)} < Y^{(2)})+P(X^{(2)} < Y^{(2)})-P(X^{(1)} < Y^{(1)})-P(X^{(2)} < Y^{(1)})] ~.
\end{equation}
%and $I_1$ analogously.
 \end{remark}
Canonical ``plug-in'' estimators based on independent $X$- and $Y$-samples of size $n$ and $m$, respectively, can be defined as follows. 
For the relative effect,
 \begin{align*}
\hat{\theta}_{mn}:=\int \hat{G}_m d \hat{F}_n = \int_0^1 \hat{G}_m \circ \hat{F}_n^{-1} (\alpha) d \alpha,
\end{align*}
and for the overlap index
\begin{align*}
\hat{I}_{2,mn} &= \int_0^1 \hat{G}_m \circ \hat{F}_n^{-1}(1- \alpha/2)d \alpha - \int_0^1 \hat{G}_m \circ \hat{F}_n^{-1}(\alpha/2) d \alpha \\
&= 2 (\int_{\hat{F}_n^{-1}(1/2)}^{\infty} \hat{G}_m d\hat{F}_n - \int^{\hat{F}_n^{-1}(1/2)}_{- \infty} \hat{G}_m  d \hat{F}_n ),
\end{align*} 
where the empirical quantile is defined as $ \hat{F}_n^{-1}(t)= \inf  \{ y: \hat{F}_n(y) \geq t \} $. 
The plug-in estimator for $I_1$ can be defined analogously.

As shown in \cite{parkinson2018fast} and \cite{brunnerbathkekonnietschke} these estimators are consistent, and the estimator of the relative effect is unbiased. 
%Additionally we will assume that $X_i \sim F$ and $Y_j \sim G $ are independent random variables.
If not otherwise mentioned we use here the same notation as in \cite{parkinson2018fast}.
To reformulate our estimator $\hat{I}_2$ we assume without loss of generality that $X_1 \leq X_2 \leq \ldots \leq X_n$. We denote by $K$ the largest integer smaller or equal to $(n+1)/2$ and all the $X$-observations below their median by $X_1, \ldots, X_K$. 
Their ranks are denoted by $R_1^{X<}, \ldots, R_K^{X<}$, and by $R_{K+1}^{X>}, \ldots, R_n^{X>}$ the ranks of the remaining $X$-observations. The corresponding rank sums are defined by 
$ %$\begin{align*}
R_{\cdot}^{X<} = \sum_{i=1}^K R_i^{X<}$ and $ R_{\cdot}^{X>} = \sum_{i=K+1}^n R_i^{X>}.
$ %\end{align*}
Our estimator $\hat{I}_{2,mn}$ for the niche overlap can then, as shown in Lemma 2.11 in \citep{parkinson2018fast}, be rearranged to 
\begin{align*}
\frac{2}{mn} (R_{\cdot}^{X>}-R_{\cdot}^{X<} ) + \frac{1}{2} c,
\end{align*}
where $c=- n/m$ for $n$ even and $c \approx -n/m$ for $n$ odd. As we restrict ourselves here mostly to continuous distributions (for an alternative method, see Section 
\ref{test-sim}), we can ignore the possibility of ties. Extensions to this case were briefly mentioned in Appendix A1 in \citep{parkinson2018fast}. \\
\cite{brunner_puri} already showed that the plug-in estimator $\hat{\theta}_{mn}$ can be rearranged as
\begin{align} \label{repr_theta}
 \hat{\theta}_{mn} &= \frac{1}{n} \bigl( \bar{R}_{2 \cdot} - \frac{m +1}{2} \bigl ) = 1- \frac{1}{m} \bigl( \bar{R}_{1 \cdot} - \frac{n +1}{2} \bigl ) = \frac{1}{N} (\bar{R}_{2 \cdot}-\bar{R}_{1 \cdot})+ \frac{1}{2},
\end{align}
where  $\bar{R}_{i \cdot}= \frac{1}{n_i} \sum_{i=1}^{ni} R_{ik}$
 for $i=1,2$ is the \emph{rank mean} for the respective group, $ R_{ik}$ the (overall) rank, and $N=m+n$. 
 
\begin{remark}
The following results regarding the overlap index have already been established in the literature.
\begin{enumerate}
\item The sum $I_1+I_2 \in [0,1]$ (Lemma 2.10 in \cite{parkinson2018fast})
    \item  If the medians of $F$ and $G$ are equal, it holds that $I_1 +I_2=1$ (Lemma 2.19 in \citep{parkinson2018fast}). 
\item Under the assumption of continuous distribution functions and equal medians, the previously defined estimators of $I_1$ and $I_2$ are finitely biased in the sense that $\mathbb{E}[\hat{I}_1 + \hat{I}_2] <1$ 
(\cite{parkinson2022testing}).
\end{enumerate}
\end{remark}

\section{Simultaneous Inference for the Relative Effect and the Overlap Index}
In this chapter, we will first consider the bivariate range of the previously introduced effect measures. Next, we will prove the asymptotic normality of their consistent estimators under some weak conditions.
\subsection{Bivariate Range of the Effect Measures}
It is obvious from the integral representation of the overlap indices $I_1, I_2$ in (\ref{def:niche}) and the relative effect $\theta$ in (\ref{def:rte}) that both functionals are related. Therefore, one may conjecture that not necessarily all values on $[0, 1] \times [0, 1]$ can be taken by the pair $(\theta, I_2)$, given two arbitrary distributions $F$ and $G$. Indeed, the following holds:
\begin{theorem} \label{thm_image}
    Let $S$ be the set of all cumulative distribution functions in an arbitrary probability space. Interpret $\theta$ and $I_2$ as functions from $S \times S$ to $[0, 1]$ and accordingly
    \begin{equation*}
        (\theta, I_2):\: S \times S \longrightarrow [0, 1] \times [0, 1].
    \end{equation*}
    Then the image of $(\theta, I_2)$ is $A := \{(x, y) | x \in [0, 1], y \in [0, \min\{2x, 2-2x\}]\}$
\end{theorem}
\begin{proof}
Appendix \ref{proofthm_image}
\end{proof}

\begin{figure}[ht!]
    \centering
    \includegraphics[width = 0.45\textwidth]{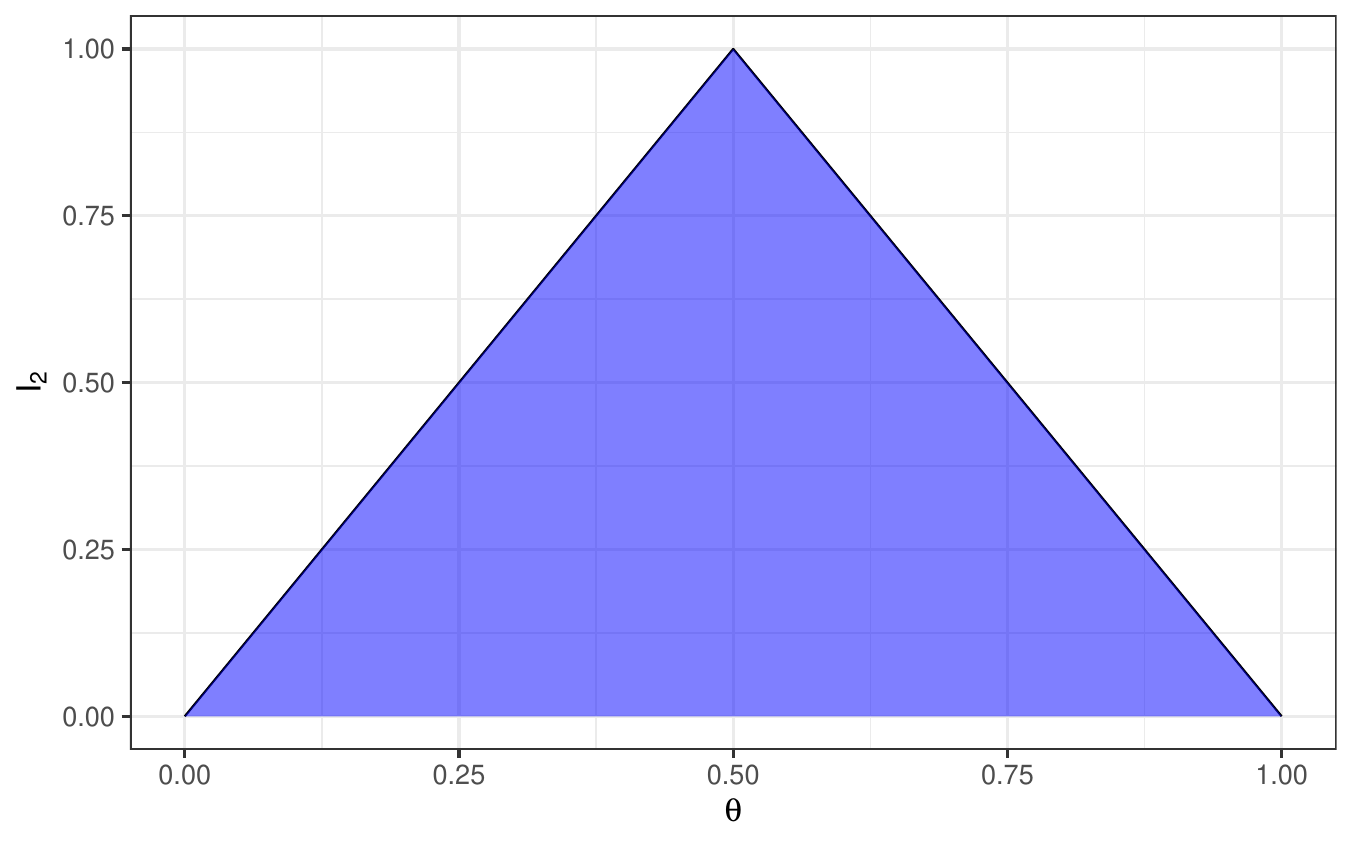}
    \caption{The blue shaded region shows the theoretically possible pairs of $(\theta, I_2)$.}
    \label{fig:theta_I_2_relationship}
\end{figure}
This can be reformulated as follows:
\begin{enumerate}
    \item All possible pairs $(\theta, I_2)$ must lie inside $A$ for any distribution functions $F$ and $G$.
    \item For every pair $(\theta, I_2) \in A$, there exists a pair of distribution functions $(F, G)$ such that their relative effect and overlap index are $\theta$ and $I_2$, respectively.
\end{enumerate}
Figure \ref{fig:theta_I_2_relationship} makes it clear that even if the relative effect $\theta$ takes the value $1/2$,
the overlap index $I_2$ can still take every value between $0$ and $1$. Taking advantage of this fact will provide us with a better way to distinguish the distributions than testing only one parameter.

\subsection{Joint Asymptotic Distribution of the Relative Effect and the Overlap Index} \label{sec_joint}
In order to effectively use the two functionals jointly for inference, one needs the joint (asymptotic) distribution of their estimators. This is what we will establish in this section.
\begin{remark}
A path towards establishing asymptotic normality could be to mimic proofs from rank statistics and approximate $I_2$ by a sum of independent random variables for which a version of the central limit theorem can be applied. However, we would first need to prove an analogue to the asymptotic equivalence theorem for the relative effect (see, e.g., \cite{brunnerbathkekonnietschke}). \\
Instead, and due to the generalizability of the approach, we use similarities to the empirical ROC-Process (\cite{hsieh_turnbull}, \cite{davidov_herman}), by incorporating the large and broad theory of empirical processes as presented in \cite{van1996weak}. 
Nevertheless, we will come back to the previously mentioned approach in section 
\ref{test-sim}, where we consider a special case, namely the null hypothesis of distribution equality.
\end{remark}

Thus, in the following, we will at first prove  asymptotic normality using similarities to the ROC-process, and we will formulate and prove our main result about the asymptotic joint normality of our estimators. \\ 
For the asymptotic theory to hold, we will always assume that $ \frac{n}{m} \rightarrow \nu \in \mathbb{R}^+$, for $n,m \rightarrow \infty$. 
Furthermore, we assume that the distribution functions $F$ and $G$ have continuous densities, and the slope of the ROC curve $GF^{-1}(t)$, which is defined as
$ %$\begin{align*}
 g(F^{-1}(t))/f(F^{-1}(t)),   
 $ % \end{align*}
is bounded on every subinterval $(a,b)$ of $(0,1)$, $a<b \in (0,1)$. 
This condition is equivalent to the assumption that the slope of the curve is finite for all $t \in (0,1)$, i.e.
\begin{align*}
  \forall t \in (0,1): \ \ \Biggl | \frac{g(F^{-1}(t))}{f(F^{-1}(t))}    \Biggl | < \infty \ .
\end{align*}
\cite{hsieh_turnbull} showed that under these conditions
\begin{align*}
  ||  \hat{G}_m(\hat{F}_n^{-1})- G(F^{-1})||_{\infty}= \sup_{t \in [0,1]} | \hat{G}_m(\hat{F}_n^{-1}(t))- G(F^{-1}(t)) |  \rightarrow 0 , \text{ a.s.}.
\end{align*}
This follows from the Glivenko-Cantelli theorem and the Dvoretzky, Kiefer, and Wolfowitz inequality (Theorem 7.12 in \cite{jiang2022large}). 
Therefore, $\hat{G}_m(\hat{F}_n^{-1}(t))$ is a strongly consistent estimator of $G(F^{-1}(t))$. Using this result, we can easily obtain a new method to prove the consistency of our plug-in estimators introduced in the previous section.
\begin{theorem} \label{thm_asym}
Let $\theta$, $\hat{\theta}_{mn}$, $I_2$, and $\hat{I}_{2,mn}$ be defined as before. Under the above conditions, 
\begin{align} \label{asymp_vec}
\sqrt{n}  \Biggl ( 
\begin{pmatrix}
\hat{\theta}_{mn} \\ \hat{I}_{2,mn}
\end{pmatrix}
-
\begin{pmatrix}
\theta \\ I_2
\end{pmatrix} 
\Biggl )
\end{align}
converges in distribution to a bivariate normal distribution $N(0,\Sigma)$ with expectation vector zero and covariance matrix 
\begin{align*}
 \Sigma= \begin{pmatrix}
\sigma^2_{\theta} & \sigma_{\theta,I_2} \\
\sigma_{\theta,I_2} & \sigma^2_{I_2}
\end{pmatrix},
\end{align*}
where
\begin{align*}
   \sigma^2_{\theta} &= \nu \text{ var} \Biggl[ \int_0^1 \ B_1(G(F^{-1}(t))) dt \Biggl ] + \text{ var} \Biggl[ \int_0^1 \ B_2(F(G^{-1}(t))) dt \Biggl ]       \\ 
   \sigma^2_{I_2}  &= \nu \text{ var} \Biggl[ \int_0^1 \ B_1(G(F^{-1}(1-t/2)))-\ B_1(G(F^{-1}(t/2))) dt \Biggl ]  \\
   & \  \ \ + \text{ var} \Biggl[ \int_0^1 \ B_2(F(G^{-1}(1-t/2))) - B_2(F(G^{-1}(t/2))) dt \Biggl ]        \\ 
   \sigma_{\theta,I_2} &=   \nu \text{ cov}  \Biggl[ \int_0^1 \ B_1(G(F^{-1}(t))) dt, \int_0^1 \ B_1(G(F^{-1}(1-t/2)))-\ B_1(G(F^{-1}(t/2))) dt \Biggl ]  \\   
   & \  \ \ +  \text{cov} \Biggl[ \int_0^1 \ B_2(F(G^{-1}(t))) dt  ,  \int_0^1 \ B_2(F(G^{-1}(1-t/2))) - B_2(F(G^{-1}(t/2))) dt   \Biggl ]    ,
\end{align*}
where $B_1$ and $B_2$  are independent Brownian bridges. 
\\
This can be simplified to
\begin{align*}
  \sigma^2_{\theta} &= \nu ||G \circ F^{-1}||^*+ ||F \circ G^{-1}||^* \\
  \sigma^2_{I_2}  &= \nu \Bigl ( ||G \circ F_1^{-1}||^* +||G \circ F_2^{-1}||^* -2   \langle G \circ F_1^{-1}, G \circ F_2^{-1} \rangle ^*   \Bigl ) \\
  &+  ||F \circ G_1^{-1}||^* +||F \circ G_2^{-1}||^* -2   \langle F \circ G_1^{-1}, F \circ G_2^{-1} \rangle ^*  \\
  \sigma_{\theta,I_2} &= \nu \Bigl (  \langle G \circ F^{-1}, G \circ F_2^{-1} \rangle ^* - \langle G \circ F^{-1}, G \circ F_1^{-1} \rangle ^* \Bigl ) \\ 
  &+ \langle F \circ G^{-1}, F \circ G_2^{-1} \rangle ^* -\langle F \circ G^{-1}, F \circ G_1^{-1} \rangle ^*,
\end{align*}
where 
\begin{align*}
    ||h||^* &= \int_0^1 h^2(t) dt - \Bigl ( \int_0^1 h(t) dt \Bigl )^2 \\
    \langle h, \Tilde{h} \rangle ^* &= 
    %\int_0^1 \int_0^1 h(t) \Tilde{h}(s) dt ds
    \int_0^1 h(t) \Tilde{h}(t) dt 
    - \Bigl ( \int_0^1 h(t) dt \Bigl ) \Bigl ( \int_0^1 \Tilde{h}(s) ds \Bigl ).
\end{align*}
\end{theorem}

\begin{proof}
Appendix \ref{proofthm1}
\end{proof}

The approximation by Brownian Bridges is quite usual for empirical processes. The strong approximation result we use here is based on the Hungarian Construction (see, e.g., Chapter 12 in \cite{shorack1986empirical})
\begin{remark} \label{rem_var}
    The variances $\sigma^2_{\theta}$ and $\sigma^2_{I_2}$ in Theorem \ref{thm_asym} simplify in the case of distribution equality to $\frac{1}{12}(\frac{1}{m} +\frac{1}{n})$ -- the same as for the classical Mann-Whitney U Statistic in the null case $F=G$ (see, e.g.,
    \cite{serfling2009approximation} or
    \cite{hollander2013nonparametric}). 
    For $F=G$, it follows also that the covariance $\sigma_{\theta,I_2}$ is zero.  
    Thus, under the null, the variances are equal and the estimators are uncorrelated. 
\end{remark}
We will consider the previous remark again when constructing a test for distribution equality. 
However, in the next section we will construct confidence intervals using resampling techniques.

\section{Resampling based Inference for \texorpdfstring{$\theta$}{Lg} and \texorpdfstring{$I_2$}{Lg}} \label{secresamp}
In the previous section, we have seen that the challenge of estimating the covariance matrix $\Sigma$ in Theorem \ref{thm_asym} becomes a rather simple task under the null hypothesis $H_0: F=G$. However, this approach can only be used for hypothesis testing. For the construction of  simultaneous confidence intervals, we have to use different techniques. \\
One possibility would be to utilize kernel density estimators, an approach that can also be used for estimating the ROC-curve (\cite{zou97}, \cite{lloyd98}). However, it requires larger sample sizes, and as we are aiming for an approach that is especially applicable for small and moderate sample sizes, we are using a resampling based method for this problem. \\
To this end, we will first introduce the empirical bootstrap process and then prove that the bootstrap is valid for the considered functionals. After that, we will present a method for constructing simultaneous confidence intervals based on the bootstrap. \\
%\subsection{Empirical Bootstrap Process}
Consider the so-called bootstrap samples 
$X_1^*, \ldots, X_n^*$, i.i.d. sample from $F_n$, and $Y_1^*, \ldots, Y_m^*$, i.i.d. sample from $G_m$. 
We call the empirical measure $P_n^* = \frac{1}{^n}\sum_{i=1}^n \delta_{X_i^*}$ the \emph{bootstrap empirical distribution} and $\hat{F}_n^*$ the corresponding empirical distribution function, and analogously defined $Q_m^* = \frac{1}{^n}\sum_{i=1}^n \delta_{X_i^*}$ and $\hat{G}_m^*$, respectively. 
Our approach will be justified by considering the \emph{empirical bootstrap process}, defined as $ \sqrt{n} ( \hat{G}_m^{*} \circ \hat{F}_n^{* -1} -\hat{G}_m \circ \hat{F}_n^{ -1} ) $. 
\begin{theorem} \label{thm_boot}
Let $\hat{G}_m^{*}$ and $\hat{F}_n^*$ be the empirical distribution function of the bootstrapped sample. Then it holds under the same conditions as in Theorem \ref{thm_asym} that the empirical bootstrap process of $\theta$ and $I_2$:
\begin{align} \label{bootstrap_process}
\sqrt{n}  \Biggl ( 
\begin{pmatrix}
\hat{\theta}^*_{mn} \\ \hat{I}^*_{2,mn}
\end{pmatrix}
-
\begin{pmatrix}
\hat{\theta}_{mn} \\ \hat{I}_{2,mn}
\end{pmatrix}
\Biggl ),
\end{align}
where $ \hat{\theta}^*_{mn}$ and $\hat{I}^*_{2,mn}$ are defined as
 \begin{align*}
\hat{\theta}^*_{mn} &=\int \hat{G}^*_m d \hat{F}^*_n = \int_0^1 \hat{G}^*_m \circ \hat{F}_n^{* -1} (\alpha) d \alpha, \\
\hat{I}^*_{2,mn} &= \int_0^1 \hat{G}^*_m(\hat{F}_n^{* -1}(1- \alpha/2))d \alpha - \int_0^1 \hat{G}^*_m(\hat{F}_n^{* -1}(\alpha/2)) d \alpha \\
&= 2 (\int_{\hat{F}_n^{* -1}(1/2)}^{\infty} \hat{G}^*_m(t) d\hat{F}^*_n(t)- \int^{\hat{F}_n^{* -1}(1/2)}_{- \infty} \hat{G}^*_m(t) d \hat{F}^*_n(t) ),
\end{align*} 
is asymptotically normally distributed with identical limit distribution as 
\begin{align*} 
\sqrt{n}  \Biggl ( 
\begin{pmatrix}
\hat{\theta}_{mn} \\ \hat{I}_{2,mn}
\end{pmatrix}
-
\begin{pmatrix}
\theta \\ I_2
\end{pmatrix} 
\Biggl )
\end{align*}
\end{theorem}
\begin{proof}
Appendix \ref{proodthm2}
\end{proof}

For more details on the empirical bootstrap process, see for example Chapter 3.6 in \cite{van1996weak}.
%\subsection{Simultaneous Confidence Intervals}
\label{sec:sim_conf_int}
Theorem \ref{thm_boot} justifies our approach to use the bootstrap for the construction of simultaneous confidence intervals for our estimators. Here we present two methods for constructing a simultaneous confidence intervals. 
\begin{enumerate}
    \item  {\bf Asymptotically multivariate-normal with bootstrapped covariance } \\
    In Theorem \ref{thm_asym}, we have shown that both estimators are asymptotically bivariate normally distributed. Hence, we can derive a symmetric, elliptical $1-\alpha$ confidence region from the bootstrap covariance by considering all bivariate vectors $ {\bf d}=(d_1,d_2)$ satisfying
    \begin{align*}
   \frac{1}{\sqrt{n}} ||(\hat{\theta}_{mn}, \hat{I}_{2,mn})- {\bf d} ||_{2}  \leq \Phi_{\Sigma^*_{n,m}}(1- \alpha),
    \end{align*}
    where $\Phi_{\Sigma^*_{n,m}}$ is the cdf of binariate normal distribution with expectation $0$ and covariance matrix $\Sigma^*_{n,m}$.
    The marginals of this region are
\begin{align*}
\Biggl [
    \begin{pmatrix}
\hat{\theta}_{mn} \\ \hat{I}_{2,mn}
\end{pmatrix}
- \frac{1}{\sqrt{n}} z_{1-\alpha/2}(\Sigma^*_{n,m}),
\begin{pmatrix}
\hat{\theta}_{mn} \\ \hat{I}_{2,mn}
\end{pmatrix}
+\frac{1}{\sqrt{n}} z_{1-\alpha/2}(\Sigma^*_{n,m}) \Biggl ],
\end{align*}
where $z(\Sigma^*_{n,m})$ denotes the quantile of a bivariate normal distribution with expectation $0$ and covariance matrix $\Sigma^*_{n,m}$, which is the sample covariance matrix of the bootstrap estimates. 
  \item {\bf Empirical bootstrap quantiles (Bonferroni adjusted) } \\
    Of course it is not necessary to incorporate the asymptotic normality of our estimators. 
    For smaller sample sizes, it may even be more accurate to construct  $1-\alpha$ simultaneous confidence intervals solely based on the empirical quantiles of the bootstrap distribution:
\begin{align*}
\Biggl [
    \begin{pmatrix}
\hat{\theta}_{mn} \\ \hat{I}_{2,mn}
\end{pmatrix}
- \frac{1}{\sqrt{n}}
 \begin{pmatrix}
z^*_{\hat{\theta}, \alpha/4} \\ z^*_{\hat{I_2}, \alpha/4}
\end{pmatrix},
\begin{pmatrix}
\hat{\theta}_{mn} \\ \hat{I}_{2,mn}
\end{pmatrix}
+\frac{1}{\sqrt{n}} 
\begin{pmatrix}
z^*_{\hat{\theta}, 1-\alpha/4} \\ z^*_{\hat{I_2}, 1-\alpha/4}
\end{pmatrix} \Biggl ],
\end{align*}
where $z^*_{\hat{\theta}}$  and $z^*_{\hat{I_2}}$ denote the quantiles of the empirical bootstrap distribution of $\hat{\theta}^*$ and $\hat{I_2}^*$, respectively. 
The intervals are Bonferroni adjusted due to multiplicity.
\end{enumerate}

In Section \ref{alt_ci}
in the Appendix, we present a short review on alternative approaches for construction simultaneous confidence intervals and compare them in a simulation study. The first method based on the quantiles of a multivariate normal distribution yields elliptical confidence regions. The other methods (including those introduced in the Appendix) are all based on some adjustment of univariate confidence intervals and have a rectangular structure. The performance of these methods will be briefly evaluated in Section \ref{sec_sim_ci} and in more detail in Section \ref{add_sim_ci} 
in the Appendix.
\begin{remark} {\bf Range Preserving Confidence Intervals} \\
In the case of very large or small parameter values, it is possible that the constructed confidence intervals take values outside the possible parameter ranges.
In such cases, we recommend to use simultaneous range preserving confidence intervals (\cite{efron1994introduction}, \cite{konietschke12}). 
The above results can be amended in order to derive a confidence region for $\hat{\theta}$ and $\hat{I}_{2} $ which lies fully within the triangle in Figure \ref{fig:theta_I_2_relationship} representing all possible bivariate parameter values.
\end{remark}

\section{Testing for Equality of Distributions} \label{chap_test}
%\cite{parkinson2022testing} developed a test for equality of distributions using the previously introduced concept of niche overlap. 
As mentioned in the introduction, we would like to use both nonparametric functionals to develop a test on distribution equality. 
Therefore, we use the asymptotic joint distribution of the estimators of $\theta$ and $I_2$, derived in Section \ref{sec_joint}. 
The main challenge is a practically useful estimation of the covariance matrix of our estimators.
The bootstrapped covariance matrix estimators from Section \ref{secresamp} may be too conservative. 
We will tackle this challenge with two different approaches: 
\subsection{Test for Distribution Equality for Absolutely Continuous Distributions} \label{test_abscont}
The first approach is based on the covariance matrix calculated in Theorem \ref{thm_asym}. 
As mentioned before in Remark \ref{rem_var}, this matrix simplifies in the case of equal distributions to:
\begin{align*}
   \Sigma= \begin{pmatrix}
\sigma^2 & 0 \\
0 & \sigma^2
\end{pmatrix},
\end{align*}
where $\sigma^2$ is defined as $ \frac{1}{12}(\frac{1}{m} +\frac{1}{n})$, as in the classical Mann-Whitney U-Test. Therefore it is quite commonly used for estimating the variance of $\theta$ in the case of absolutely continuous distributions. We then use the test statistic
\begin{align*}
 \begin{pmatrix}
 \ (\hat{\theta}_{mn}-1/2)/\sigma \\  \ (\hat{I}_{2,mn}-1/2)/ \sigma
\end{pmatrix}
 \end{align*}
 which follows under the null hypothesis of distribution equality, $H_0: F=G$, asymptotically a bivariate standard normal distribution. 
 \begin{remark} \label{consis_test1} Consistency Region: \\
As mentioned even in the title of this paper, our proposed test is not an omnibus test, such as for example the Kolmogorov-Smirnov or Cramer von Mises tests. 
Instead, we test the two-dimensional null hypothesis $\theta =I_2=1/2$. 
Due to equations \eqref{theta} and \eqref{I2}, this is equivalent to
\begin{align} \label{cons_test}
P(X^{(1)} < Y^{(1)})+ P(X^{(1)} < Y^{(2)}) =3/2  \ &\wedge \
P(X^{(2)} < Y^{(2)}) +P(X^{(2)} < Y^{(1)}) = 1/2. 
%\label{cons_test1}
\end{align}
Therefore, the new test proposed here will be consistent for alternatives to \eqref{cons_test}. 
%or \eqref{cons_test1}.
 \end{remark}
 The performance of this procedure is quite good (for details see Chapter \ref{simulation}). 
 However, our approach for the proof of Theorem \ref{thm_asym} relied on results assuming continuity of the distributions. 
 Therefore, we have additionally developed an approach based on the theory of \cite{brunner2017rank} and used this approach in the construction of a univariate test based on the overlap index by \cite{parkinson2022testing}. The main idea is that under the assumption of equal medians, both functionals $\theta$ and $I_2$ simplify. 
 Then, we use these "simplified" estimators for testing the null hypothesis of equal distributions. As they consist of independent components, we can incorporate the theory of multiple contrast test developed by \cite{konietschke12}.
Using these results, we are able to derive the joint asymptotic distribution of $\theta$ and $I_2$ 
with fewer assumptions than in Section \ref{sec_joint}, especially without assuming absolutely continuous distributions, 
and additionally we obtain a consistent covariance estimator. 
We present this approach  in the following section in
more detail. 
%However, the reduction in the assumptions also comes with a disadvantage since this method is quite liberal for very small sample sizes and performs poorly for pure location-shift effects, likely due to the induced bias for the Mann-Whitney functional (see the simulation results in Section \ref{add_sim_test}).

\subsection{Test for Distribution Equality Using Adjusted Statistical Functionals} \label{test-sim}
Using the definition of the conditional versions of the random variables $X$ and $Y$ in Remark \ref{alt_i2}, it follows directly from the representation of $\theta$ and $I_2$ in Equations \eqref{theta} and \eqref{I2} that under the assumption of equal medians, $P(X^{(1)} < Y^{(2)})= P(Y^{(1)} < X^{(2)})=1$ and $P(X^{(2)} < Y^{(1)})=P(Y^{(2)} < X^{(1)})=0$. 
Therefore,
\begin{equation} \label{theta_medeq}
\theta= \frac{1}{4} [ 1+ P(X^{(1)} < Y^{(1)})+P(X^{(2)} < Y^{(2)})]
\end{equation}
and
\begin{equation} \label{I2_medeq}
I_2= \frac{1}{2} [ 1+P(X^{(2)} < Y^{(2)})-P(X^{(1)} < Y^{(1)})].
\end{equation}
Under the null hypothesis $H_0 \: \ F=G$, in the case of continuous distributions $F$ and $G$, 
simplified estimators of the probabilities 
$P(X^{(1)} < Y^{(1)})$ and $P(X^{(2)} < Y^{(2)})$
can be devised when taking advantage of the fact that under this null hypothesis, both $X$ and $Y$ have the same median.
To this end, we split both the $X$- and the $Y$-samples at their joint median and define the corresponding canonical estimators to be 
$\hat{\theta}^{(1)}_{mn}$ and $\hat{\theta}^{(2)}_{mn}$,
respectively. 
In our simulations, splitting at the joint median proved superior to splitting at the individual medians which had been proposed by 
\cite{parkinson2022testing}. 

In order to keep the notation simple, we call the  samples split at the joint median again $X^{(1)}$, $X^{(2)}$, $Y^{(1)}$ and $Y^{(2)}$.  
Under the null hypothesis implying median equality, and due to the split at the same value, the estimators $\hat{\theta}^{(1)}_{mn}$ and $\hat{\theta}^{(2)}_{mn}$ are independent. Thus, and using equations \eqref{theta_medeq} and \eqref{I2_medeq}, 
we can derive a consistent estimator for the covariance matrix of our estimators $\hat{\theta}_{mn}$ and $\hat{I}_{2,mn}$. 

To this end, we first need an extension to the rank definition introduced in Section \ref{sec2}. 
We assume again without loss of generality that the samples are ordered $X_1 < \ldots, < X_n$ and $Y_1 < \ldots  < Y_m$. 
The $X$-observations below their joint median we denote by $X_1, \ldots, X_K$ and the observations above by $X_{K+1}, \ldots, X_n$. We do the same for the $Y$ observations $Y_1, \ldots, X_L$ and  $Y_{L+1}, \ldots, Y_m$. If there are ties at the joint median we split the values equally\\
The ranks of $X$ and $Y$ below their median we denote by $R_i^{X(<)}$ and $R_i^{Y(<)}$ and their average by $\bar{R}_.^{X(<)}$ and $\bar{R}_.^{Y(<)}$. Analogously we define $R_i^{X(>)}$ and $R_i^{Y(>)}$, $\bar{R}_.^{X(>)}$ and $\bar{R}_.^{Y(>)}$, respectively, for those above the median. 
The ranks within each of the four groups are denoted  $R_i^{(X<)}$, $R_i^{(Y<)}$, $R_i^{(X>)}$, and $R_i^{(Y>)}$. \\
As the relative effect and the niche overlap are here only used for the purpose of hypothesis testing, we can define the following adjusted versions of both effect measures.
The \emph{adjusted relative effect} is defined as
\begin{equation} \label{theta_sim}
\theta^{\text{adj}}= \frac{1}{4} [ 1+ P(X^{(1)} < Y^{(1)})+\frac{1}{2} P(X^{(1)} = Y^{(1)}) +P(X^{(2)} < Y^{(2)}) + \frac{1}{2} P(X^{(2)} = Y^{(2)})]
\end{equation}
and the \emph{adjusted overlap index} is defined as
\begin{equation} \label{I2_sim}
I_{2} ^{\text{adj}}= \frac{1}{2} [ 1+P(X^{(2)} < Y^{(2)})+ \frac{1}{2} P(X^{(2)} = Y^{(2)})-P(X^{(1)} < Y^{(1)})- \frac{1}{2} P(X^{(1)} = Y^{(1)})].
\end{equation}
Employing the canonical plug-in estimators again, we obtain  consistent estimators for both functionals. 
Noting that we have linear combinations of relative effects,  
the representation in equation \eqref{repr_theta} can be used, only substituting $X$ and $Y$ by their conditional counterparts. Keep in mind that equations \eqref{theta_sim} and \eqref{I2_sim} are equivalent to equations \eqref{theta_medeq} and \eqref{I2_medeq} in the case of continuous distributions and equal medians.
\begin{proposition} (Variance Estimation) \label{var_est} \\
\begin{comment}
Under the assumption of continuous distribution function $F$ and $G$, with equal medians $F^{-1}(1/2)=G^{-1}(1/2)$ a consistent estimator of the covariance matrix $\hat{\Sigma}_{mn}$ of $(\hat{\theta}_{mn},\hat{I}_{2,mn})$ is given by
\end{comment}
Assume that the variances $\sigma_1^2= \text{Var} (G(X_1))$ and $\sigma_2^2= \text{Var} (F(Y_1))$ satisfy $\sigma_1^2 > 0$ and $\sigma_2^2 > 0$ and that $\frac{n+m}{n} \leq N_1 < \infty $ and $\frac{n+m}{m} \leq N_2 < \infty $. 
Then, a consistent estimator of the covariance matrix $\hat{\Sigma}_{mn}$ of $(\hat{\theta}_{mn}^{\text{adj}},\hat{I}_{2,mn}^{\text{adj}})$ is given by
\begin{align*}
\hat{\Sigma}_{mn}=
\begin{pmatrix}
s^2/4 & s_{\theta,I_2} \\
s_{\theta,I_2} & s^2
\end{pmatrix},
\end{align*}
where
\begin{align*}
s^2 &= \frac{l+k}{2} \Bigl ( \frac{s_{X_1}^2}{k}+ \frac{s_{Y_1}^2}{l} \Bigl) +\frac{n+m-l-k}{2} \Bigl ( \frac{s_{X_2}^2}{n-k}+ \frac{s_{Y_2}^2}{m-l} \Bigl), \\
s_{\theta,I_2} &= \frac{l+k}{2} \Bigl ( \frac{s_{X_1}^2}{k}+ \frac{s_{Y_1}^2}{l} \Bigl) -\frac{n+m-l-k}{2} \Bigl ( \frac{s_{X_2}^2}{n-k}+ \frac{s_{Y_2}^2}{m-l} \Bigl),
\end{align*}
and
\begin{align*}
s_{X_1}^2 &=  \frac{1}{l^2(k-1)} \sum_{i=1}^k \Bigl( R_i^{X(<)}- R_i^{(X<)}- \bar{R}_.^{X(<)} +\frac{k+1}{2} \Bigl)^2, \\
s_{X_2}^2  &=  \frac{1}{(m-l)^2(n-k-1)} \sum_{j=k+1}^n \Bigl( R_j^{X(>)}- R_i^{(X>)}- \bar{R}_.^{X(>)} +\frac{n-k+1}{2} \Bigl)^2,   \\
s_{Y_1}^2   &=  \frac{1}{k^2(l-1)} \sum_{i=1}^l \Bigl( R_i^{Y(<)}- R_i^{(Y<)}- \bar{R}_.^{Y(<)} +\frac{l+1}{2} \Bigl)^2,   \\
s_{Y_2}^2 =    &=  \frac{1}{(n-k)^2(m-l-1)} \sum_{j=l+1}^m \Bigl( R_j^{Y(>)}- R_i^{(Y>)}- \bar{R}_.^{Y(>)} +\frac{m-l+1}{2} \Bigl)^2.
\end{align*}
Additionally, under the null hypothesis $H_0: F=G$, the test-statistic
 \begin{align} \label{statistic_sim}
 \begin{pmatrix}
\sqrt{m+n} \ (\hat{\theta}_{mn}^{\text{adj}}-1/2) \\ \sqrt{m+n} \ (\hat{I}_{2,mn}^{\text{adj}}-1/2)
\end{pmatrix}
 \end{align}
 is asymptotically distributed according to a multivariate normal distribution with expectation vector zero and covariance matrix $\hat{\Sigma}_{mn}$.
\end{proposition}
\begin{proof}
Appendix \ref{proofprop}
\end{proof}
Using the result from Proposition \ref{var_est} regarding the asymptotic distribution of the test statistic \eqref{statistic_sim}, we can test the null hypothesis $H_0:F=G$ versus $H_1:F \neq G$ without having to assume  absolutely continuous distributions. 

In the simulation study, we have empirically compared the different approaches for testing for equality of distributions.

We conclude this section with some short remarks on the consistency region of the test proposed here.
\begin{remark} Consistency Region: \\
Analogously to Remark \ref{consis_test1}, 
also the test proposed in this section and using the results from Proposition \ref{var_est}, 
is not consistent for every deviation from the null hypothesis of distribution equality. 
The procedure implicitly assumes the following:
\begin{align*}
    P(X^{(1)} < Y^{(2)}) &=1 \\
 P(X^{(2)} < Y^{(1)})&=0.
\end{align*}
Essentially, the test presented in this section evaluates
\begin{align*}
P(X^{(1)} < Y^{(1)}) + \frac{1}{2} P(X^{(1)} = Y^{(1)}) &=1/2 \\
P(X^{(2)} < Y^{(2)}) + \frac{1}{2} P(X^{(2)} = Y^{(2)}) &= 1/2.
\end{align*}
Therefore, it tests whether there is a stochastic tendency to larger or smaller values of $X$ with respect to $Y$, stratified by considering values below the joint median and above the joint median separately. 
\end{remark}

\section{Simulation Study} \label{simulation}
In this section, we present results from a simulation study evaluating the finite sample behavior of the methods introduced above. First, we focus on the performance of the hypothesis testing approach described in Section \ref{test_abscont} and in Section 
\ref{test-sim}, also comparing it with the most popular nonparametric procedures regarding empirical type I error and power.
%While doing this we will look if the test holds the proposed $\alpha$-level by evaluating the Type-I error and compare then, the power of our new test procedure with state of the art methods. 
Subsequently, we consider the different proposed approaches for constructing simultaneous confidence intervals and evaluate them with regard to empirical coverage probability and average interval length.
The quantiles of the multivariate normal distribution were calculated with the \textsc{R}-package \textsc{mvtnorm} \citep{mvtnorm}, \citep{genz_bretz}. 

\subsection{Hypothesis Test} \label{add_sim_test}
In our first simulation study, we have compared our newly proposed tests with two popular omnibus tests, namely the Kolmogorov-Smirnov and the Cramer-von Mises test. Additionally, the closely related and highly popular two-sample rank sum test (Wilcoxon, Mann, and Whitney) is included in the simulation comparison. 
We have simulated $1,000$ times from different distributions, for sample sizes from $10$ to $100$, in order to investigate the small (to moderate) sample performance of our proposed test procedures.
For simplicity we call the method for hypothesis testing based on the empirical ROC process described in Section \ref{chap_test} \emph{new joint test}, and we denote the method based on the adjusted statistical functionals described in Section \ref{test-sim} as \emph{adjusted joint test}.

\begin{figure}  
\begin{subfigure}{0.4\textwidth}
\includegraphics[width=\linewidth]{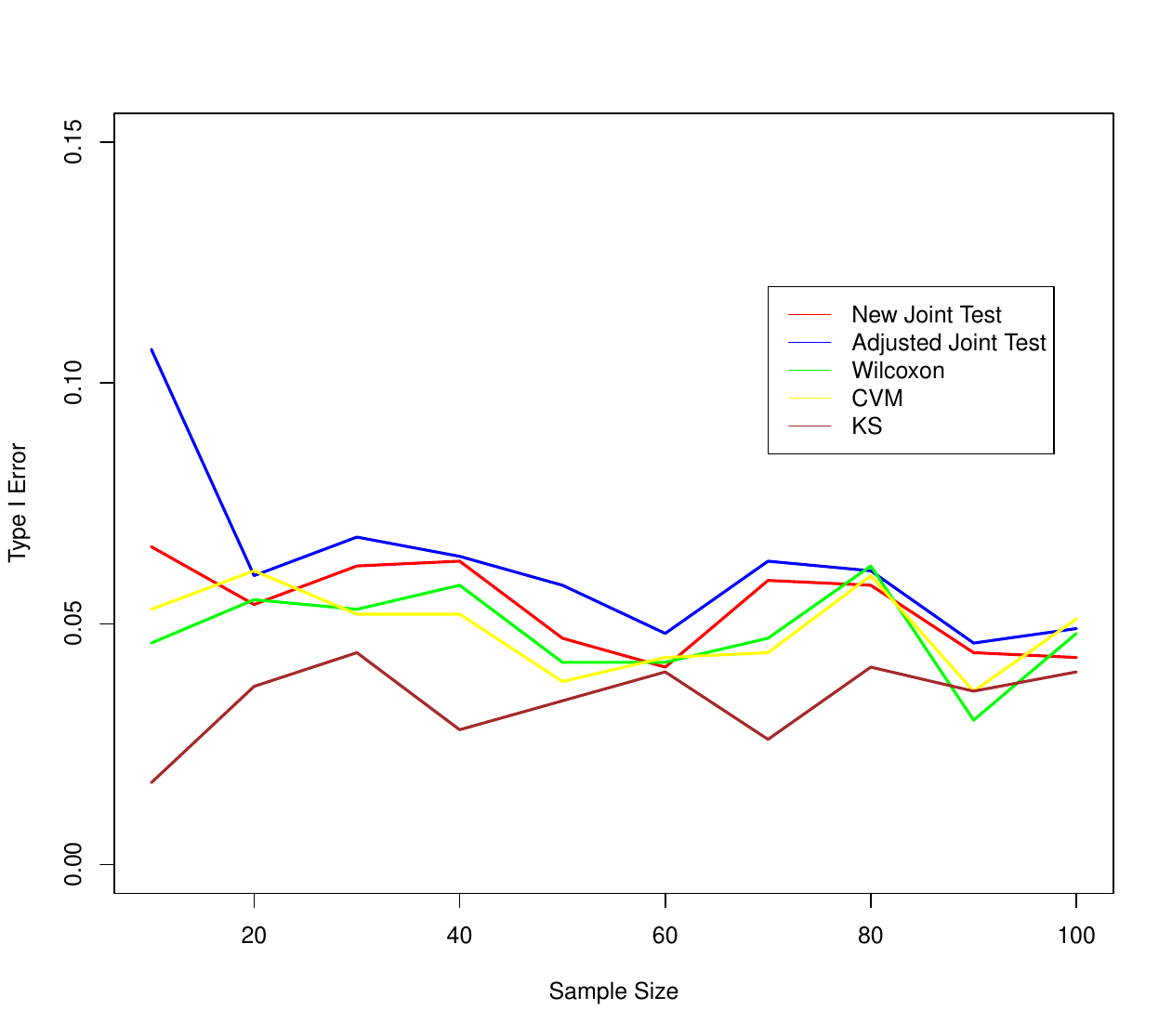}
\caption{Type I Error for $N(0, 1)$ } 
\end{subfigure}\hspace*{\fill}
\begin{subfigure}{0.4\textwidth}
\includegraphics[width=\linewidth]{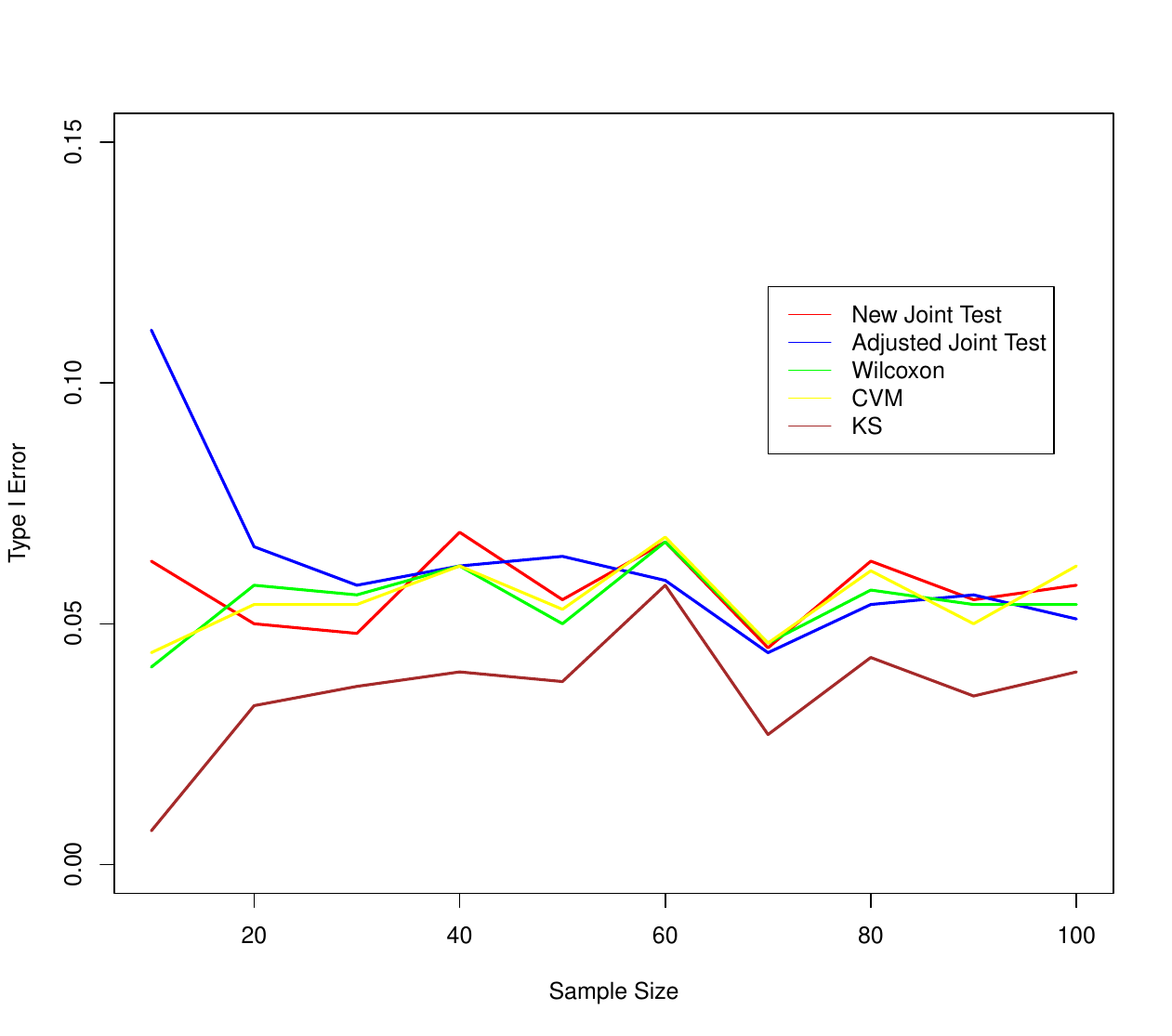}
\caption{Type I Error for Exp$(1)$ } 
\end{subfigure}

\medskip
\begin{subfigure}{0.4\textwidth}
\includegraphics[width=\linewidth]{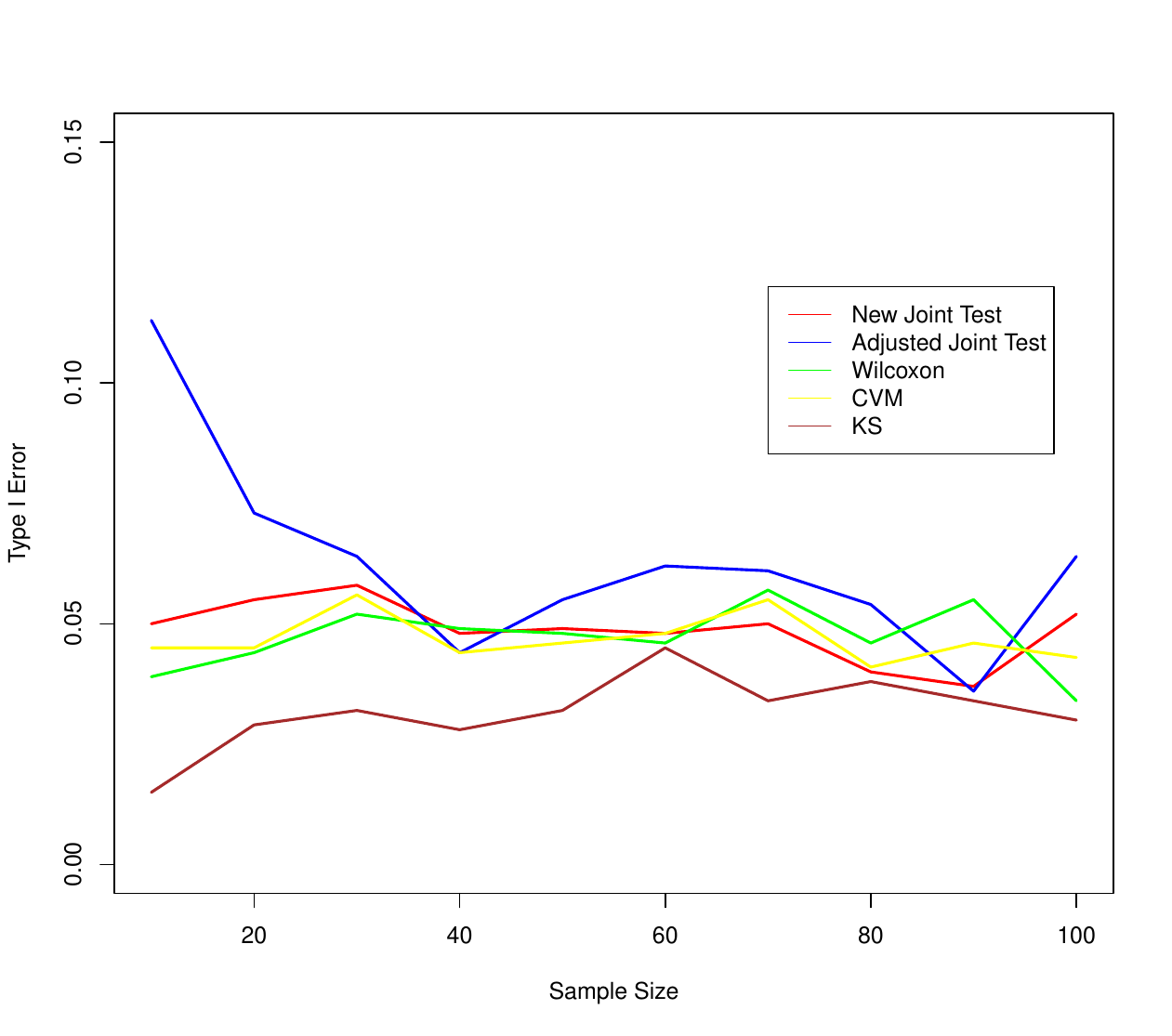}
\caption{Type I Error for  $U[0,1]$} 
\end{subfigure}\hspace*{\fill}
\begin{subfigure}{0.4\textwidth}
\includegraphics[width=\linewidth]{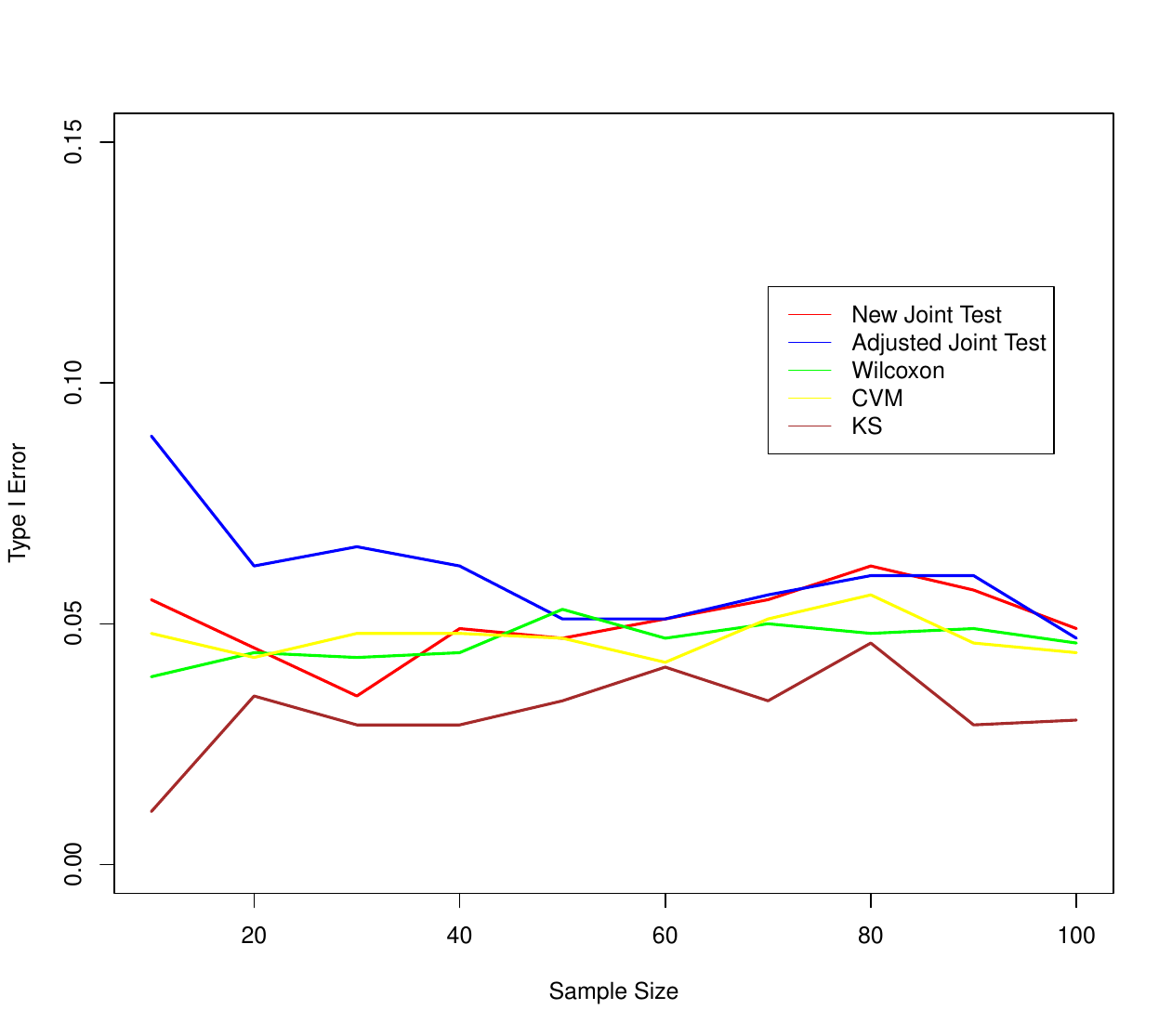}
\caption{Type I Error for Beta$(2, 3)$ } 
\end{subfigure}

\medskip
\begin{subfigure}{0.4\textwidth}
\includegraphics[width=\linewidth]{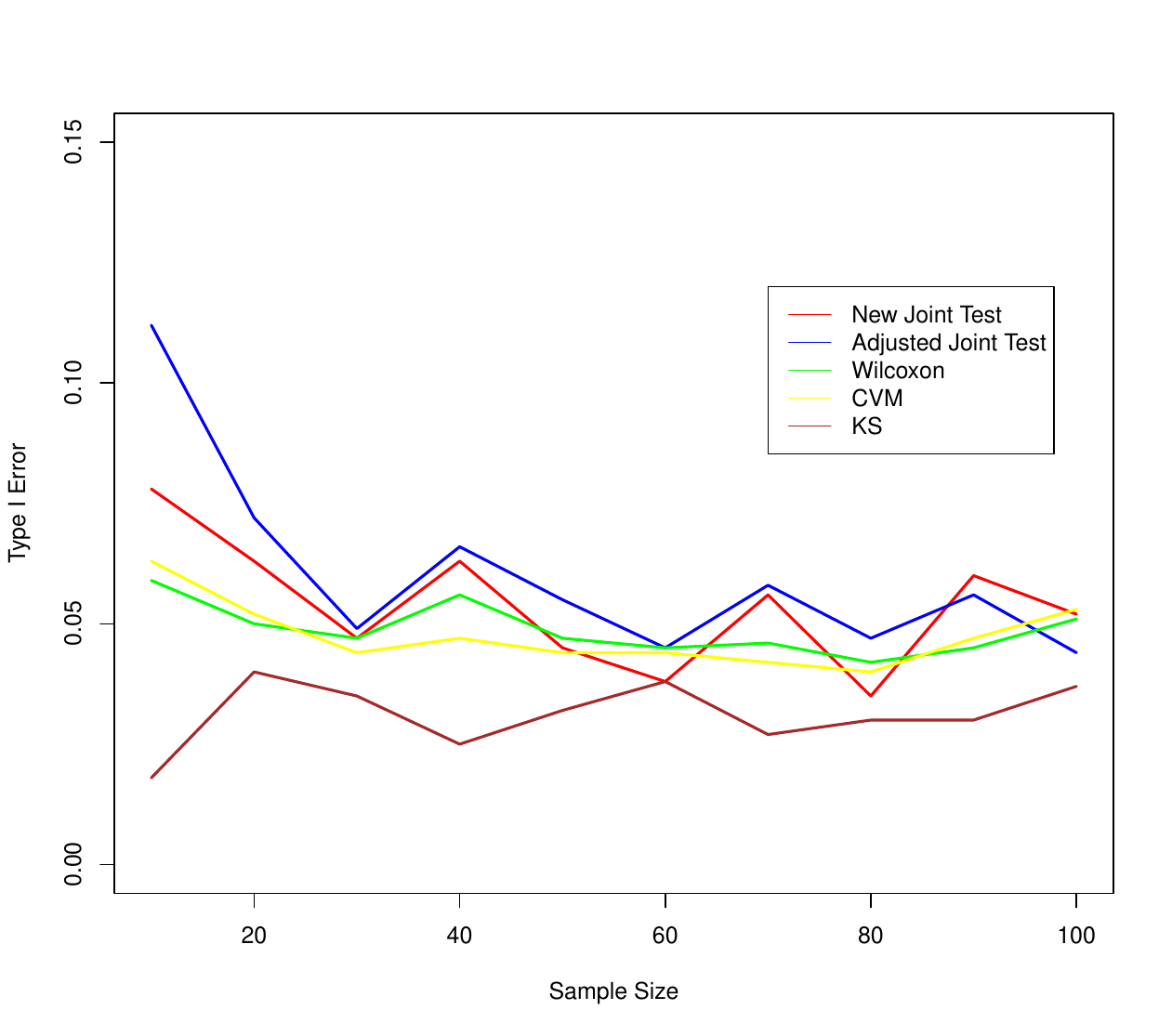}
\caption{Type I Error for Cauchy$(0, 1)$ } 
\end{subfigure}\hspace*{\fill}
\begin{subfigure}{0.4\textwidth}
\includegraphics[width=\linewidth]{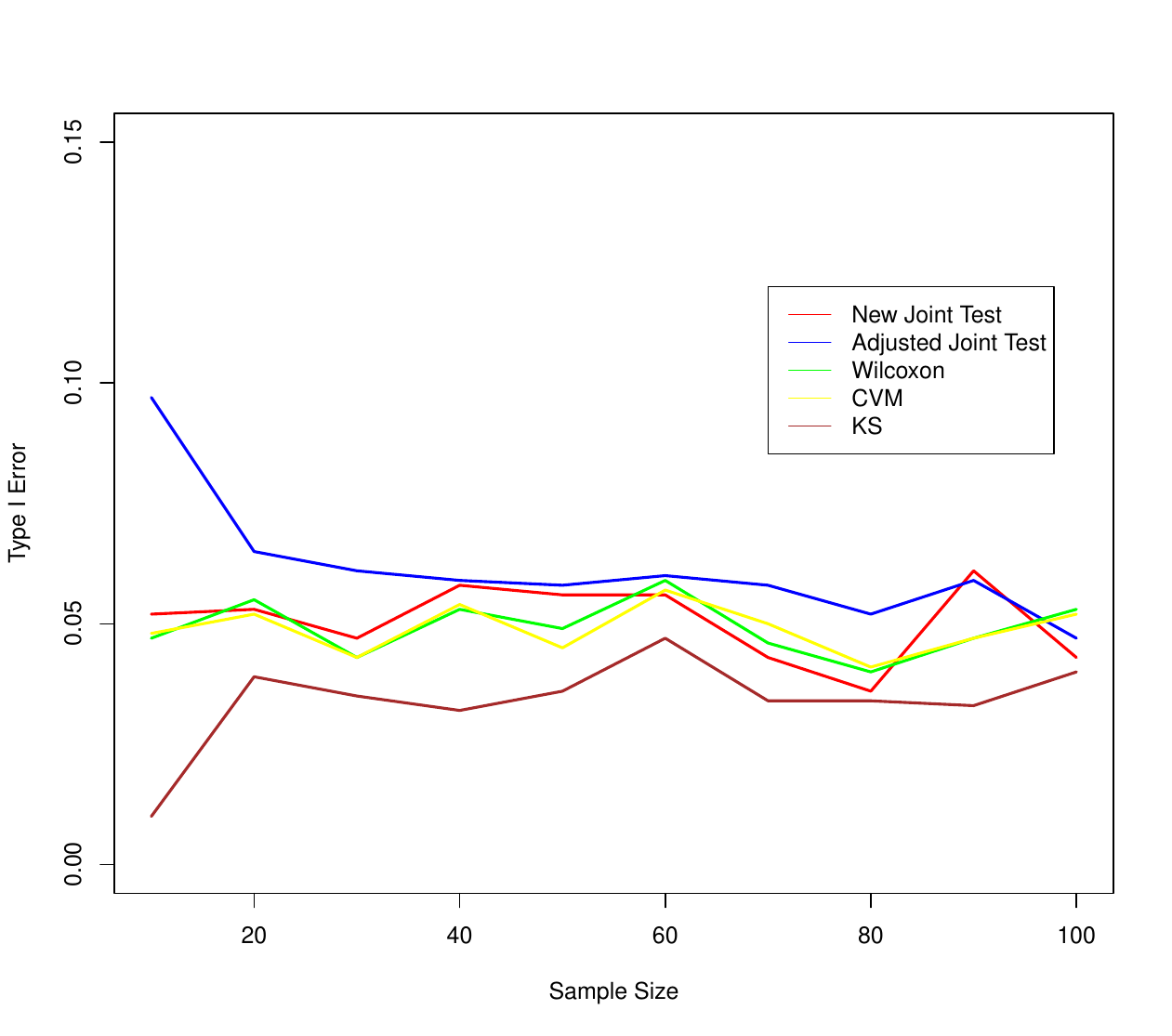}
\caption{Type I Error for $\chi_1^2$ } 
\end{subfigure}
\caption{Empirical type I error rates for different nonparametric tests, with different underlying distributions} \label{type1}
\end{figure}
For the simulation of the type I error in Figure \ref{type1}, we have assumed the null hypothesis and therefore drawn both samples from the same distribution. The underlying distributions include symmetric and skewed, as well as heavy- and light-tailed distributions. 
The method based on the adjusted estimators was too liberal for very small sample sizes ($n=10$) and should not be recommended for this situation, while our new joint test performed well, with similar type I error rates as the classical procedures. There were no major differences regarding the simulated test performance across the underlying distributions. Due to the nonparametric nature of the procedures considered, this is not unexpected. The Kolmogorov-Smirnov test stood out as being the most conservative, in particular for small samples.

\begin{figure} [ht!]
\begin{subfigure}{0.4\textwidth}
\includegraphics[width=\linewidth]{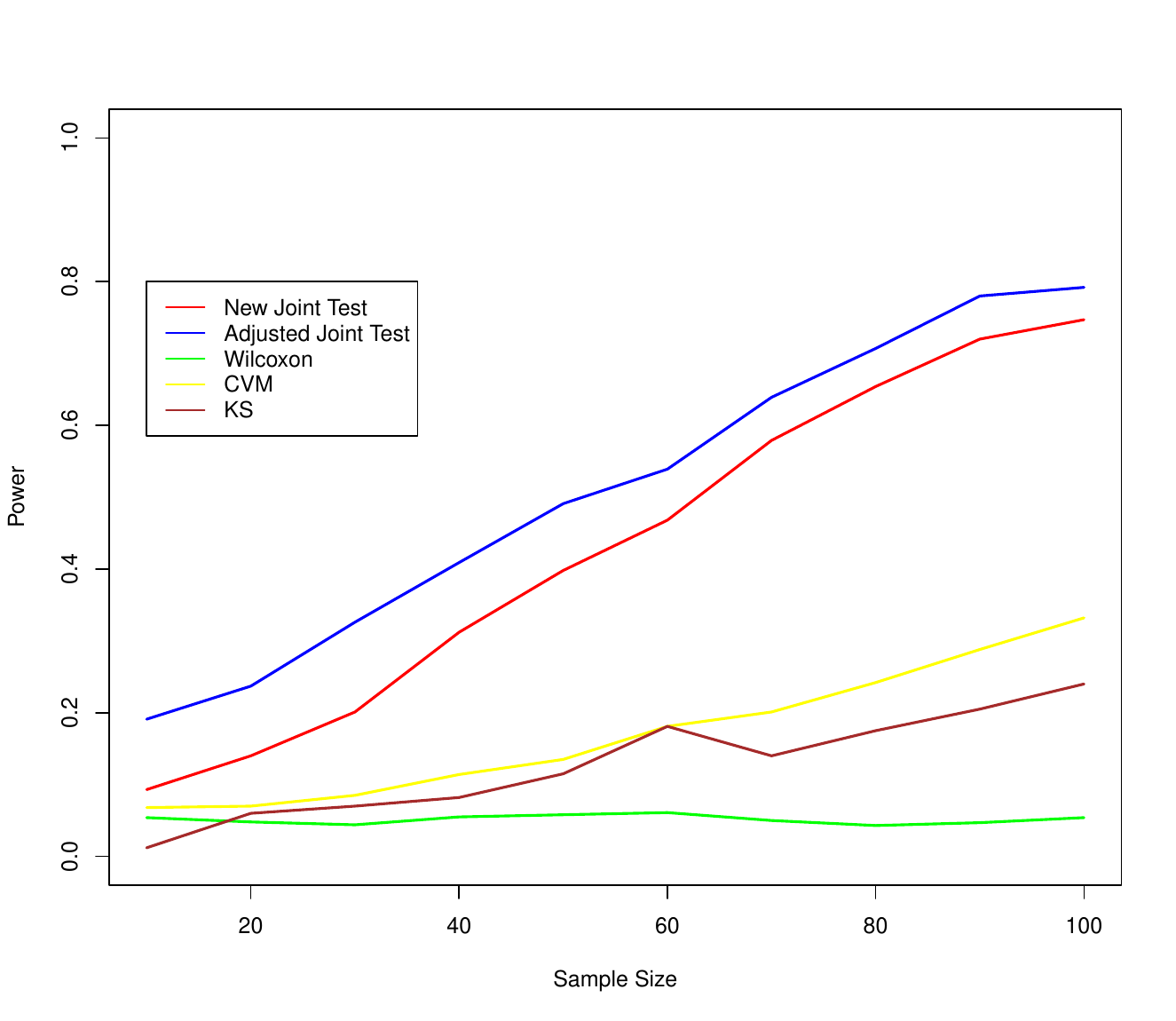}
\caption{Power for $N(0, 1)$ vs. $N(0, 1.5)$   }  \label{pow:1}
\end{subfigure}\hspace*{\fill}
\begin{subfigure}{0.4\textwidth}
\includegraphics[width=\linewidth]{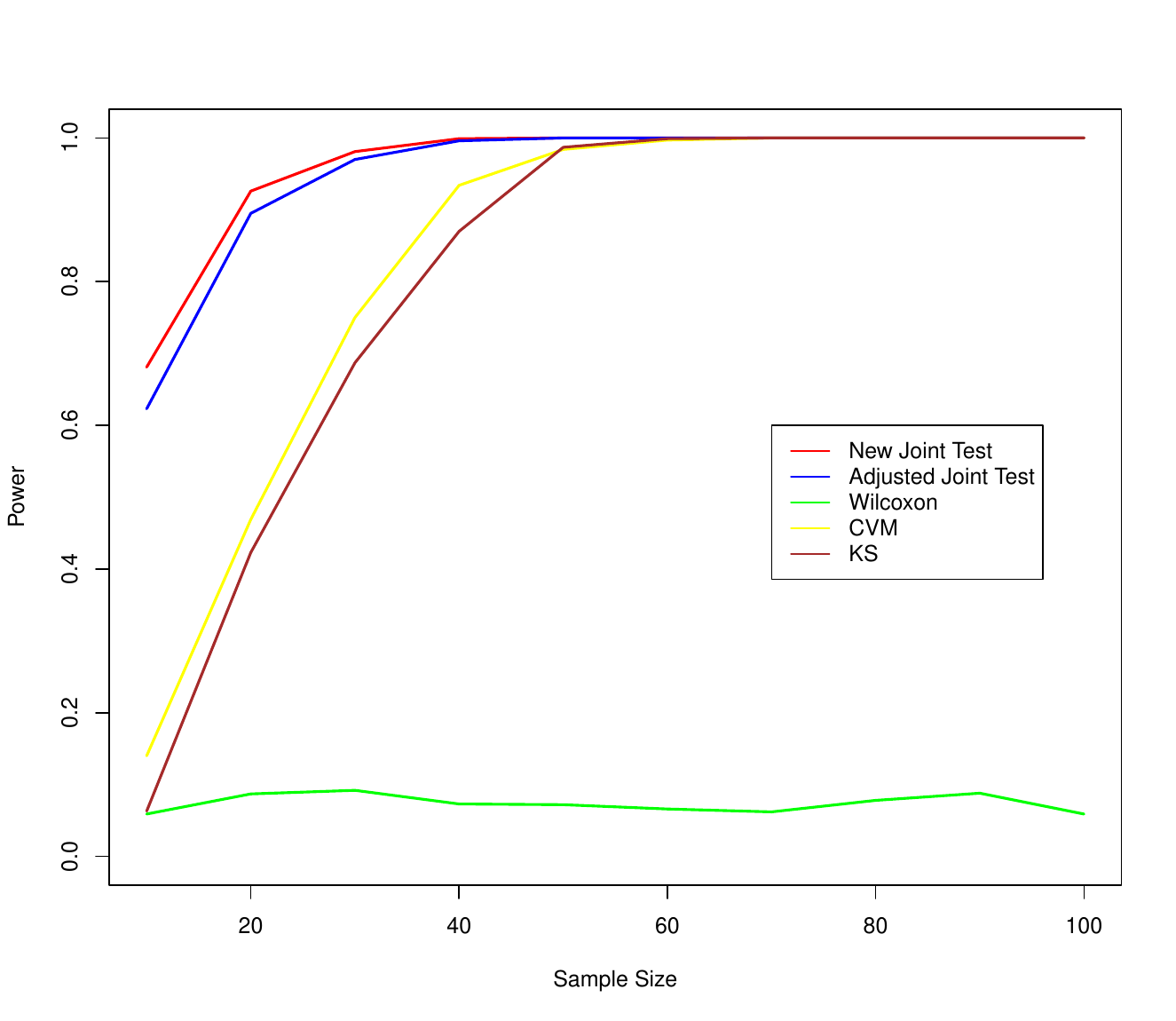}
\caption{Power for $N(0.5, 1)$ vs. $U[0,1]$ }  \label{pow:2}
\end{subfigure}

\medskip
\begin{subfigure}{0.4\textwidth}
\includegraphics[width=\linewidth]{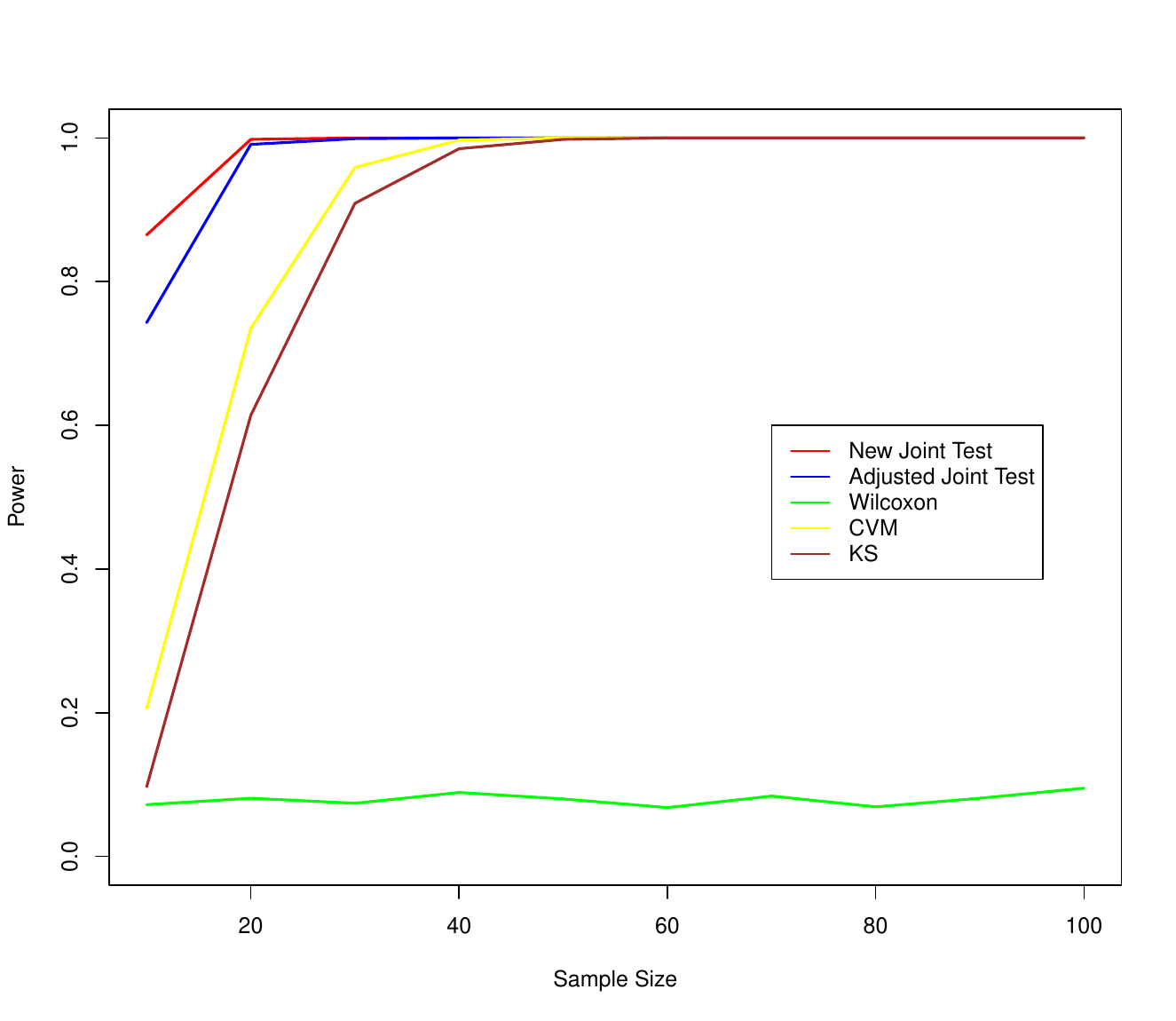}
\caption{Power for  $N(0.4, 1)$ vs. Beta$(2, 3)$} \label{pow:3}
\end{subfigure}\hspace*{\fill}
\begin{subfigure}{0.4\textwidth}
\includegraphics[width=\linewidth]{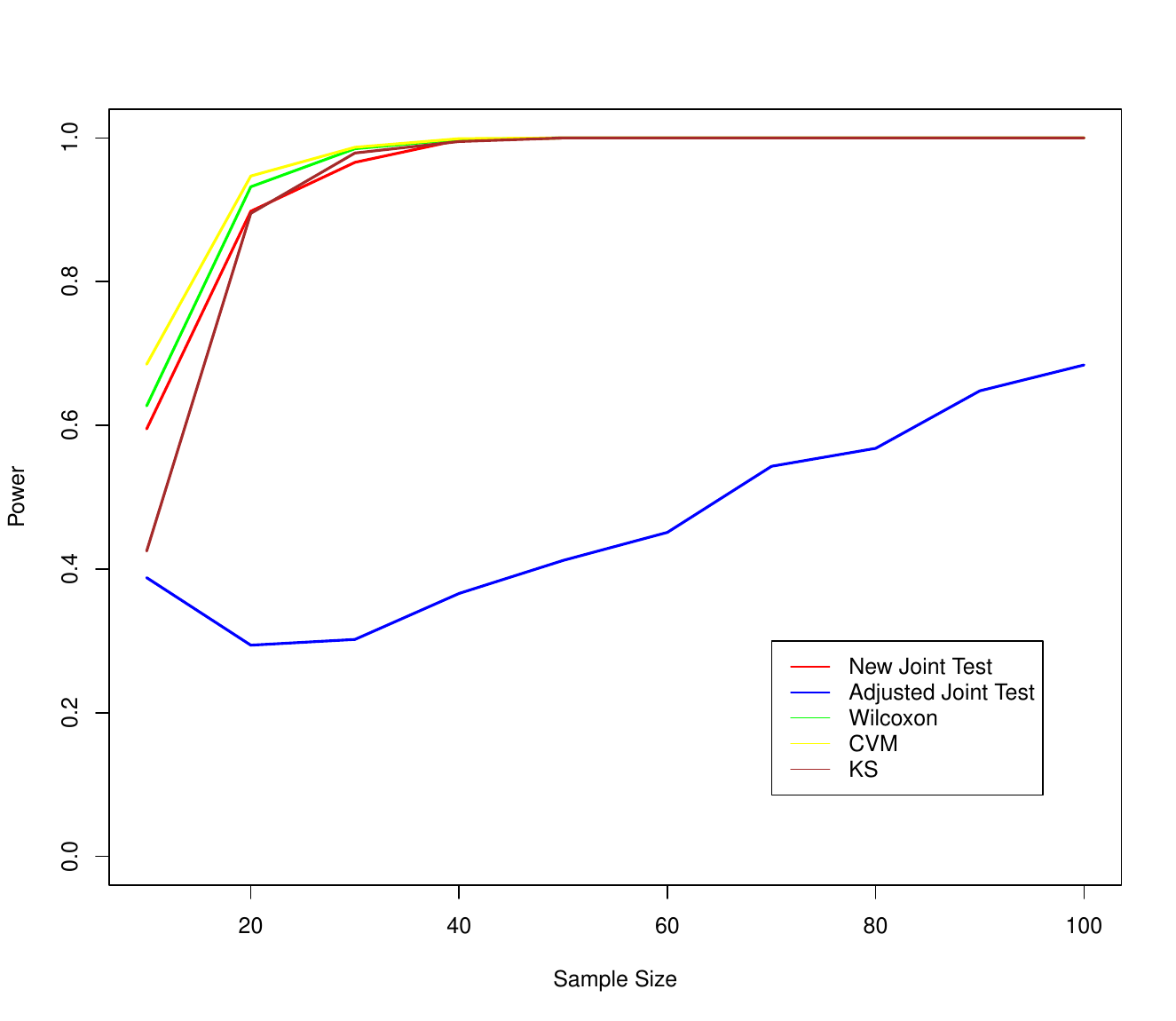}
\caption{Power for $N(2, 1)$ vs. Exp$(1)$ } \label{power-n21}
\end{subfigure}

\medskip
\begin{subfigure}{0.4\textwidth}
\includegraphics[width=\linewidth]{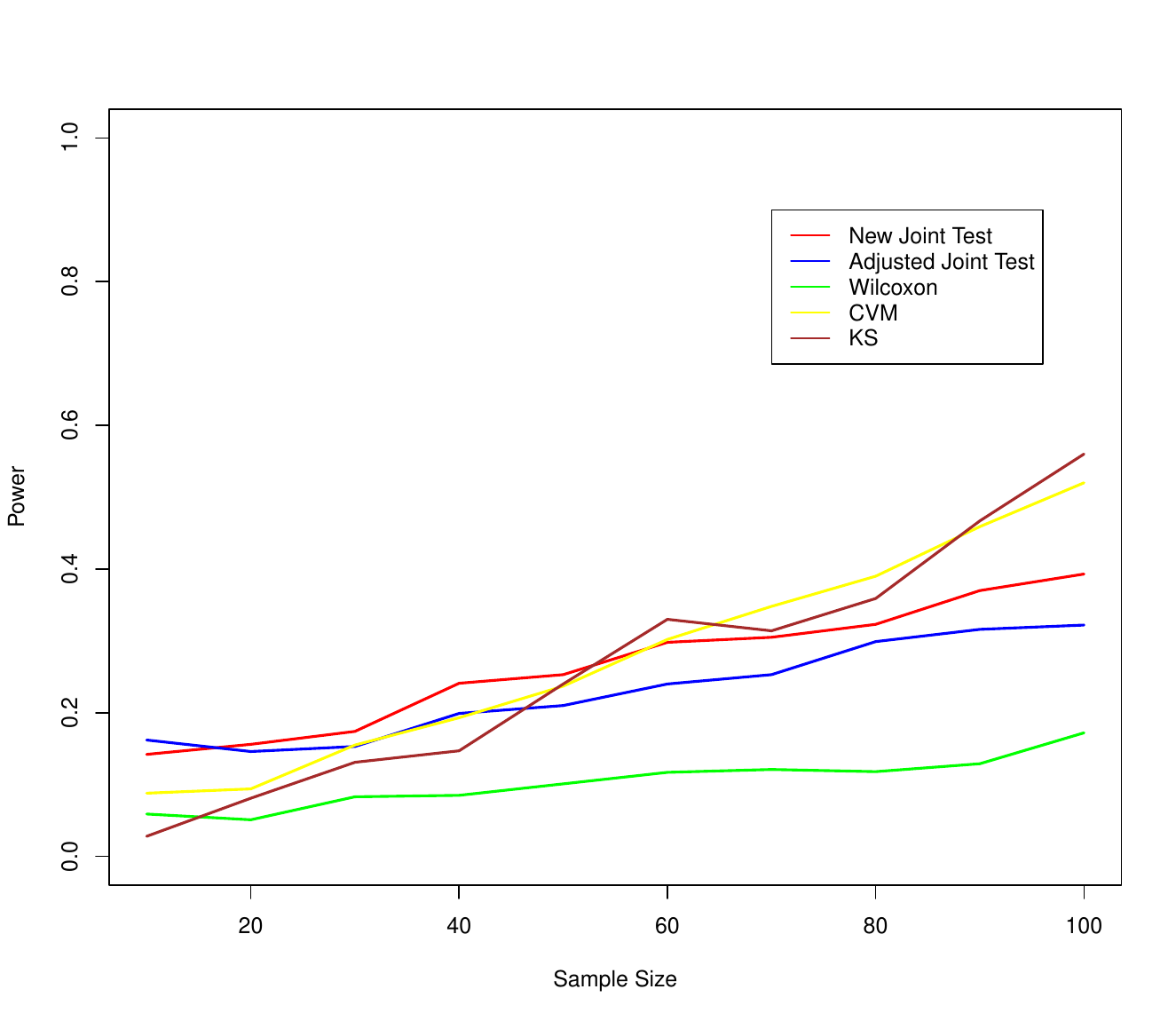}
\caption{Power for $N(2, 1)$ vs. Exp$(1)$ } \label{pow:5}
\end{subfigure}\hspace*{\fill}
\begin{subfigure}{0.4\textwidth}
\includegraphics[width=\linewidth]{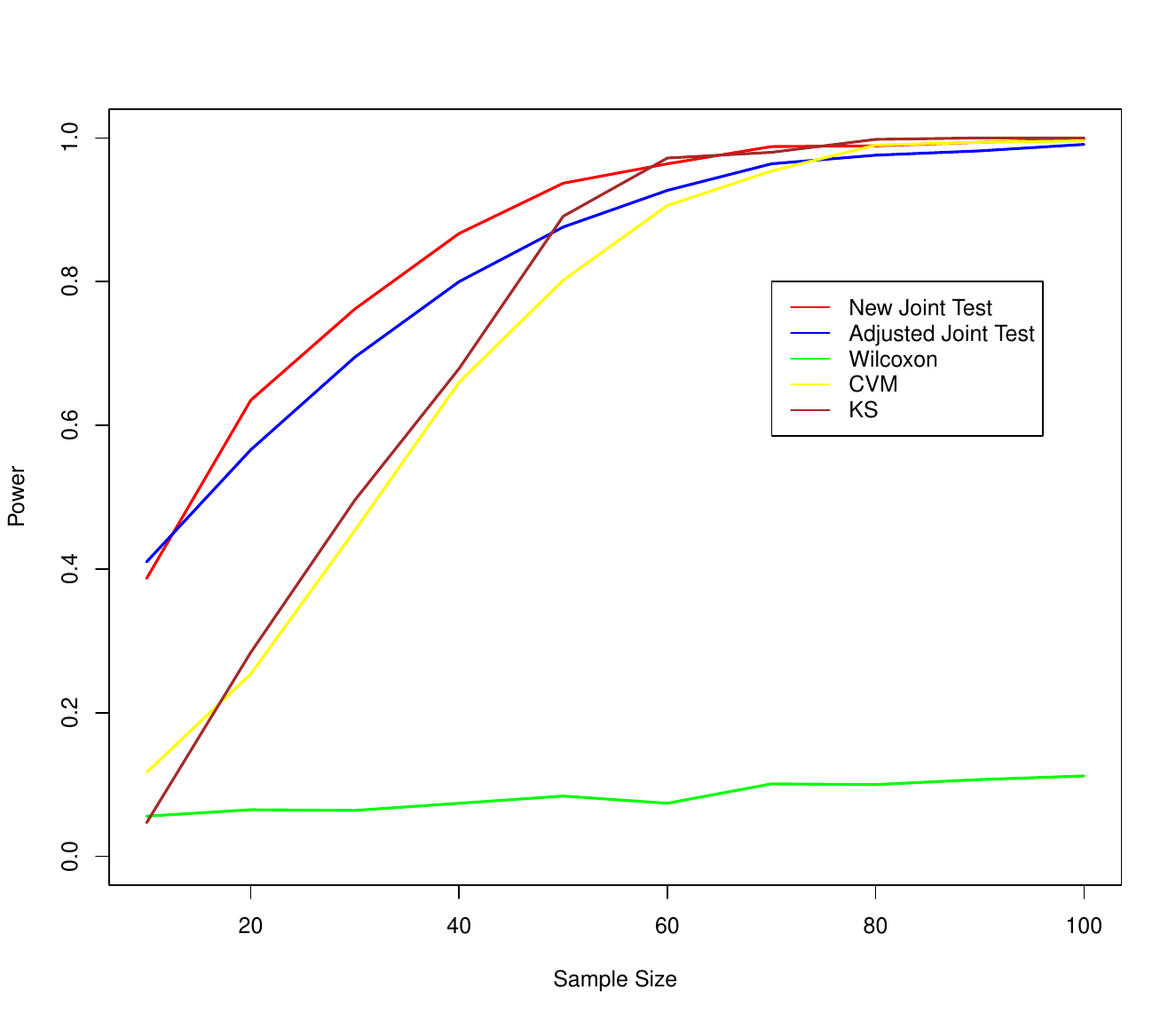}
\caption{Power for $N(1,2)$  vs. $\chi_1^2$ }  \label{pow:6}
\end{subfigure}
\caption{Empirical power of different nonparametric tests, for different pairs of underlying  distributions} \label{power} 
\end{figure}

Results from an empirical power comparison with different underlying distributions can be seen in Figure \ref{power}. 
Figures \ref{pow:1} and \ref{pow:2} show situations were the two samples originated from distributions with the same expectation but different variances. 
Here, both new test procedures performed quite well compared to the classical omnibus tests, especially for small sample sizes.
Not surprisingly, the rank sum test (Wilcoxon, Mann, Whitney) which does not include pure scale changes in its consistency region, was not able to detect such alternatives. 
Among the remaining tests, the Kolmogorov-Smirnov test which already showed conservative behavior in the type I error simulation was also the most conservative with regard to power.

In the setting displayed in Figure \ref{power-n21}, both underlying distributions have the same variance, but different location and shape. Here, the new joint test performed quite similar as the classical tests.
In this situation, the procedure based on the adjusted estimators led to lower power. This may be due to bias induced while ``simplifying'' the statistical functionals. 
Therefore, we only recommend the ``adjusted joint test'' if the "new joint test'' is not applicable, which is the case when the assumption of absolutely continuous distributions is not tenable. 
In all other cases, both newly proposed tests, with and without adjustment, yielded satisfactory power.  The new joint (unadjusted) testing method appears advisable in every simulated situation.
In Figures \ref{pow:5} and \ref{pow:6}, we used distributions with the same expectation and the same variance, but different shapes. The new joint test again performed rather well in comparison with the other test procedures. For small and moderate samples, it even showed the best performance of all tests considered.

Figure \ref{fig:power_norm-sig} adds results from a simulation addressing the small sample situation (sample size $10$). 
Here, the two-sample comparison involved one sample from a standard normal distribution, and another sample from an $N(0, \sigma^2)$ distribution with $\sigma^2$ ranging between $1$ and $10$. 
Thus, only a scale difference is present in this situation. 
We had already noted earlier that the 
``adjusted joint test" was too liberal for very small sample sizes. 
In the figure, its simulated performance is only shown for completeness of the comparison.
In this small sample situation, 
the simulated power of our newly proposed (unadjusted) joint test exceeded that of the classical omnibus tests substantially.
The rank sum test is not consistent for this alternative.

\begin{figure}
\centering
    \includegraphics[width=0.65\linewidth]{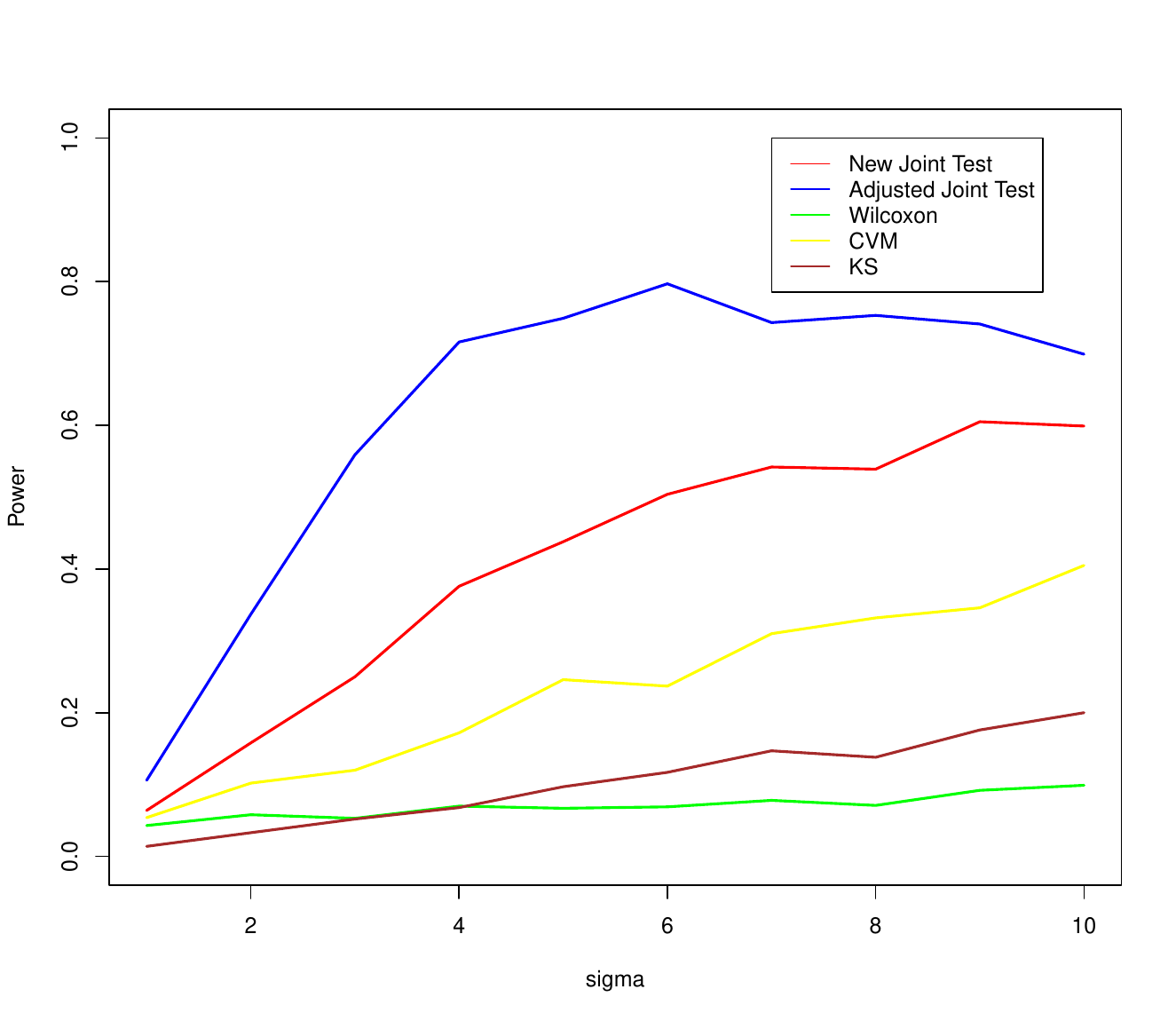}
    \caption{Power of the Tests for $N(0, 1)$ vs. $N(0, \sigma)$ }
    \label{fig:power_norm-sig}
\end{figure}

Of course the simulation results presented here are only a short excerpt from all different settings of underlying distributions that we simulated. But we believe these few examples give an overview on the advantages and disadvantages of our proposed testing methods. 

\subsection{Simultaneous Confidence Intervals} \label{sec_sim_ci}
In the following, we evaluate the performance of the simultaneous confidence intervals mentioned in Section \ref{secresamp}. 

\begin{table}[H]
    \centering
    \begin{tabular}{|c|c|c|c|c|c|c|c|c|c|c|}
  $X,Y$ / $n=m$  & 10 & 20 & 30 & 40 & 50 & 60 & 70 & 80 & 90 & 100 \\
  \hline
   I  & & & 0,953 & 0,956 & 0,958 & 0,961 & 0,965 & 0,9635 & 0,951 & 0,966 \\
   \hline
   II  & &  0,962 & 0,955 & 0,968 & 0,971 & 0,969 & 0,968 & 0,983 & 0,967 & 0,971  \\  
   \hline
  % III  &  & & & 0,929 &  0,944 & 0,936 & 0,941 & 0,938 & 0,930 & 0,937   \\
  % \hline
 %  IV  &    \\
 %  \hline
   III  &   &  0,974 & 0,966 & 0,971 & 0,969 & 0,969 & 0,971 & 0,963 & 0,969 & 0,970 \\
   \hline
   IV  &  & & 0,939 &0,948 & 0,952 & 0,961 & 0,961 & 0,960 & 0,961 & 0,954
    \end{tabular}
    \caption{Coverage Probability for Multivariate Normal CI}
\end{table}
\begin{table}[h]
    \centering
    \begin{tabular}{|c|c|c|c|c|c|c|c|c|c|c|}
  $X,Y$ / $n=m$   & 10 & 20 & 30 & 40 & 50 & 60 & 70 & 80 & 90 & 100 \\
  \hline
   I  & 0,950 & 0,961 & 0,971 & 0,968 & 0,971 & 0,969 & 0,975 & 0,971 & 0,963 & 0,974
  \\
   \hline
   II  &  0,941 & 0,951 & 0,945 & 0,958 & 0,964 & 0,958 & 0,955 & 0,970 & 0,959 & 0,957
  \\  
  % \hline
 %  III  &  0,889 & 0,944 & 0,965 & 0,961 & 0,967 & 0,967 & 0,959 & 0,973 & 0,961 & 0,964 \\
%   \hline
 %  IV  &  0,474 & 0,668 & 0,770 & 0,799 & 0,832 & 0,874 & 0,885 & 0,897 & 0,905 & 0,909    \\
   \hline
   III  &  0,961 & 0,978 & 0,975 & 0,978 & 0,975 & 0,972 & 0,973 & 0,967 & 0,973 & 0,974
     \\
   \hline
   IV  &   0,927 & 0,957 & 0,962 & 0,964 & 0,963 & 0,972 & 0,972 & 0,969 & 0,969 & 0,960
    \end{tabular}
    \caption{Coverage Probability for Empirical Bootstrap CI}
\end{table}
\begin{table}[ht]
    \centering
    \begin{tabular}{|c|c|c|c|c|c|c|c|c|c|c|}
 $X,Y$ /$n=m$   & 10 & 20 & 30 & 40 & 50 & 60 & 70 & 80 & 90 & 100 \\
   \hline
   I  & & & 0,223 & 0,192 & 0,172 & 0,157 & 0,145 & 0,136 & 0,128 & 0,122
  \\
   \hline
   II  &  &  0,375 & 0,301 & 0,256 & 0,226 & 0,202 & 0,187 & 0,170 & 0,160 & 0,150 \\
  % \hline
  % III  &  & & & 0,257 & 0,198 & 0,183 & 0,169 & 0,157 & 0,149 & 0,141  \\
 %  \hline
 %  IV  &    \\
   \hline
  III  &   &   0,321 & 0,259 & 0,224 & 0,199 & 0,182 & 0,168 & 0,157 & 0,148 & 0,141  \\
   \hline
   IV  &   & & 0,217 & 0,189 & 0,168 & 0,154 & 0,143 & 0,134 & 0,126 & 0,119
    \end{tabular}
    \caption{Interval Length for Multivariate Normal CI}
\end{table}
\begin{table}[ht]
    \centering
    \begin{tabular}{|c|c|c|c|c|c|c|c|c|c|c|}
  $X,Y$ / $n=m$   & 10 & 20 & 30 & 40 & 50 & 60 & 70 & 80 & 90 & 100 \\
  \hline
   I  &   0,392  & 0,279 & 0,228 & 0,197 & 0,177 & 0,161 & 0,150 & 0,140 & 0,132 & 0,126
  \\
   \hline
   II  &   0,373 & 0,301 & 0,258 & 0,228 & 0,205 & 0,190 & 0,174 & 0,164 & 0,154
  \\ 
 %  \hline
  % III  &  0,417 & 0,333 & 0,283 & 0,248 & 0,222 & 0,206 & 0,190 & 0,177 & 0,168 & 0,159  \\
 %  \hline
%   IV  &   0,148 & 0,107 & 0,095 & 0,083 & 0,077 & 0,074 & 0,071 & 0,066 & 0,063 & 0,061  \\
   \hline
   III  &  0,458 & 0,319 & 0,259 & 0,224 & 0,199 & 0,182 & 0,168 & 0,157 & 0,148 & 0,141
   \\
   \hline
   IV  &   0,365 & 0,269 & 0,220 & 0,192 & 0,171 & 0,157 & 0,146 & 0,136 & 0,129 & 0,122
    \end{tabular}
    \caption{Interval Length for Empirical Bootstrap CI}
\end{table}
The ``exact'' values of $\theta$ and $I_2$ for each combination of distributions can be computed by numerical integration. 
They are given as follows, rounded to four decimal places:\\
 I : $X \sim N(0,1)$, $Y \sim N(1,1)$.
Exact values: $\theta = 0.7602$ and $I_2 =0.3645$. \\
II : $X \sim N(0,4)$, $Y \sim U[-0.5,0.5]$. 
Exact values: $\theta = 0.5$ and $I_2 = 0.9008$. \\
% III : $X \sim N(1,1)$, $Y \sim U[-0.5,0.5]$. 
 %Exact values: $\theta = 0.8315$ and $I_2 =  0.3370$. \\
%IV : $X \sim N(2,1)$, $Y \sim U[-0.5,0.5]$. 
%Exact values: $\theta =  0.9727$ and $I_2 =  0.0546$. \\
III : $X \sim N(1,1)$, $Y \sim $ Exp$(1)$. 
Exact values: $\theta =  0.5381 $ and $I_2 =  0.5389$. \\
IV : $X \sim N(2,1)$, $Y \sim $ Exp$(1)$. 
Exact values: $\theta =  0.7895$ and $I_2 =  0.2794$.  \\
For each setting, we simulated $1,000$ data sets, 
and for each data set, we generated $1,000$ bootstrap samples. 
For the calculation of the simulated coverage probability, a simulated confidence region was considered as ``covering`` if it covered both components of the parameter vector.
The length of each simulated confidence region was calculated by the Euclidean distance: 
\begin{align*}
\frac{1}{2} \sqrt{(\theta_{upper}- \theta_{lower})^2+(I_{2,upper}- I_{2,lower})^2}.
\end{align*}
The method based on the asymptotic multivariate normal approximation technically requires some overlap between the support of the distributions. However, for small sample size settings, perfectly separated samples may still occur. Thus, for this method, the algorithm stops if perfect sample separation in the simulated data occurs. A similar situation can occur in hypothesis testing where we switched from asymptotic to exact tests in these cases. This would also be our recommendation for such situations.

All methods performed quite similarly. The desired confidence level of $95 \%$ was reached even for small sample sizes. The interval length decreased fast for increasing sample sizes. For further simulation studies including even more methods, we refer to Section \ref{add_sim_ci} 
in the Appendix.

\section{Case Study}
In order to illustrate application of the proposed methods, we used data from the Paracelsus 10,000 Study. This is a prospective cohort study with more than 10,000 participants in the city of Salzburg and nearby regions. 

All participants were between 40 and 75 years old at the baseline examination. More details on the study can be found in~\cite{frey2023paracelsus}. 
We chose the variables height, weight, heart rate, systolic blood pressure, serum iron, fasting glucose, C-reactive protein (CRP, a marker of inflammation), hours of sport per week, and Beck's Depression Inventory II (BDI-II) for an illustrative analysis, respectively comparing the values for men and women. 
For 8,729 participants (4,485 women and 4,244 men), data on all mentioned variables were available. 
All variables except the BDI-II, which is ordinal, were measured at a ratio scale. 
However ties were still frequent, either due to the precision of the measurement instrument, 
or because some values were more common (e.g., the value 0 in the sports variable). 
Boxplots of the values for men and women can be seen in Figure~\ref{fig:example_boxplots}. Joint confidence regions using a multivariate normal approximation with bootstrapped covariance, as described in Section~\ref{sec:sim_conf_int} can be seen in Figure~\ref{fig:example_confidence_intervals}. The distributions of male vs.~female participants differ significantly via $\theta$ and/or $I_2$ in each variable. Note however that for CRP and Sport the difference is only detectable by the overlap index, with the values of males having a smaller spread than the values of females for CRP and a higher spread for Sport.

\begin{figure} [ht!]
    \centering
    \includegraphics[width =0.65 \textwidth]{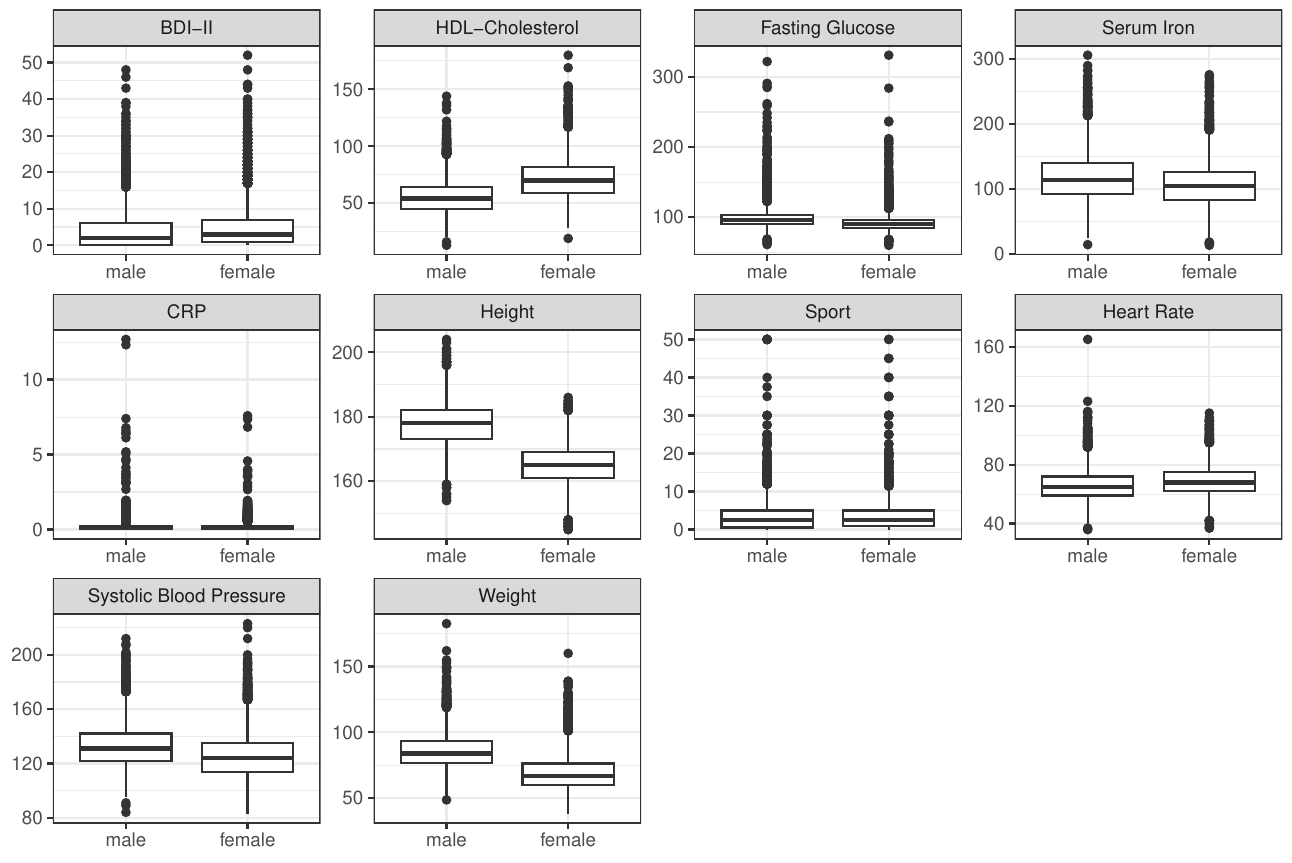}
    \caption{Boxplots for the chosen variables in men and women participating in the Paracelsus 10.000 study.}
    \label{fig:example_boxplots}
\end{figure}

\begin{figure}[ht!]
    \centering
    \includegraphics[width = 0.9\textwidth]{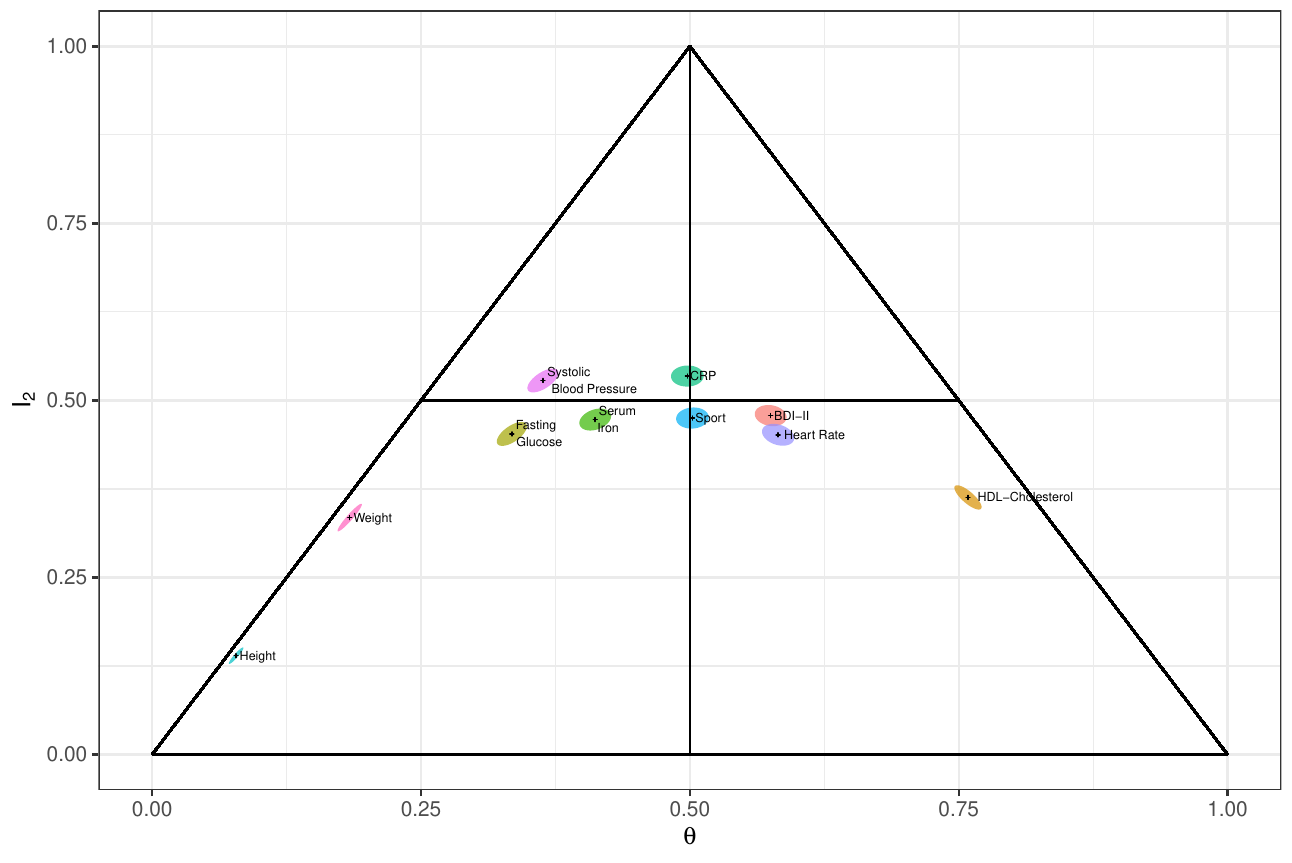}
    \caption{$\theta$ and $I_2$ for men vs women for each variable. Point estimates are shown as crosses and joint confidence regions as coloured ellipses. The possible joint region for $(\theta, I_2)$ and the values of $0.5$ for both functionals are represented by black lines. Values to the left/right of the vertical line indicate larger/smaller values for men vs.~women. Values above/below the horizontal line indicate that values of the male sample have larger/smaller than $0.5$ probability of being contained within one value above and one value below the median of the female sample.}
    \label{fig:example_confidence_intervals}
\end{figure}

\section{Discussion}
In this manuscript, we have proposed a new inferential approach for nonparametric comparison of two different distributions 
(or two independent samples)
that is capable of substantially expanding the consistency set of classical nonparametric procedures, 
without sacrificing effectiveness in those situations where the classical procedures can be applied. 
Our approach is based on combining different 
nonparametric functionals, 
deriving appropriate estimators and their  asymptotic joint (multivariate) distribution, 
and devising practically useful finite sample approximations which we have evaluated in a simulation study.

The functionals utilized in our combination approach are the classical Mann-Whitney parameter (nonparametric relative effect) 
and the recently introduced overlap index which was developed to quantify ecological niches. 
Our new almost omnibus inference framework is indeed fully nonparametric. 
The idea of using a location scale model within a semiparametric statistical context is quite old, 
going back to \cite{lepage},
and it has gained attention again recently (\cite{marozzi_09}, \cite{marozzi_13} and \cite{SOAVE2015125}). 
However, to the best of our knowledge, 
a fully nonparametric approach to this situation did not exist yet.

We have developed two different testing procedures. The first one (``new joint test'') is based on the theory of empirical processes and the close relation of the presented statistical functionals to the empirical ROC-curve, while the second one (``adjusted joint test'') uses the theory of multiple contrast tests. 
In our simulations, the first method was almost uniformly superior to the second one. 
Based on these empirical evaluations, we would therefore recommend to use the first method whenever the assumption of continuous distributions can be justified. 

Our proposed method has the additional advantage that simultaneous confidence intervals can be calculated, in line with the requirement that, for example in clinical data analyses, confidence intervals should be provided wherever possible. 
The major added benefit of confidence intervals as compared to hypothesis tests is that users can see if there is a difference in location or in dispersion between the distributions of interest. We have illustrated this in the case study based on a large epidemiological cohort. 

As our almost omnibus approach is new, there are still several follow-up research questions. 
First, an extension to more than two samples would be useful. However in the case of three or more independent samples, it is well known \citep{THANGAVELU2007720} that there is the possibility of intransitivity regarding pairwise nonparametric effect measures. 
Tackling this issue will therefore be a fundamental task in such an extension. 
%Second it is not clear if the split of the sample at the median is optimal. For some research questions it could be useful to split the sample at a different quantile. 
%or divide the sample at a sequence of quantiles. %Developing an asymptotic theory for this case and investigate the optimality of this different approaches will be an interesting research direction. 
Finally, the new approach is currently restricted to univariate samples. 
Up to now, no satisfactory extensions to a multivariate relative effect or a multivariate overlap index exist. 
Some rather new contributions try to use a component-wise analysis (\cite{kong_20} and \cite{kong22}) for the relative effect, and different ideas have been proposed in order to quantify multivariate niche overlap (\cite{geange}, \cite{blonder_14} and \cite{swanson_15}). 
The possible link between multivariate  overlap of niches and concepts of data depth (\cite{liu_99}, \cite{chernozhukov}) and multivariate ranks (\cite{hallin_21} and \cite{deb_sen}) needs to be explored further, as well.

\appendix

\section{Motivating Example}
As mentioned in the main part of the manuscript, there is to the best of our knowledge no purely nonparametric method for comparing two independent samples with a good small and medium sample size performance and the capability to detect general differences between distributions. 
In order to illustrate this, we present an example from \cite{oleynik}. They suggested that research studies in education science should not only be concerned with location effects, but also with possible differences in the variability. 
As an illustration, consider a teacher who pays more attention to students having difficulties with the topic. 
By doing so, the teacher may reduce the variability of the results. 

%Of course the other way around is also possible. 
The following data show achievement scores on an examination (10  highest possible score) given to two classes of ten students each. 
%The first class was taught by a teacher with a masters degree and the second one by a teacher who only has a bachelor degree: 
The scores for class 1 are:  $7,4,4,5,4,6,6,4,3,7$ and for class 2: $3,6,7,9,3,2,4,8,2,6 $.

\begin{figure}[ht!]
\begin{subfigure}{0.5\textwidth}
\includegraphics[width=\linewidth]{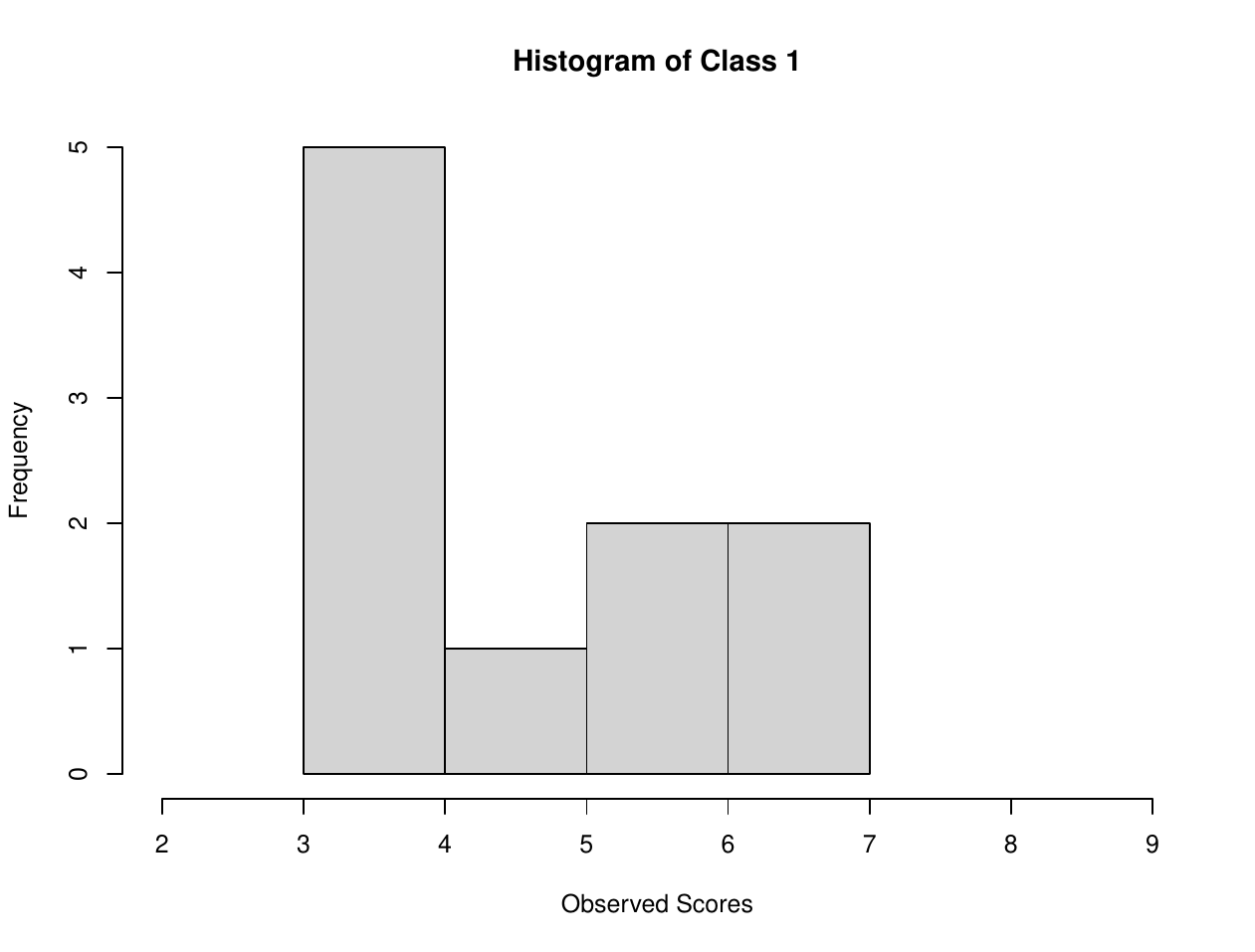}
%\caption{Histogram of the data for the first class} 
\end{subfigure}\hspace*{\fill}
\begin{subfigure}{0.5\textwidth}
\includegraphics[width=\linewidth]{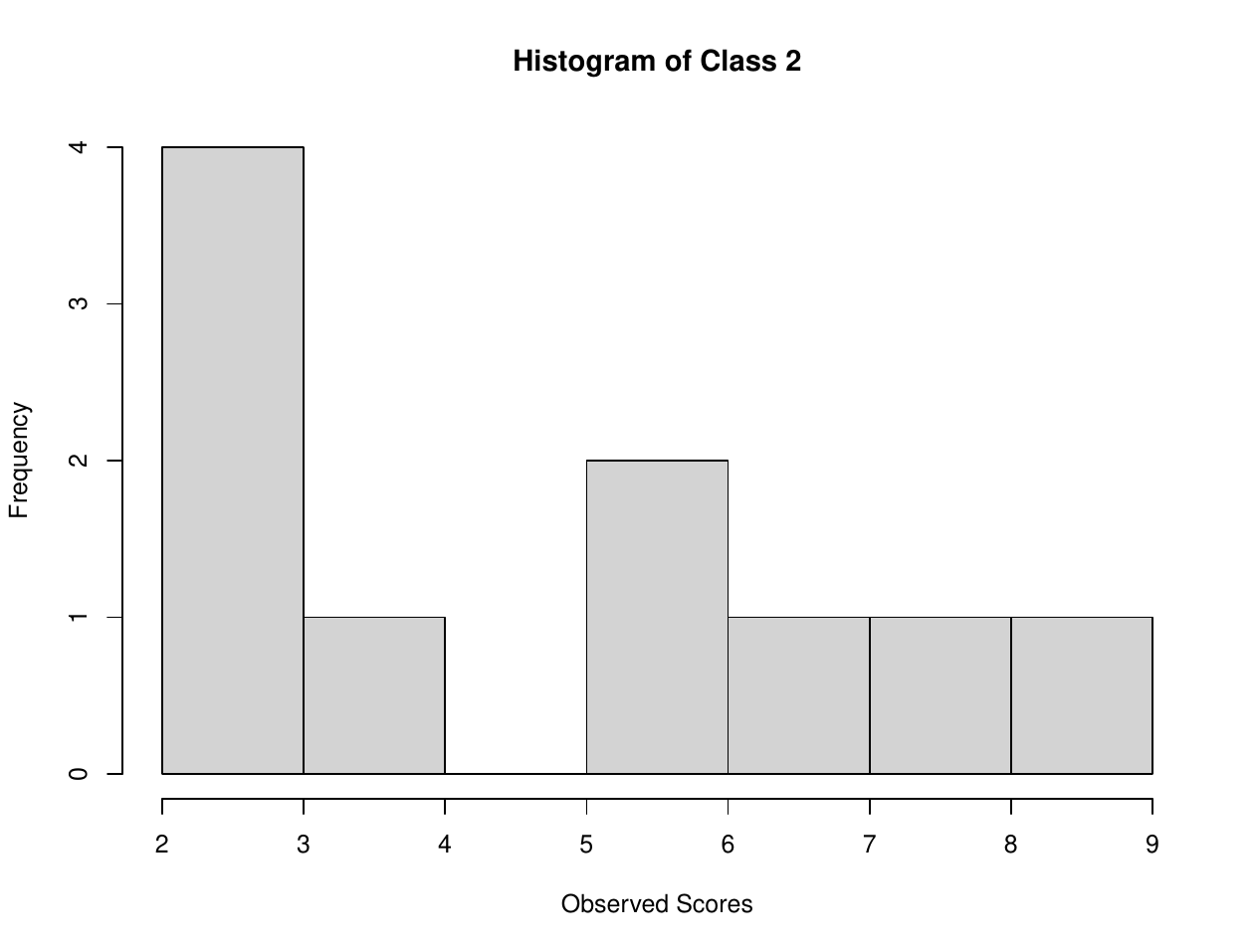}
%\caption{Histogram of the data for the second class} 
\end{subfigure}
\caption{Histogram of the data} \label{data_example}
\end{figure}

The low power of the Kolmogorov-Smirnov test in small-sample situations is reflected by a large p-value of $p=0.62$ in this example, while the WMW rank sum test is not consistent against scale differences and in this case yields an even larger p-value ($p=0.85$). 
Our newly proposed test is capable of recognizing the difference that is present in the variability, which can be seen easily in Figure \ref{data_example}, it results in a much smaller p-value of $p=0.015$.  
The complete study consists of $10$ different data sets, with differences in location and variability.
Thus, only a test which can detect both alternatives is useful here. A difference in each of the two functionals is of practical relevance for the study.

Our newly introduced method has the major advantage that it can detect alternatives in the variability, while still being consistent for all the alternatives that the classical WMW test can detect (location/stochastic tendency), with only a minor loss in power. Therefore, the pair $ (\theta, I_2)$ can be seen as a natural extension of the Mann-Whitney parameter $\theta$ to quantify simultaneously the stochastic tendency and the dispersion of two distributions. Additionally, the procedure is completely nonparametric, thus not relying on specific parametric or semiparametric model formulations. 
\section{Visualization of Functionals}
In this example we evaluate the previously introduced parameters numerically for $X \sim N(0,1)$ and $Y \sim N( \mu, \sigma^2)$. In the following plots we let 
$\mu \in [-5,5]$, and $\sigma \in [0.01, 5] $ vary. The black line shows, respectively, where each of the two parameters equals $1/2$ and thus could not separate the two distributions when considered alone. 
Figure \ref{Re_normal} displays the relative effect, while Figure \ref{I2_normal} and Figure \ref{I1_normal} show the overlap indices $I_1$ and $I_2$.
\begin{figure}[ht!]
    \centering
    \includegraphics[width = 0.7\textwidth]{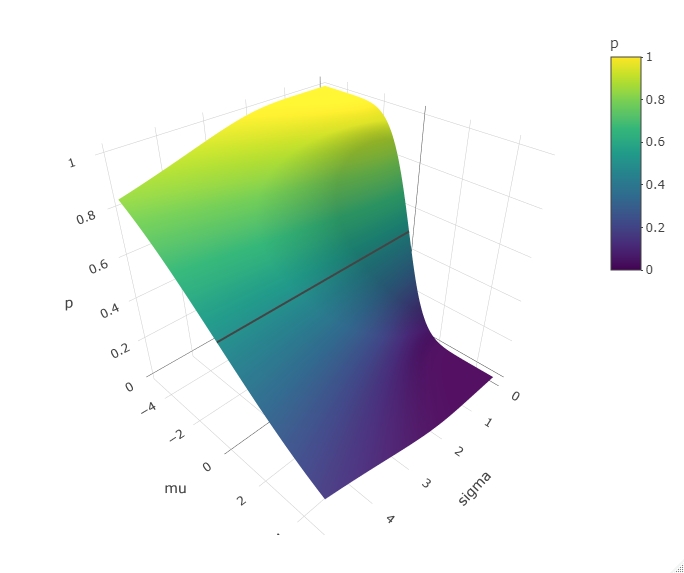}
    \caption{Relative Effect $\theta$ for normally distributed random variables $N(\mu, \sigma^2)$ against a standard normal variable}
    \label{Re_normal}
\end{figure}
\begin{figure}[ht!]
    \centering
    \includegraphics[width = 0.7\textwidth]{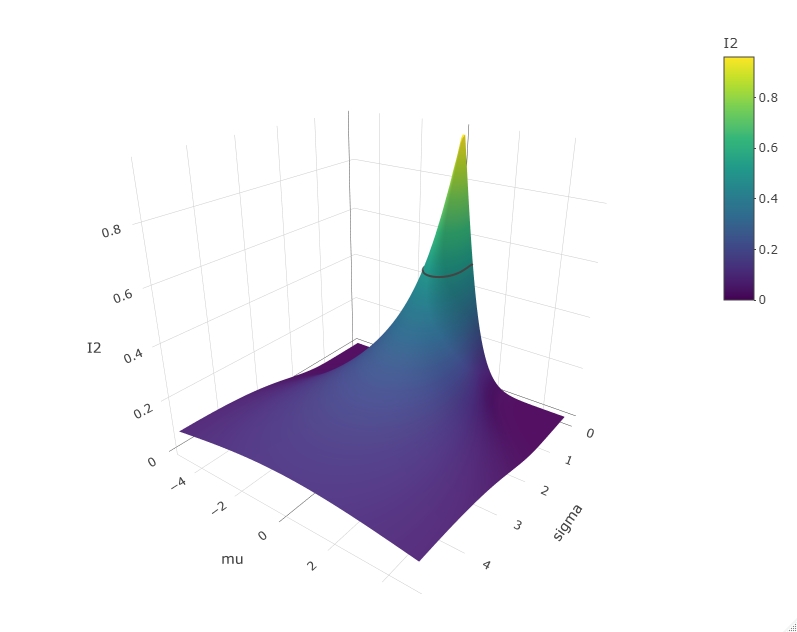}
    \caption{Overlap Index $I_2$ for normally distributed random variables $N(\mu, \sigma^2)$ against a standard normal variable}
    \label{I2_normal}
\end{figure}
\begin{figure}[ht!]
    \centering
    \includegraphics[width = 0.7\textwidth]{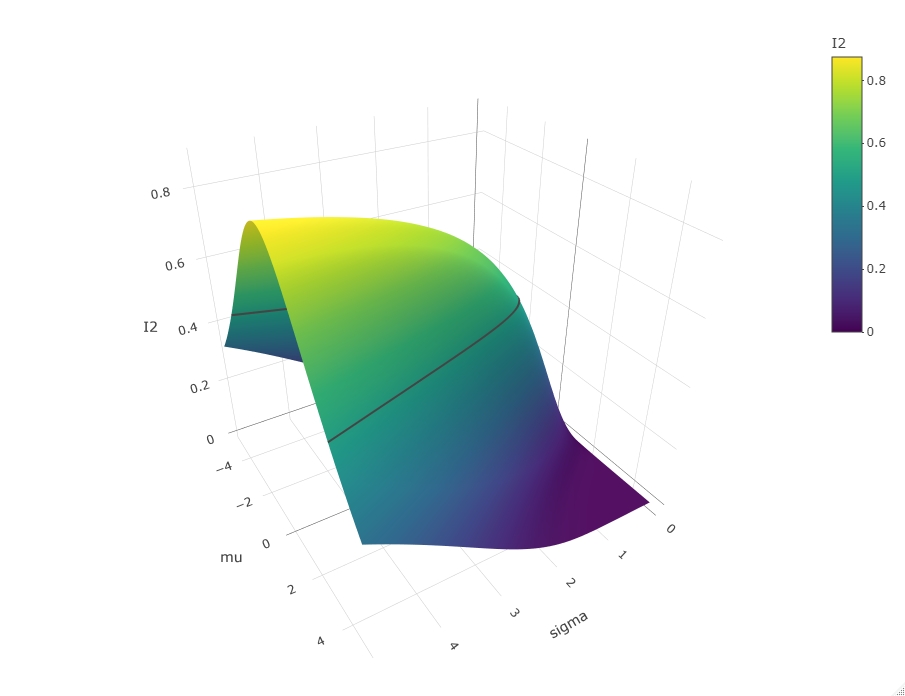}
    \caption{Overlap Index $I_1$ for normally distributed random variables $N(\mu, \sigma^2)$ against a standard normal variable}
    \label{I1_normal}
\end{figure}

\section{Alternative Methods for Constructing Simultaneous Confidence Intervals} \label{alt_ci}
There exist several different options for constructing bootstrap based simultaneous confidence intervals. 
Here, we briefly review a few possible methods: 
\begin{enumerate}
\item {\bf Asymptotically normal with bootstrapped variance (Bonferroni adjusted)} \\
    Both estimators $\hat{\theta}_{mn}$ and $\hat{I}_{2,mn} $ are asymptotically univariate normally distributed. Therefore, we can derive for each of the  estimators an asymptotic confidence interval, and then adjust them with a Bonferroni correction. 
    Here, we estimate both variances by the variance of the bootstrapped estimates. 
    \item {\bf Empirical bootstrap based on Mandel and Betensky } \label{mb} \\
    The Bonferroni adjustment can be quite conservative. Therefore, we consider another approach in order to obtain shorter intervals. 
    A first idea for bootstrapped simultaneous confidence intervals was given by \cite{davison_hinkley_1997} who calculate the coverage of simultaneous confidence intervals by counting the number of bootstrap samples outside of the confidence region and repeat this algorithm for different confidence limits until the desired coverage rate is achieved. This ``trial and error" approach is quite inefficient. Therefore, we consider instead the algorithm introduced by \cite{mandel_betensky}:
    \begin{enumerate}
        \item Generate $B$ bootstrap samples from $\hat{F}$. For each sample calculate the bootstrap estimates $(\Tilde{\theta}_b, \Tilde{I}_{2,b})$.
        \item Order the bootstrap estimates $\Tilde{\theta}_{(1)}, \Tilde{\theta}_{(2)}, \ldots, \Tilde{\theta}_{(B)}$ and $\Tilde{I}_{2,(1)}, \Tilde{I}_{2,(2)}, \ldots, \Tilde{I}_{2,(B)}$. In the case of ties, we use random ranks. We define by $r_{\theta}(b)$ the rank of $\Tilde{\theta}_b$ and analogously by $r_{I_2}(b)$ the rank of $\Tilde{I}_{2,b}$. We then define $r(b)= \max \{ r_{\theta}(b), r_{I_2}(b)\} $.
        \item Calculate $r_{1- \alpha/2}$, the $1-\alpha/2$ percentile of $r(b)$ 
        \item Take the upper limits of the simultaneous confidence interval to be $ \Tilde{\theta}_{(r_{1- \alpha/2})}$ and $\Tilde{I}_{2,(r_{1- \alpha/2})}$
    \end{enumerate}
    The lower limit can be calculated analogously.
    \item {\bf Empirical bootstrap based on Gao, Konietschke and Li } \\
    The previous method was improved by \cite{gao_21}, by sharpening the intervals for both coordinates. Their algorithm can be presented as follows:
    \begin{enumerate}
        \item Repeat the first three steps as in the algorithm in \ref{mb}
        \item We call the collection of bootstrap samples with the maximum sample rank equal or below $r_{1- \alpha/2}$ by $\Phi$, $b \in \Phi$ iff $r(b) \leq r_{1- \alpha/2}$
        \item Order the bootstrap estimates in set $ \Phi$. We denote the new rank of $\Tilde{\theta}_b$ and $ \Tilde{I}_{2,b}$ in set $\Phi$ by $r_{\theta}'(b)$ and $r_{I_2}'(b)$, respectively.  Then we call by $r'(b)= \min \{ r'_{\theta}(b), r'_{I_2}(b)\} $ the smallest rank associated with the $b$th sample and set $\Phi$. Calculate then the $\frac{\alpha}{2-\alpha}$ percentile $r'_{\frac{\alpha}{2-\alpha}}$ of $r'(b)$.
        \item Denote by $\Psi= \{ b \in \Phi | r'(b) \geq r'_{\frac{\alpha}{2-\alpha}} \}$ the collection of bootstrap samples within $\Phi$ with the minimum sample rank above $r'_{\frac{\alpha}{2-\alpha}}$
        \item Calculate $t_{\theta}= \max_{b \in \Psi} r_{\theta}(b)$ and $t_{I_2}= \max_{b \in \Psi} r_{I_2}(b)$. And analogously $w_{\theta}= \min_{b \in \Psi} r_{\theta}(b)$ and $w_{I_2}= \min_{b \in \Psi} r_{I_2}(b)$. The simultaneous confidence interval for $\theta$ and $I_2$ are then given by the upper limits $(\Tilde{\theta}_{t_{\theta}}, \Tilde{I}_{2,t_{I_2}})$ and the lower limits $(\Tilde{\theta}_{w_{\theta}}, \Tilde{I}_{2,w_{I_2}})$
    \end{enumerate}
    \cite{gao_21} proved that this method produces  two-sided confidence intervals which have a simultaneous coverage of $(1-\alpha)$ and are admissible at level $(1-\alpha)$.
    This means that there are no other two-sided simultaneous bootstrap confidence intervals which have $(1-\alpha)$ coverage, but are uniformly shorter than or equal in all coordinates and strictly smaller in at least one.
\end{enumerate}
We will compare all these approaches in a simulation study in Section \ref{add_sim_ci}. 
Note that only the method based on the asymptotic multivariate normal distribution of our quantities (as presented in section \ref{secresamp}) gives us an elliptical confidence region. The other four methods presented here are all based on some adjustment of univariate confidence intervals and have a rectangular structure.

\section{Proofs}
\subsection{Proof of Theorem \ref{thm_image}} \label{proofthm_image}
We will prove the theorem in two steps. First we show that Im$((\theta, I_2)) \subseteq A$ and then that $(\theta, I_2):\: S \times S \longrightarrow A$ is onto.
To show the first part, define
\begin{align*}
\tau_1 &:= \int_{-\infty}^{F^{-1}(\frac{1}{2})} G dF \text{ and}\\
\tau_2 &:= \int_{F^{-1}(\frac{1}{2})}^{\infty} G dF.
\end{align*}
Then
\begin{align*}
\theta &= \tau_1 + \tau_2 \text{ and}\\
I_2 &= 2(\tau_2 - \tau_1) \text{ and thus} \\
I_2 &= 2(\theta - 2\tau_1).
\end{align*}
Since $2\tau_1$ and $2\tau_2$ are both probabilities, it follows, that $\tau_1, \tau_2 \in [0, \frac{1}{2}]$.
It follows immediately that
\begin{equation*}
    I_2 = 2 \theta - 4 \tau_1 \leq 2 \theta
\end{equation*}
and
\begin{equation*}
    0 \leq 2 - 4\tau_2 \Rightarrow 2\tau_2 \leq 2 - 2\tau_2 \Rightarrow 2\tau_2 - 2\tau_1 \leq 2 - 2\tau_2 - 2\tau_1 \Rightarrow I_2 \leq 2 - 2\theta.
\end{equation*}
Thus we have
\begin{align*}
0 \leq I_2 \leq 2\theta \text{ and } &0 \leq I_2 \leq 2 - 2\theta.
\end{align*}
Thus, given a value for $\theta$, the overlap index $I_2$ (and by symmetry also $I_1$) must lie in the interval $[0, \min\{2\theta, 2 - 2\theta\}]$. The region of possible pairs of $(\theta, I_2)$ is displayed in Figure~\ref{fig:theta_I_2_relationship}.
%\begin{figure}
  %  \centering
  %  \includegraphics[width = 0.75\textwidth]{Grenzen_I_plot.pdf}
  %  \caption{The blue shaded region shows the theoretically possible pairs of $(\theta, I_2)$.}
 %   \label{fig:theta_I_2_relationship}
% \end{figure}
This completes the first part of the proof. 

For the second part it remains to be shown that for each pair of points $(\theta, I_2)$ in $A$ there exists a pair of cumulative probability functions $(F, G)$ whose corresponding relative effect and overlap index are $(\theta, I_2)$. Let $F$ be the cdf of a uniform distribution on $[0, 1]$ and $G$ the cdf of a uniform distribution on $[a, b]$ where $a \leq b$. Then, by directly solving the integrals in ~\eqref{def:rte} and~\eqref{def:niche}, we obtain
\begin{align*}
    (b - a) \cdot \theta = \begin{cases}
    b - a & \text{ if } a < 0 \land b < 0\\
    -0.5b^2 + b - a & \text{ if } a < 0 \land b \in \left[0, 1\right]\\
    0.5 - a & \text{ if } a < 0 \land b > 1\\
    0.5a^2 - 0.5b^2 - a + b & \text{ if } a \in [0, 1] \land b \in [0, 1]\\
    0.5a^2 - a + 0.5 & \text{ if } a \in [0, 1] \land b > 1\\
    0 & \text{ if } a > 1 \land b > 1
    \end{cases}
\end{align*}
and
\begin{align*}
    (b - a) \cdot I_2 = \begin{cases}
    0 & \text{ if } a < 0 \land b < 0 \text{ or } a > 1 \land b > 1\\
    b^2 & \text{ if } a < 0 \land b \in [0, 0.5]\\
    -b^2 + 2b - 0.5 & \text{ if } a < 0 \land b \in (0.5, 1]\\
    0.5 & \text{ if } a < 0 \land b > 1\\
    b^2 - a^2 & \text{ if } a \in [0, 0.5] \land b \in [0, 0.5]\\
    -b^2 + 2b - a^2 - 0.5 & \text{ if } a \in [0, 0.5] \land b \in [0.5, 1]\\
    -a^2 + 0.5 & \text{ if } a \in [0, 0.5] \land b > 1\\
    a^2 - b^2 - 2a + 2b & \text{ if } a \in (0.5, 1] \land b \in (0.5, 1]\\
    a^2 - 2a + 1 & \text{ if } a \in (0.5, 1] \land b > 1
    \end{cases}
\end{align*}

\begin{figure}
    \centering
    \includegraphics[width = 0.75\textwidth]{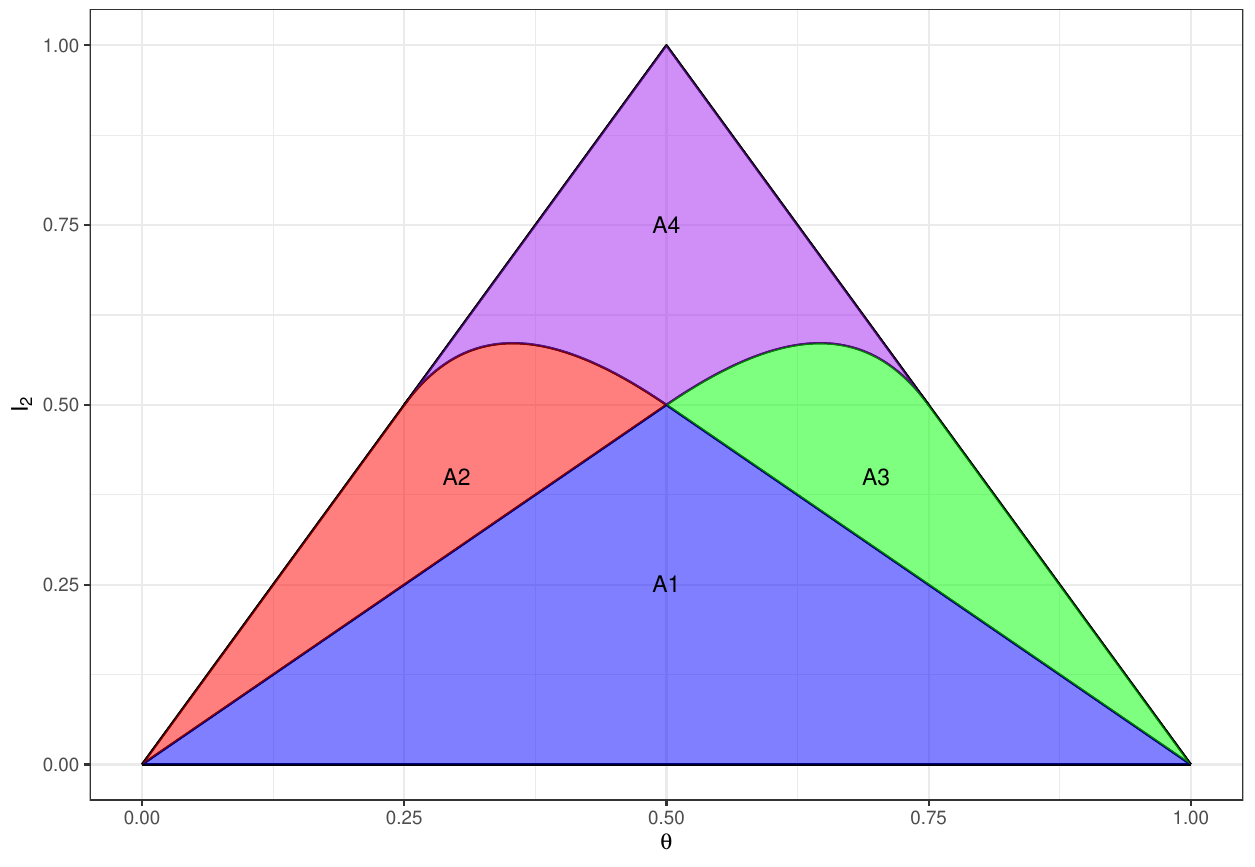}
    \caption{The possible pairs of $(\theta, I_2)$ can be divided into four subsets.}
    \label{fig:triangle_regions}
\end{figure}

We will split $A$ into four subareas, termed $A_1, A_2, A_3, A_4$ and illustrated in Figure~\ref{fig:triangle_regions}. For each of these areas we demonstrate here how to construct a pair of probability measures of the type described above for every point in the area.

$A_1$ is analytically defined as $\{(\theta, I_2) | \theta \in (0, 1), I_2 \in (0, \min\{\theta, 1 - \theta\})\}$. $A_1$ is covered by the case where $a < 0$ and $b > 1$. It follows from the case distinctions above, that in this case we have $I_2 = \frac{1}{1 - 2a} \theta$. That is, for fixed $a$, $I_2$ is a linear function of $\theta$. $\frac{1}{1 - 2a}$ is continuous in $a$ and ranges from asymptotically $1$ when $a \rightarrow 0$ to $0$ asymptotically as $a \rightarrow -\infty$. Given a certain value for $a$, the minimal $\theta$ and $I_2$ that can be achieved is $0$ asymptotically as $b \rightarrow \infty$. The maximal value is achieved asymptotically by letting $b \rightarrow 1$ and is $\frac{0.5 - a}{1 - a}$ for $\theta$ and $\frac{1}{2 - 2a}$ for $I_2$. Between these values, both $\theta$ and $I_2$ are strictly decreasing and continuous in $b$. Therefore, in the case where $a < 0$ and $b > 1$, we can create any pair $(\theta, I_2)$ in $A_1$ by choosing an incline via fixing $a$ and then choosing the distance from $(0, 0)$ by fixing $b$.

We will carry out the proof for $A_2$ in a similar manner. Here we use the case where $a \in [0, 0.5]$ and $b > 1$. This gives $I_2 = \frac{-a^2 + 0.5}{0.5a^2 - a + 0.5} \theta$, so $I_2$ is again a linear function of $\theta$ once $a$ is fixed and the incline is continuous in $a$ again. 
When $a = 0$ the incline is $1$, and when $a = 0.5$ the incline is $2$. 
For any $a$, we approach $(\theta, I_2) = (0, 0)$ as $b \rightarrow \infty$ and $I_2 = \frac{0.5 - a^2}{1 - a}$ and $\theta = \frac{0.5 a^2 - a + 0.5}{1 - a}$ as $b \rightarrow 1$. Solving the previous equation of $\theta$ for $a$ and inserting into the equation for $I_2$ gives $I_2 = \frac{-4\theta^2 + 4\theta - 0.5}{2\theta}$ as the maximal value for $I_2$ dependent on $\theta$. These considerations identify $A_2$ as $\{(\theta, I_2) | \theta \in [0, 0.5], I_2 \in [\theta, 2\theta] \text{ when } \theta < 0.25 \text{ and } I_2 \in [\theta, \frac{-4\theta^2 + 4\theta - 0.5}{2\theta} \text{ when } \theta > 0.25\}$ and prove that for every point $(\theta, I_2) \in A_2$ there exists a pair $(a, b) \in [0, 0.5] \times (1, \infty)$ such that $F$ and $G$ show the desired pair of relative effect and overlap index.

The proof for $A_3$ is very similar to $A_2$ but instead of fixing $a$ to get a certain incline, one fixes $b$ and uses the cases where $a < 0$ and $b \in [0.5, 1]$.

Lastly we will look at
\begin{align*}
    A_4 = A_{4, 1} \cup A_{4, 2} &= \{(\theta, I_2) | \theta \in [0.25, 0.5], I_2 \in [(-4\theta^2 + 4\theta - 0.5)/ 2\theta, 2 \theta]\} \cup\\
    &\{(\theta, I_2) | \theta \in [0.5, 0.75], I_2 \in [(-4\theta^2 + 4\theta - 0.5)/(2 - 2\theta), 2 - 2\theta]\}.
\end{align*}
Here we will use the cases where $a \in [0, 0.5]$ and $b \in [0.5, 1]$. When fixing either $a$ or $b$, $I_2$ is not a linear function of $\theta$. However, one can derive the inverse function of $(a, b) \mapsto (\theta, I_2)$ which is
\begin{align*}
    a &= \frac{-2\theta + I_2 + 1 \pm \sqrt{4\theta(1 - \theta) + I_2(I_2 -2 )}}{2} \\
    b &= 2 - a - 2\theta = \frac{-2\theta - I_2 + 3 \pm \sqrt{4\theta(1 - \theta) + I_2(I_2 -2 )}}{2}
\end{align*}
Since this is only a valid inverse function for $a \in [0, 0.5]$ and $b \in [0.5, 1]$, it remains to be shown that for all $(\theta, I_2) \in A_4$, the corresponding $a$ and $b$ are actually in those intervals. We will prove this for $A_{4, 1}$, the proof for $A_{4, 2}$ is similar. For this we use for $a$ the formula
\begin{align*}
    a &= \frac{-2\theta + I_2 + 1 - \sqrt{4\theta(1 - \theta) + I_2(I_2 -2 )}}{2}
\end{align*}
Since $I_2 \leq 1$, $a$ is non-decreasing in $I_2$. 
Moreover, $I_2$ lies between $(-4\theta^2 + 4 \theta - 0.5)/2\theta$ and $2\theta$. 
Substituting the lower limit of $I_2$ into the formula for $a$ yields $a = 1 - 2\theta$ as the lower limit for $a$. This must lie between $0$ and $0.5$ since $\theta$ is between $0.25$ and $0.5$. 
Substituting $2\theta$ for $I_2$ gives $0.5$ as the upper limit for $a$. Similarly, $b$ is non-increasing in $I_2$ and must lie between $0.5$ and $1$ when plugging in $2\theta$ for $I_2$ and is equal to $1$ when plugging in the lower limit of $I_2$.

\subsection{Proof of Theorem \ref{thm_asym}  } \label{proofthm1}
To prove the result of the theorem, we us the Cramér-Wold device (e.g., \cite{vaart_1998}, p.16), proving the multivariate normality of the random vector in \eqref{asymp_vec} by proving the equivalent statement that all linear combinations of the vector are univariate normal. \\
To this end, let $\mathbb{D}=D[0,1]$ be the Skorokhod space (the space of all càdlàg functions on $[0,1]$) and $\mathbb{E}=\mathbb{R}$. We define the map
\begin{align} \label{phi}
\phi: \mathbb{D} \rightarrow \mathbb{E}, \phi(f)= \lambda_1 \Biggl( \int_0^1 f \biggl(1- \frac{\alpha}{2} \biggl )d \alpha- \int_0^1 f \biggl(\frac{\alpha}{2} \biggl) d \alpha \Biggl ) + \lambda_2 \Biggl ( \int_0^1 f(\alpha) d \alpha \Biggl ),
\end{align}
where $\lambda_1, \lambda_2 \in \mathbb{R}$. It is clear that $\phi$ is a continuous linear map.
Of course in general $\hat{G}_m \circ \hat{F}_n^{-1}$ is not a càdlàg function, but as a monotonically decreasing function, $\hat{G}_m \circ \hat{F}_n^{-1}$ can be transformed into a càdlàg function with identical values of $ \phi$. 
For notational simplicity we still write $\hat{G}_m \circ \hat{F}_n^{-1}$ instead of their transformation.  \\ 
\cite{hsieh_turnbull} showed (Theorem 2.2) that under the conditions mentioned above there exists a probability space on which one can define sequences of two independent Brownian bridges $ B_1^n(t), B_2^n(t), t \in [0,1]$ such that
\begin{align} \label{thm_22}
 \sqrt{n}( \hat{G}_m(\hat{F}_n^{-1}(t))- G(F^{-1}(t))) = & \sqrt{\nu} \ B_1^n(G(F^{-1}(t))) + \frac{g(F^{-1}(t))}{f(F^{-1}(t))}   B_2^n(t) \\
 &+ o(n^{-1/2}( \log n)^2) \notag
\end{align}
almost surely, uniformly on $[a,b]$. As $a$ and $b$ were chosen arbitrarily, this is true for every subinterval of $[0,1]$. For simplicity we omit in the following the superscript $n$, but we will keep in mind that $B_1$ and $B_2$ still depend on the sample size $n$.\\
 As $\phi$ is a linear continuous map  we can apply the continuous mapping theorem, the delta method for empirical processes (Theorem 3.9.4 in \cite{van1996weak}) and Lemma 3.9.8 in \cite{van1996weak} on $ \phi (\sqrt{n}( \hat{G}_m(\hat{F}_n^{-1}(t))- G(F^{-1}(t))))$ and \eqref{thm_22}. This gives us the normality of all linear combinations of the components of \eqref{asymp_vec}, proving the asymptotic normality of the vector itself. \\
 It is clear that the expectation of the first component of the vector \eqref{asymp_vec} is $0$, the result for the second component follows due to similar arguments as in the proof of Theorem 2.13 in \cite{parkinson2018fast}. \\
 The variances can be calculated as
 \begin{align*}
 \sigma_{\theta}&= \text{var} \Biggl[ \sqrt{\nu} \int_0^1 \ B_1(G(F^{-1}(t))) dt + \int_0^1 \frac{g(F^{-1}(t))}{f(F^{-1}(t))}   B_2(t) dt \Biggl ] \\
 &= \nu \text{ var} \Biggl[ \int_0^1 \ B_1(G(F^{-1}(t))) dt \Biggl ] + \text{ var} \Biggl[ \int_0^1 \ B_2(F(G^{-1}(t))) dt \Biggl ] \\
 &=\nu ||G \circ F^{-1}||^*+ ||F \circ G^{-1}||^*
 \end{align*}
 and 
 \begin{align*}
    \sigma_{I_2} &= \nu \text{ var} \Biggl[ \int_0^1 \ B_1(G(F^{-1}(1-t/2)))-\ B_1(G(F^{-1}(t/2))) dt \Biggl ]  \\
   & \  \ \ + \text{ var} \Biggl[ \int_0^1 \frac{g(F^{-1}(1-t/2))}{f(F^{-1}(1-t/2))}   B_2(1-t/2) - \frac{g(F^{-1}(t/2))}{f(F^{-1}(t/2))}   B_2(t/2) dt \Biggl ]        \\ 
    &= \nu \text{ var} \Biggl[ \int_0^1 \ B_1(G(F^{-1}(1-t/2)))-\ B_1(G(F^{-1}(t/2))) dt \Biggl ]  \\
   & \  \ \ + \text{ var} \Biggl[ \int_0^1 \ B_2(F(G^{-1}(1-t/2))) - B_2(F(G^{-1}(t/2))) dt \Biggl ]        \\ 
   &= \nu \Bigl ( ||G \circ F_1^{-1}||^* +||G \circ F_2^{-1}||^* -2   \langle G \circ F_1^{-1}, G \circ F_2^{-1} \rangle ^*   \Bigl ) \\
  & \  \ \  +  ||F \circ G_1^{-1}||^* +||F \circ G_2^{-1}||^* -2   \langle F \circ G_1^{-1}, F \circ G_2^{-1} \rangle ^*.  
 \end{align*}
Incorporating the independence of $B_1$ and $B_2$, we obtain the covariance
\begin{align*}
     \sigma_{\theta,I_2} &=   \nu \text{ cov}  \Biggl[ \int_0^1 \ B_1(G(F^{-1}(t))) dt, \int_0^1 \ B_1(G(F^{-1}(1-t/2)))-\ B_1(G(F^{-1}(t/2))) dt \Biggl ]  \\   
   & \  \ \ +  \text{cov} \Biggl[ \int_0^1 \ \frac{g(F^{-1}(t))}{f(F^{-1}(t))}   B_2(t) dt  ,  \int_0^1 \ \frac{g(F^{-1}(1-t/2))}{f(F^{-1}(1-t/2))}   B_2(1-t/2) -\frac{g(F^{-1}(t/2))}{f(F^{-1}(t/2))}   B_2(t/2) dt   \Biggl ]  \\
     &=   \nu \text{ cov}  \Biggl[ \int_0^1 \ B_1(G(F^{-1}(t))) dt, \int_0^1 \ B_1(G(F^{-1}(1-t/2)))-\ B_1(G(F^{-1}(t/2))) dt \Biggl ]  \\   
   & \  \ \ +  \text{cov} \Biggl[ \int_0^1 \ B_2(F(G^{-1}(t))) dt  ,  \int_0^1 \ B_2(F(G^{-1}(1-t/2))) - B_2(F(G^{-1}(t/2))) dt   \Biggl ]  \\
  &=  \nu \Bigl (  \langle G \circ F^{-1}, G \circ F_2^{-1} \rangle ^* - \langle G \circ F^{-1}, G \circ F_1^{-1} \rangle ^* \Bigl ) \\ 
  & \  \ \ + \langle F \circ G^{-1}, F \circ G_2^{-1} \rangle ^* -\langle F \circ G^{-1}, F \circ G_1^{-1} \rangle ^*,
\end{align*}
where 
\begin{align*}
    ||h||^* &= \int_0^1 h^2(t) dt - \Bigl ( \int_0^1 h(t) dt \Bigl )^2 \\
    \langle h, \Tilde{h} \rangle ^* &= 
    %\int_0^1 \int_0^1 h(t) \Tilde{h}(s) dt ds
    \int_0^1 h(t) \Tilde{h}(t) dt 
    - \Bigl ( \int_0^1 h(t) dt \Bigl ) \Bigl ( \int_0^1 \Tilde{h}(s) ds \Bigl ).
\end{align*}

\subsection{Proof of Theorem \ref{thm_boot}} \label{proodthm2}
Similar to the proof of Theorem \ref{thm_asym} we consider again the convergence of all the linear combinations of the vector \eqref{bootstrap_process} in the main paper instead of the equivalent convergence of the vector itself. Therefore we have to prove that $ \sqrt{n} ( \phi (\hat{G}_m^{*} \circ \hat{F}_n^{* -1} )-\phi (\hat{G}_m \circ \hat{F}_n^{ -1} ) )$ is asymptotically normally distributed with identical limit distribution as $ \sqrt{n} ( \phi (\hat{G}_m \circ \hat{F}_n^{ -1} )-\phi (G \circ F^{ -1} ) )$, where $\phi$ is again defined by:
\begin{align*} 
\phi: \mathbb{D} \rightarrow \mathbb{E}, \phi(f)= \lambda_1 \Biggl( \int_0^1 f \biggl(1- \frac{\alpha}{2} \biggl )d \alpha- \int_0^1 f \biggl(\frac{\alpha}{2} \biggl) d \alpha \Biggl ) + \lambda_2 \Biggl ( \int_0^1 f(\alpha) d \alpha \Biggl ),
\end{align*}
with $\lambda_1, \lambda_2 \in \mathbb{R}$. \\
We will prove this by applying the Delta-method for empirical bootstrap processes (e.g., Theorem 3.9.11 in \cite{van1996weak}). To check the assumptions of the theorem we use a similar construction of the function set $ \mathscr{F}$ as in the proof of Theorem 2.16 in \cite{parkinson2018fast}, which is defined as
\begin{align*}
    f_n(X,Y,t,i)= \hat{G}_m(X_i) \mathbbm{1}_{ \{ \hat{F}_n(X_i) \leq t \}} \prod_{j=1, j \neq i}^n (1- \mathbbm{1}_{ \{ \hat{F}_n(X_i) \leq t \}} \mathbbm{1}_{ \{ X_j >X_i  \}} ),
\end{align*}
with $t \in [0,1]$, and $\hat{G}_m(X_i)$ and $\hat{F}_n(X_i)$ the empirical distribution functions evaluated at the random value $X_i$. 
Note that $f_n$ is just the empirical version of $G$ evaluated at the quantile of the empirical version of $F$. As the set of all indicator functions in $\mathbb{R}$ is a Donsker class and the product of two Donsker classes is again a Donsker class (e.g., Section 2 in \cite{van1996weak}), it is clear that the set $ \mathscr{F}$ is also a Donsker class, i.e.
\begin{align*}
\sqrt{n} \Biggl ( ( \frac{1}{n} \sum_{i=1}^n f_n(X,Y,t,i) -G(F^{-1}(t)) \Biggl ) \xrightarrow{w}  \mathbbm{G},
\end{align*}
where $\mathbbm{G}$ is a tight Borel measurable element. \\
Hence we can apply Theorem 3.6.1 in \cite{van1996weak} and get
\begin{align*}
   \sup_{h \in BL} |E [h(\mathbbm{G}_n)]- E [h(\mathbbm{G})] |  \xrightarrow{P^*} 0,
\end{align*}
where $\mathbbm{G}_n = \sqrt{n} ( \hat{G}_m^{*} \circ \hat{F}_n^{* -1} -\hat{G}_m \circ \hat{F}_n^{ -1} )$. Here $ \xrightarrow{P^*}$ means weak convergence in outer probability and $BL$ is the set of all bounded Lipschitz functions on $\mathbb{R}$. \\
Due to the Hadamard-differentiability of $\phi$, we can now apply the Delta-method for empirical bootstrap processes (e.g., Theorem 3.9.11 in \cite{van1996weak}). 
Therefore it holds that $ \sqrt{n} ( \phi (\hat{G}_m^{*} \circ \hat{F}_n^{* -1} )-\phi (\hat{G}_m \circ \hat{F}_n^{ -1} ) )$ is asymptotically normally distributed with identical limit distribution as $ \sqrt{n} ( \phi (\hat{G}_m \circ \hat{F}_n^{ -1} )-\phi (G_m \circ F_n^{ -1} ) )$ . \\
The asymptotically normal distribution of the bootstrap estimator of the vector $ (\theta, I_2)$ follows again due to the Cramér-Wold device.

\subsection{Proof of Proposition \ref{var_est}} \label{proofprop}
First note that under the null hypothesis of distribution equality, the estimators of the two random variables $P(X^{(1)} < Y^{(1)})$ and $P(X^{(2)} < Y^{(2)})$ are independent. Furthermore we have the two independent samples
\begin{align*}
    X_1, ..., X_k \overset{i.i.d.}{\sim} F_1
\end{align*}
and
\begin{align*}
    Y_1, ..., Y_l \overset{i.i.d.}{\sim} G_1.
\end{align*}
Using the results in Section 3 of \cite{brunner2017rank}, we can derive that $\sqrt{k + l} (\int \hat{F}_1 d\hat{G}_1 - \int F_1 dG_1)$ follows asymptotically the same distribution as
\begin{equation*}
    U_{k + l} = \sqrt{k + l} \left( \frac{1}{l} \sum_{i = 1}^{l} \left[F_1(Y_{i}) - \int F_1 dG_1 \right] - \frac{1}{k} \sum_{i = 1}^{k} \left[G_1(X_i) - \int G_1 dF_1 \right] \right)
\end{equation*}
with $Var(U_{k + l}) = (k + l) (\frac{\sigma_{X_1}^{2}}{k} + \frac{\sigma_{Y_1}^{2}}{l})$ where $\sigma_{X_1}^2 = Var(G_1(X_1))$ and $\sigma_{Y_1}^2 = Var(F_1(Y_1))$. These quantities can be consistently estimated by $s_{X_1}^2$ and $s_{Y_1}^2$, and $\sqrt{k + l}(\int \hat{F}_1 d \hat{G}_1 - \int F_1 dG_1) / \sqrt{Var(U_{k + l})}$ follows asymptotically a standard normal distribution. \\
Similar arguments and the independence of the two samples
\begin{align*}
    X_{k + 1}, ..., X_n \overset{i.i.d.}{\sim} F_2
\end{align*}
and 
\begin{align*}
  Y_{l + 1}, ..., Y_m \overset{i.i.d.}{\sim} G_2  
\end{align*} lead to the asymptotic equivalence of $\sqrt{(n-k)+(m-l)}(\int \hat{F}_2 d\hat{G}_2 - \int F_2 dG_2)$ and 
\begin{align*}
    U_{(n-k)+(m-l)} = \sqrt{(n-k)+(m-l)}  \Biggl( & \frac{1}{m-l} \sum_{i = l+1}^{m} \left[F_2(Y_{i}) - \int F_2 dG_2 \right ] \\
    &- \frac{1}{n-k} \sum_{i = k+1}^{n} \left[G_2(X_i) - \int G_2 dF_2 \right ] \Biggl).
\end{align*} This quantity follows again asymptotically a standard normal distribution. 
The variance $Var(U_{(n-k)+(m-l)}) = ((n-k)+(m-l)) (\frac{\sigma_{X_2}^{2}}{n-k} + \frac{\sigma_{Y_2}^{2}}{m-l})$, where $\sigma_{X_2}^2 = Var(G_2(X_2))$ and $\sigma_{Y_2}^2 = Var(F_2(Y_2))$, can be consistently estimated by $s_{X_2}^2$ and $s_{Y_2}^2$. \\
Incorporating these results into the equations \eqref{theta_medeq} and \eqref{I2_medeq} we get (w.l.o.g. $m$ and $n$ are even):
\begin{align*}
   & Var (\sqrt{m+n}(\hat{I}_2^{\text{adj}}- I_2^{\text{adj}}) \\
    &= Var \Biggl ( \frac{\sqrt{m+n}}{2} \Biggl (\int \hat{F}_2 d\hat{G}_2 - \int F_2 dG_2 \Biggl)- \frac{\sqrt{m+n}}{2} \Biggl(\int \hat{F}_1 d\hat{G}_1 - \int F_1 dG_1 \Biggl) \Biggl ) \\
    &= \frac{1}{2}Var \Biggl ( \sqrt{\frac{m+n}{2}} \Biggl (\int \hat{F}_2 d\hat{G}_2 - \int F_2 dG_2 \Biggl)- \sqrt{\frac{m+n}{2}} \Biggl(\int \hat{F}_1 d\hat{G}_1 - \int F_1 dG_1 \Biggl) \Biggl ) \\
    &= \frac{1}{2}\Biggl ((L+K) \Bigl ( \frac{s_{X_1}^2}{K}+ \frac{s_{Y_1}^2}{L} \Bigl) +(n+m-L-K) \Bigl ( \frac{s_{X_2}^2}{n-K}+ \frac{s_{Y_2}^2}{m-L} \Bigl),
\end{align*}
using the fact that $L+k=n+m-L-K=\frac{m+n}{2}$ under the null hypothesis $H_0 \ : \ F=G$. \\
The asymptotic normality of the estimators $\hat{\theta}^{(1)}_{mn}$ and $\hat{\theta}^{(2)}_{mn}$ of the random variables $P(X^{(1)} < Y^{(1)})$ and $P(X^{(2)} < Y^{(2)})$ follows directly from Result 3.18 in \cite{brunnerbathkekonnietschke}. Hence $\hat{\theta}^{(1)}_{mn}$ and $\hat{\theta}^{(2)}_{mn}$ are independent asymptotically normal distributed random variables, thus a linear combination of $\hat{\theta}^{(1)}_{mn}$ and $\hat{\theta}^{(2)}_{mn}$ follows also an asymptotically normal distribution. This and again the Cramér-Wold device prove the asymptotic normality of the test statistic in equation \eqref{statistic_sim}.

\section{Additional Simulations}
 \label{add_sim_ci}
In this section, we empirically compare the different procedures for computing simultaneous confidence regions for the relative effect and the overlap index, as introduced in Section \ref{secresamp} and Section \ref{alt_ci}, 
using different underlying distributions and different sample sizes. 

The ``exact'' values of $\theta$ and $I_2$ for each combination of distributions can be computed by numerical integration. 
They are given as follows, rounded to four decimal places:\\
 $X \sim N(0,1)$, $Y \sim N(1,1)$.
Exact values: $\theta = 0.7602$ and $I_2 =0.3645$. \\
$X \sim N(0,2)$, $Y \sim U[-0.5,0.5]$. 
Exact values: $\theta = 0.5$ and $I_2 = 0.9008$. \\
$X \sim N(1,1)$, $Y \sim U[-0.5,0.5]$. 
Exact values: $\theta = 0.8315$ and $I_2 =  0.3370$. \\
$X \sim N(2,1)$, $Y \sim U[-0.5,0.5]$. 
Exact values: $\theta =  0.9727$ and $I_2 =  0.0546$. \\
$X \sim N(1,1)$, $Y \sim $ Exp$(1)$. 
Exact values: $\theta =  0.5381 $ and $I_2 =  0.5389$. \\
$X \sim N(2,1)$, $Y \sim $ Exp$(1)$. 
Exact values: $\theta =  0.7895$ and $I_2 =  0.2794$.  \\
For each setting, we simulated $1,000$ data sets, 
and for each data set, we generated $1,000$ bootstrap samples. 
For the calculation of the simulated coverage probability, a simulated confidence region was considered as ``covering`` if it covered both components of the parameter vector.
The length of each simulated confidence region was calculated by the Euclidean distance: 
\begin{align*}
\frac{1}{2} \sqrt{(\theta_{upper}- \theta_{lower})^2+(I_{2,upper}- I_{2,lower})^2}.
\end{align*}

\begin{figure} [H]
\begin{subfigure}{0.4\textwidth}
\includegraphics[width=\linewidth]{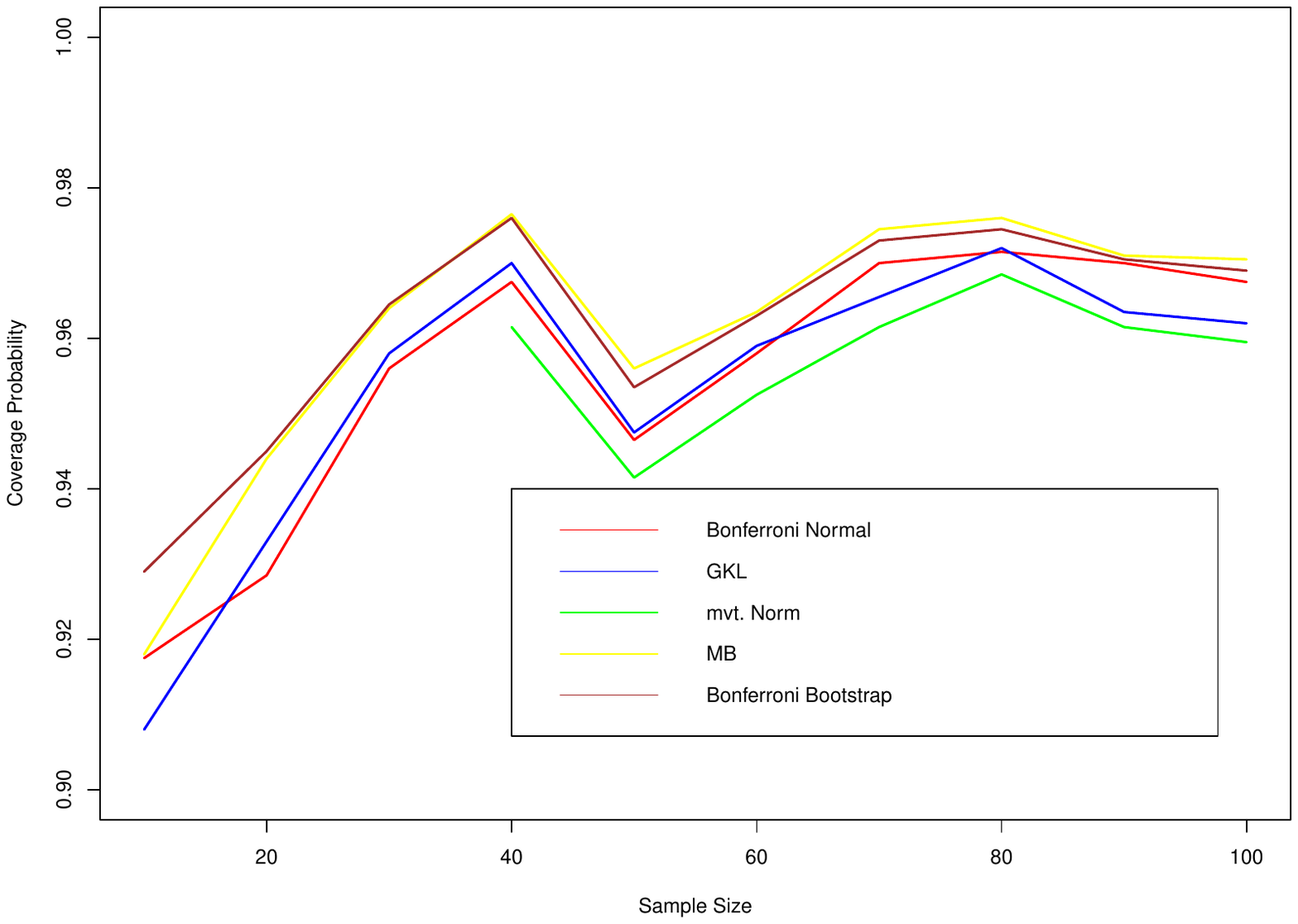}
\caption{Coverage probability for $N(0, 1)$ vs. $N(1, 1)$} \label{fig:a}
\end{subfigure}\hspace*{\fill}
\begin{subfigure}{0.4\textwidth}
\includegraphics[width=\linewidth]{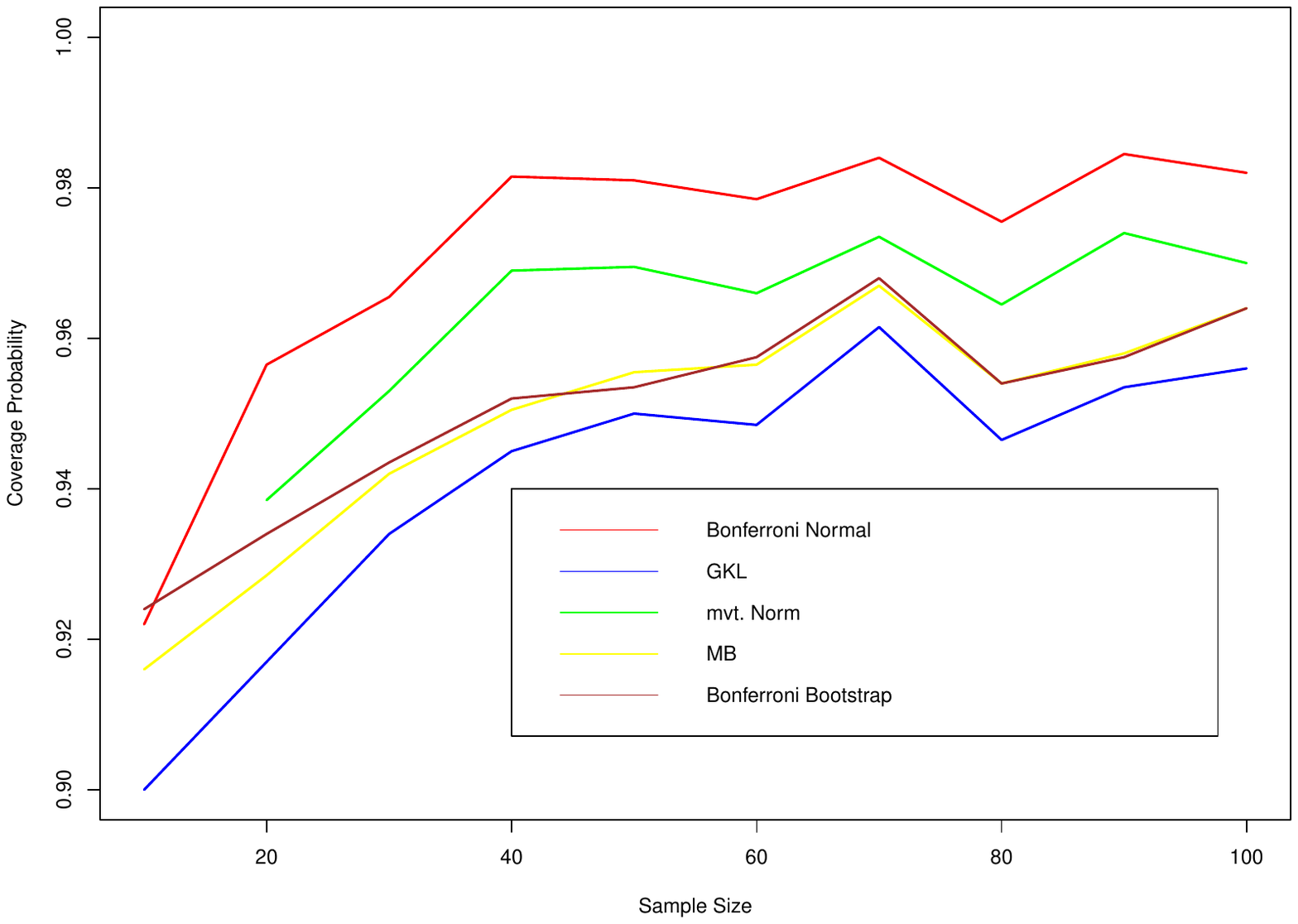}
\caption{Coverage probability for $N(0, 2)$ vs. $U[-0.5,0.5]$} \label{fig:b}
\end{subfigure}

\medskip
\begin{subfigure}{0.4\textwidth}
\includegraphics[width=\linewidth]{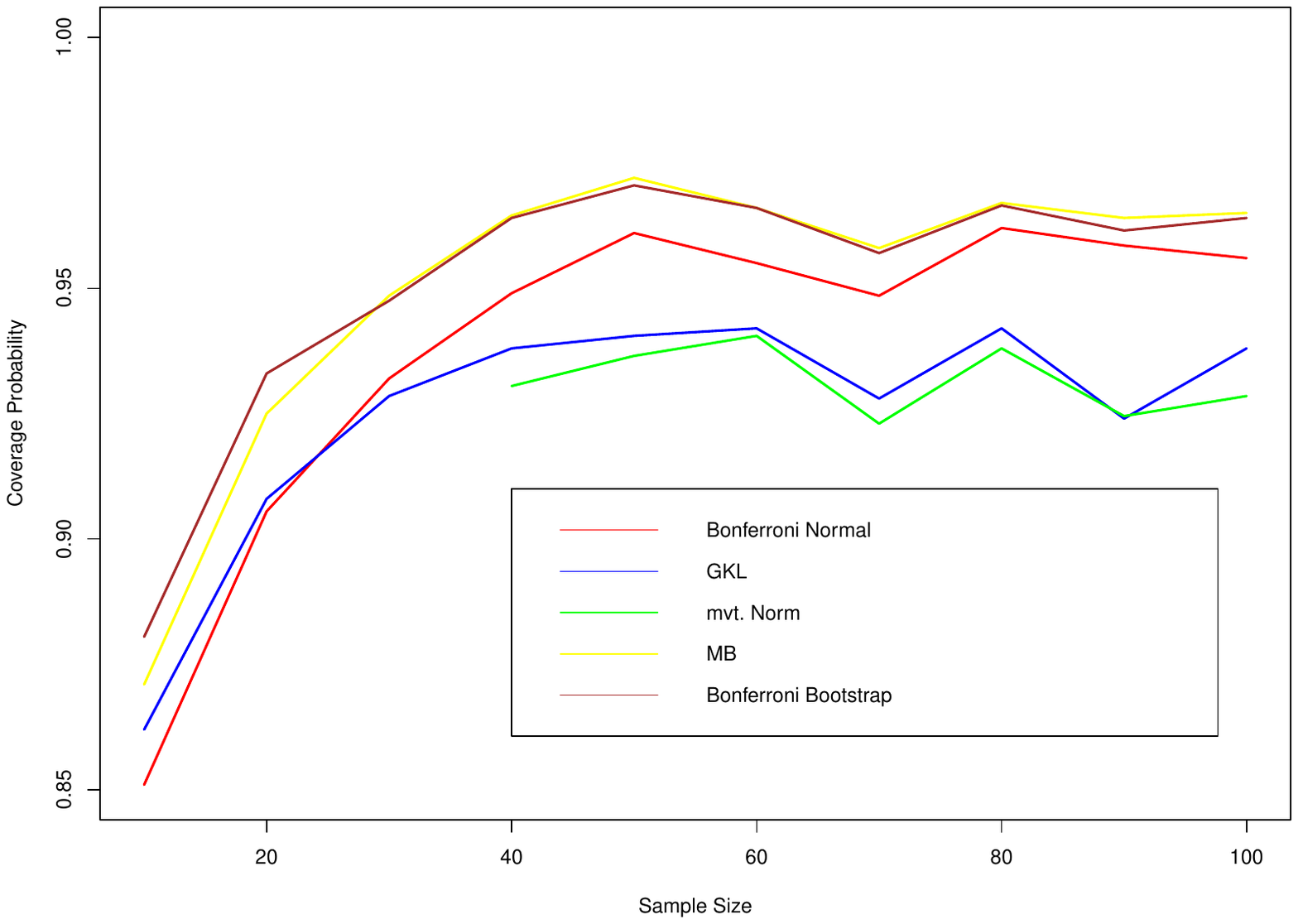}
\caption{Coverage probability for $N(1, 1)$ vs. $U[-0.5,0.5]$} \label{fig:c}
\end{subfigure}\hspace*{\fill}
\begin{subfigure}{0.4\textwidth}
\includegraphics[width=\linewidth]{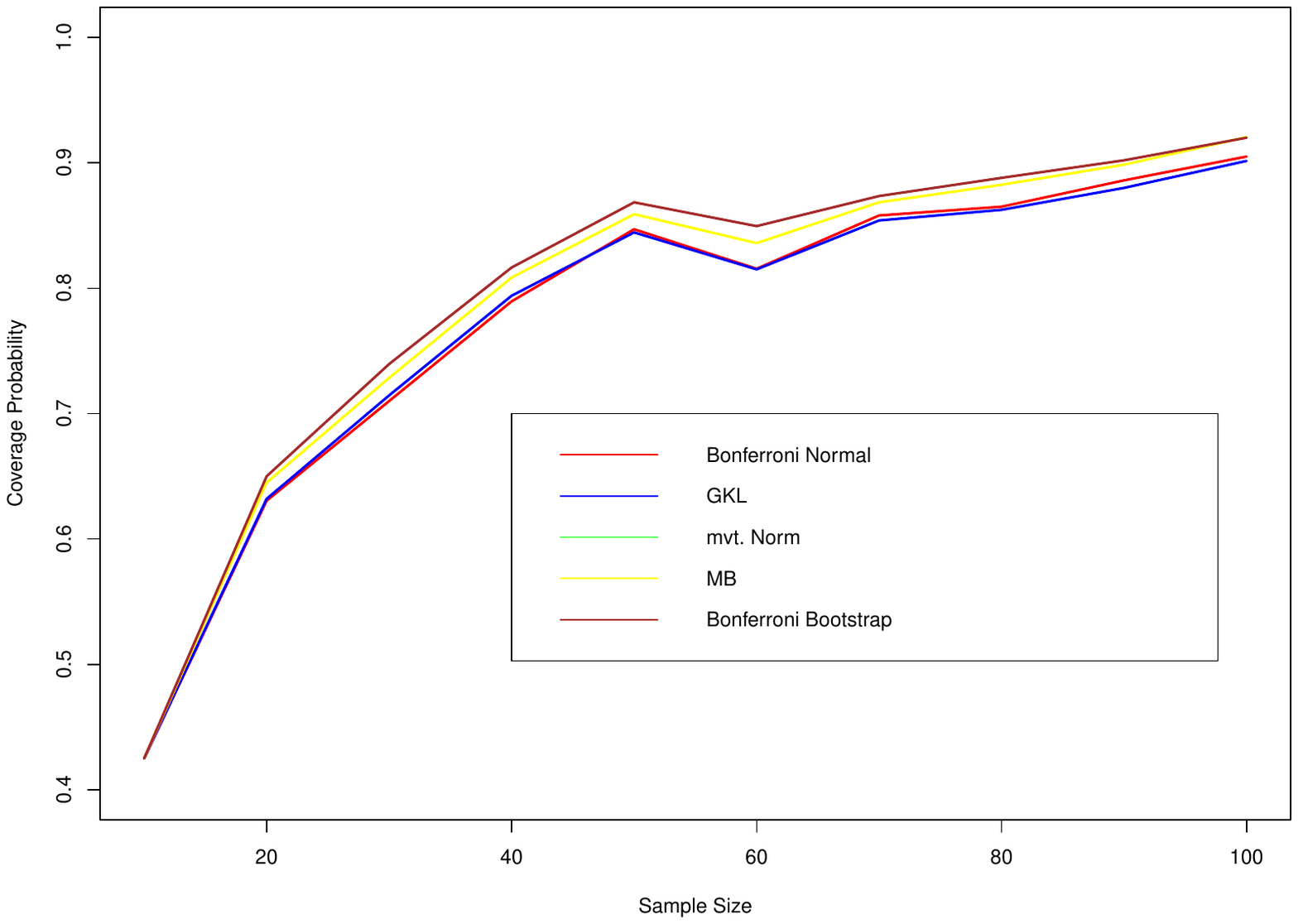}
\caption{Coverage probability for $N(2, 1)$ vs. $U[-0.5,0.5]$} \label{fig:d}
\end{subfigure}

\medskip
\begin{subfigure}{0.4\textwidth}
\includegraphics[width=\linewidth]{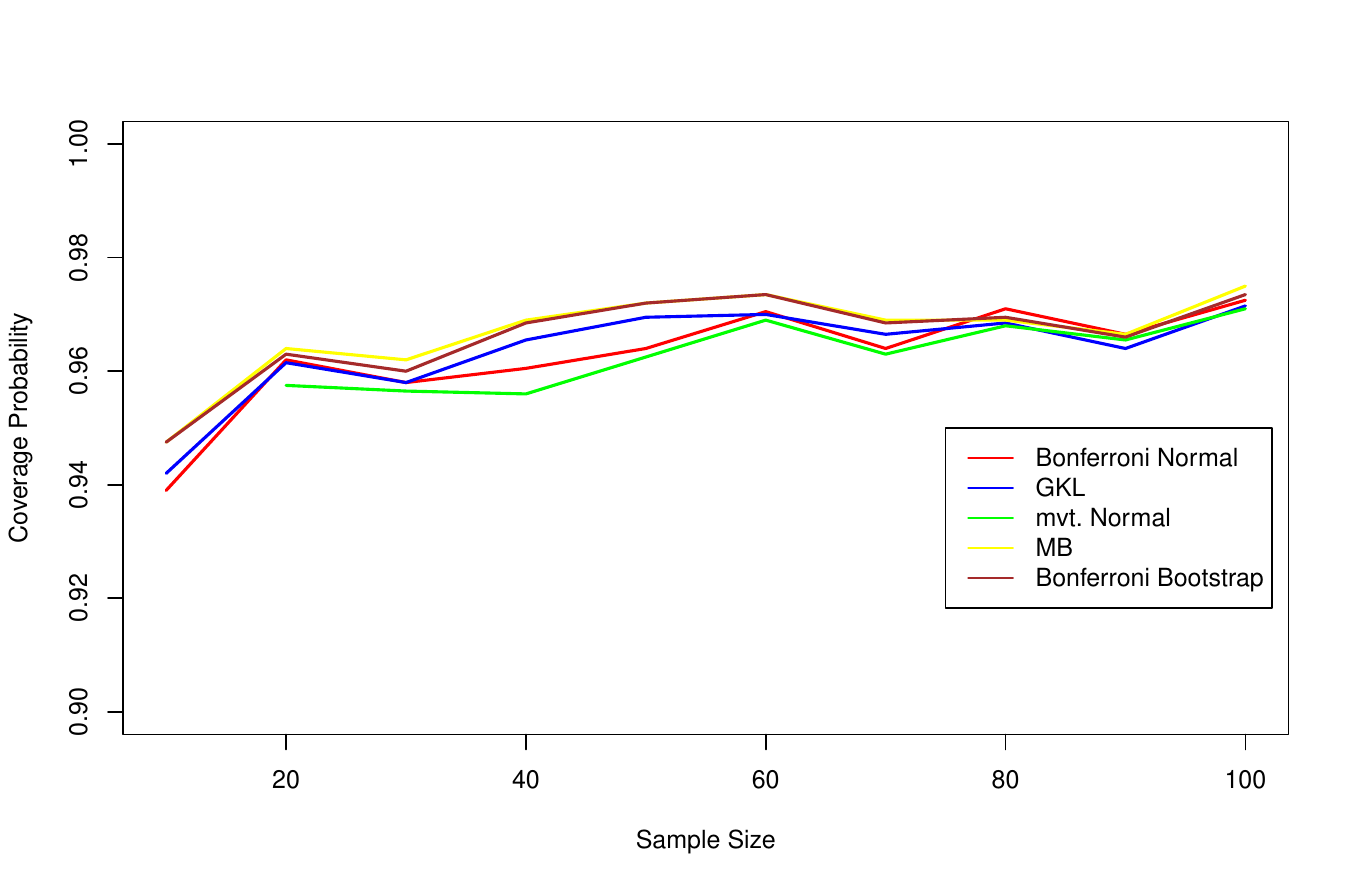}
\caption{Coverage probability for $N(1, 1)$ vs. Exp$(1)$} \label{fig:e}
\end{subfigure}\hspace*{\fill}
\begin{subfigure}{0.4\textwidth}
\includegraphics[width=\linewidth]{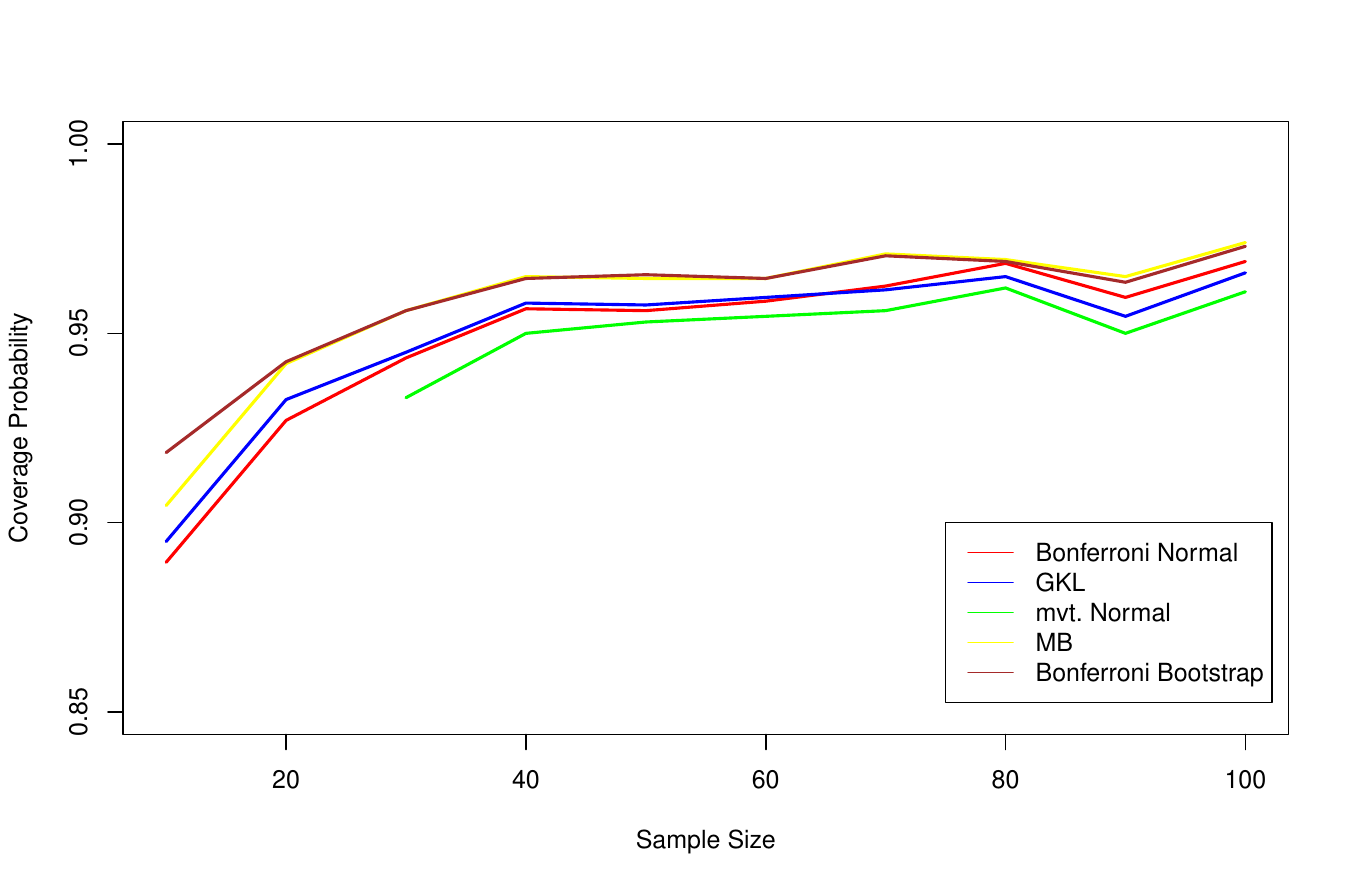}
\caption{Coverage probability for $N(2, 1)$ vs. Exp$(1)$} \label{fig:f}
\end{subfigure}
\caption{Empirical coverage probabilities for different combinations of underlying distributions}
\end{figure}

\begin{figure} [ht!]
\begin{subfigure}{0.48\textwidth}
\includegraphics[width=\linewidth]{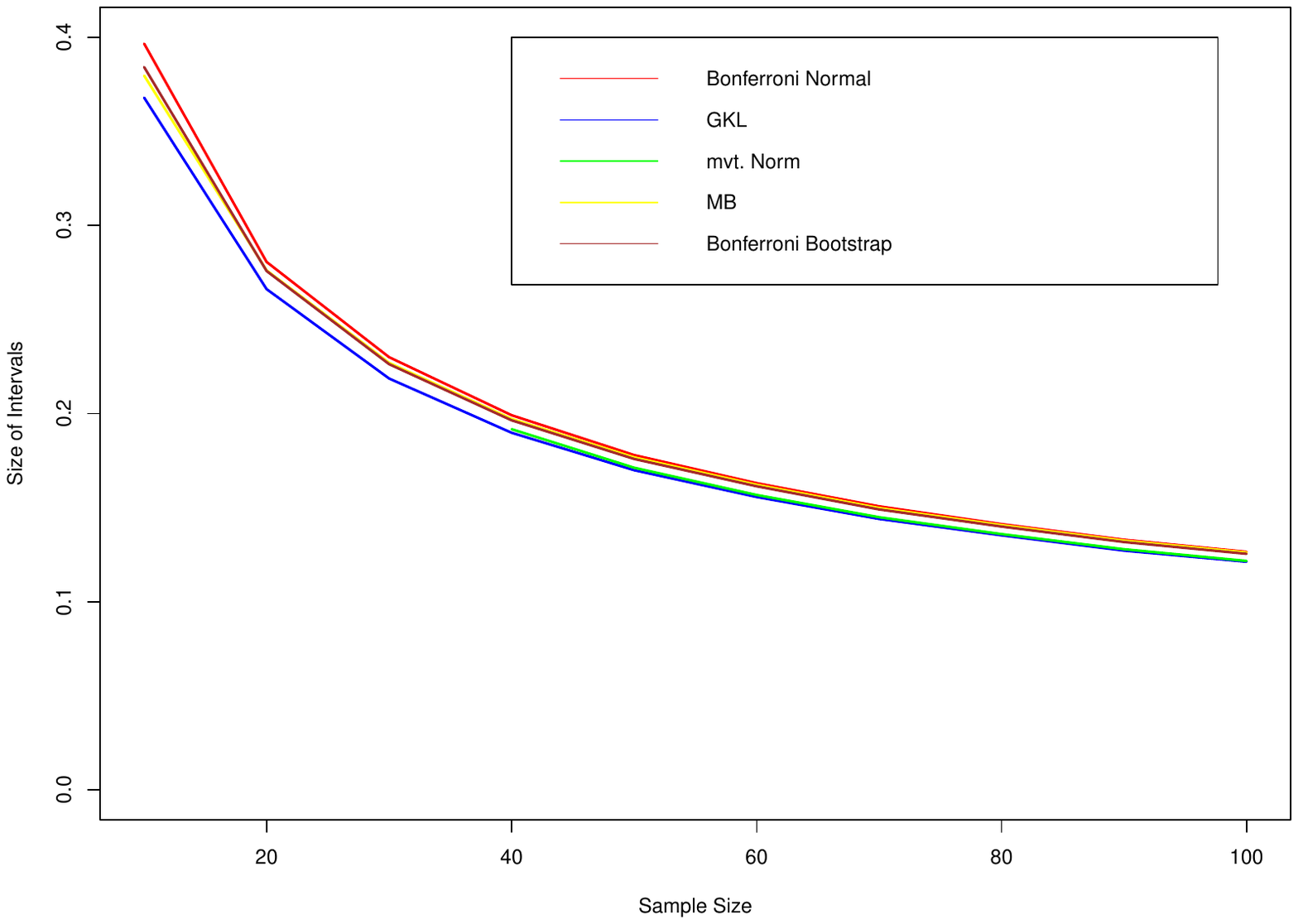}
\caption{Interval length for $N(0, 1)$ vs. $N(1, 1)$} \label{fig:a2}
\end{subfigure}\hspace*{\fill}
\begin{subfigure}{0.48\textwidth}
\includegraphics[width=\linewidth]{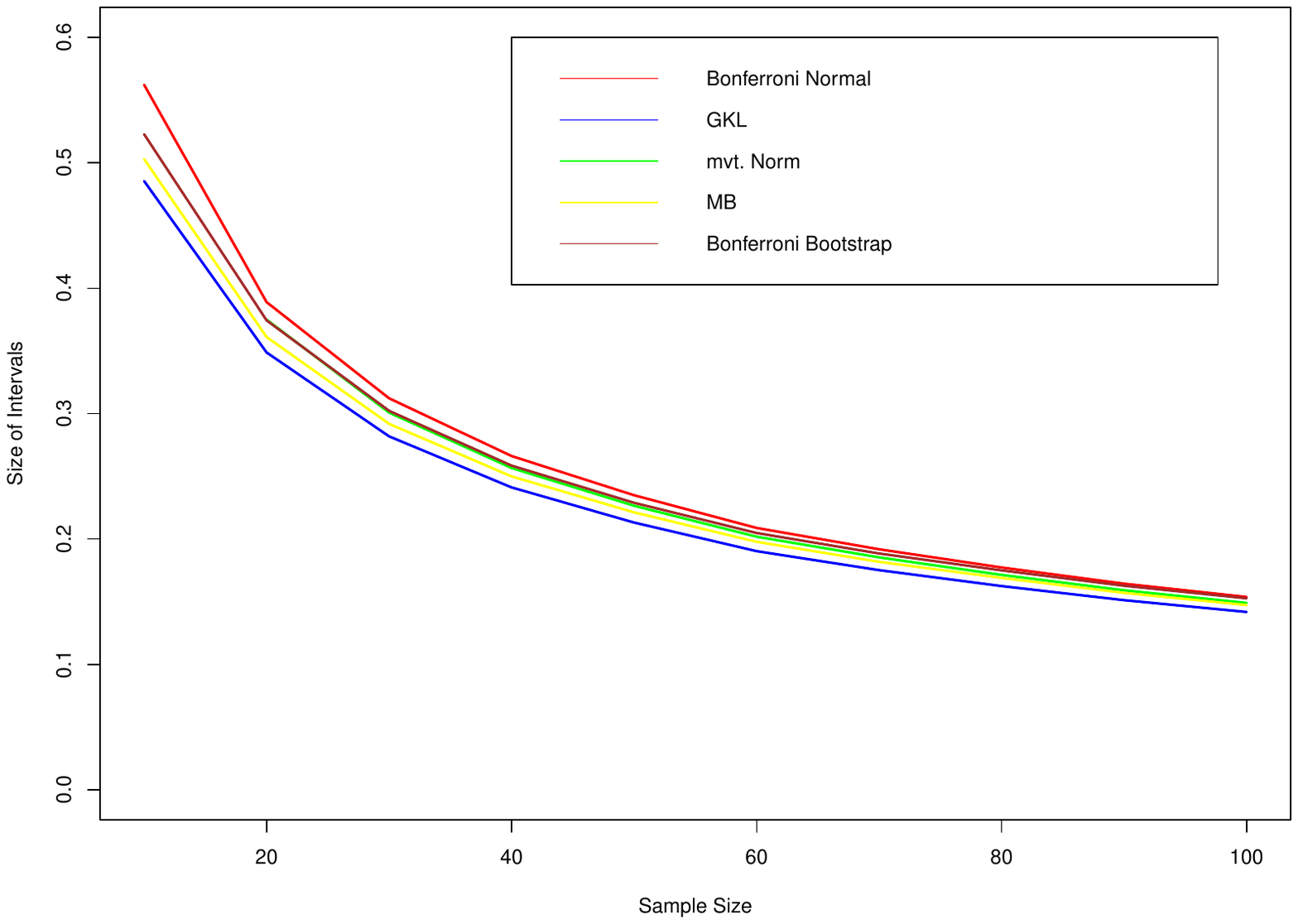}
\caption{Interval length for $N(0, 2)$ vs. $U[-0.5,0.5]$} \label{fig:b2}
\end{subfigure}

\medskip
\begin{subfigure}{0.48\textwidth}
\includegraphics[width=\linewidth]{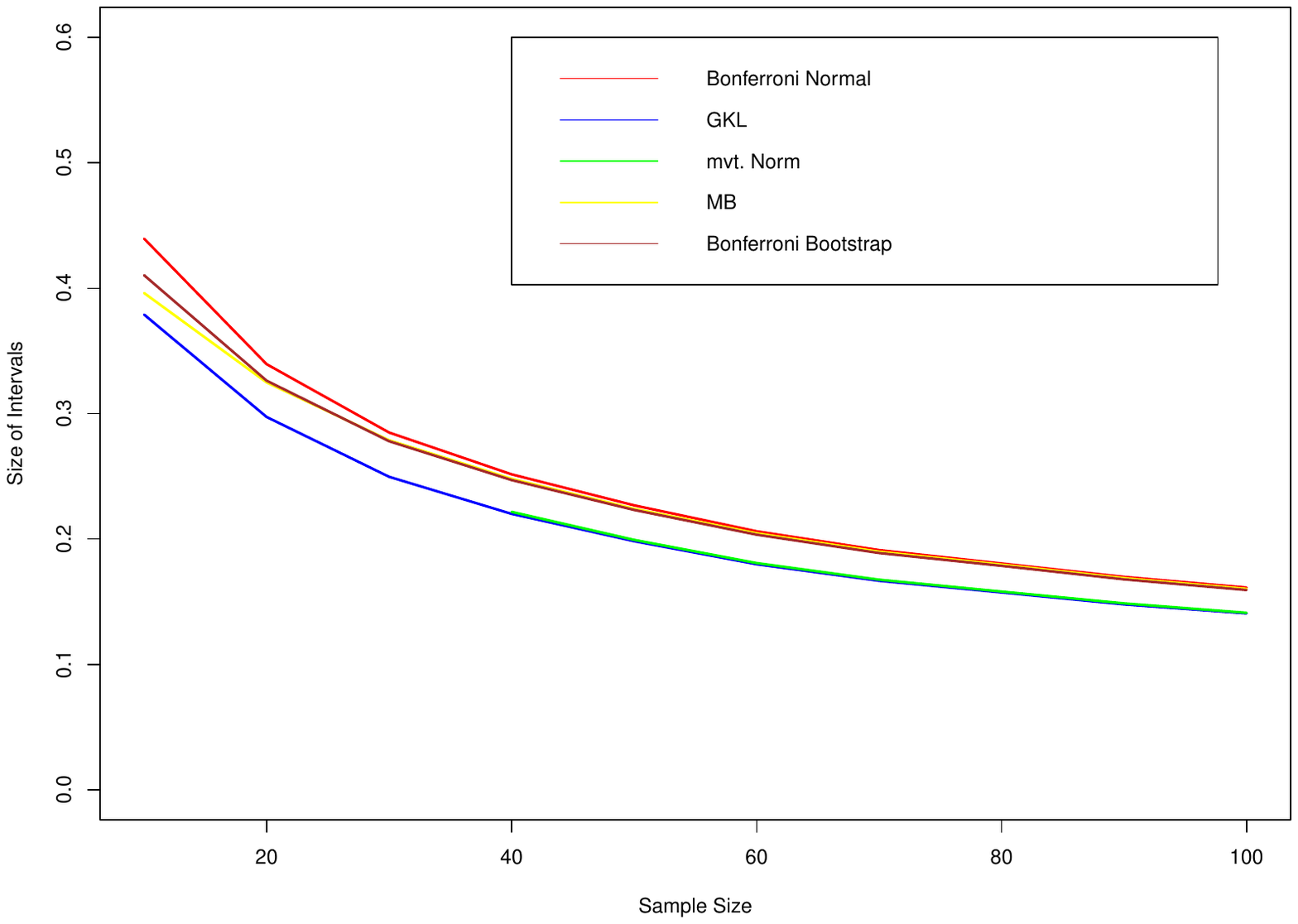}
\caption{Interval length for $N(1, 1)$ vs. $U[-0.5,0.5]$} \label{fig:c2}
\end{subfigure}\hspace*{\fill}
\begin{subfigure}{0.48\textwidth}
\includegraphics[width=\linewidth]{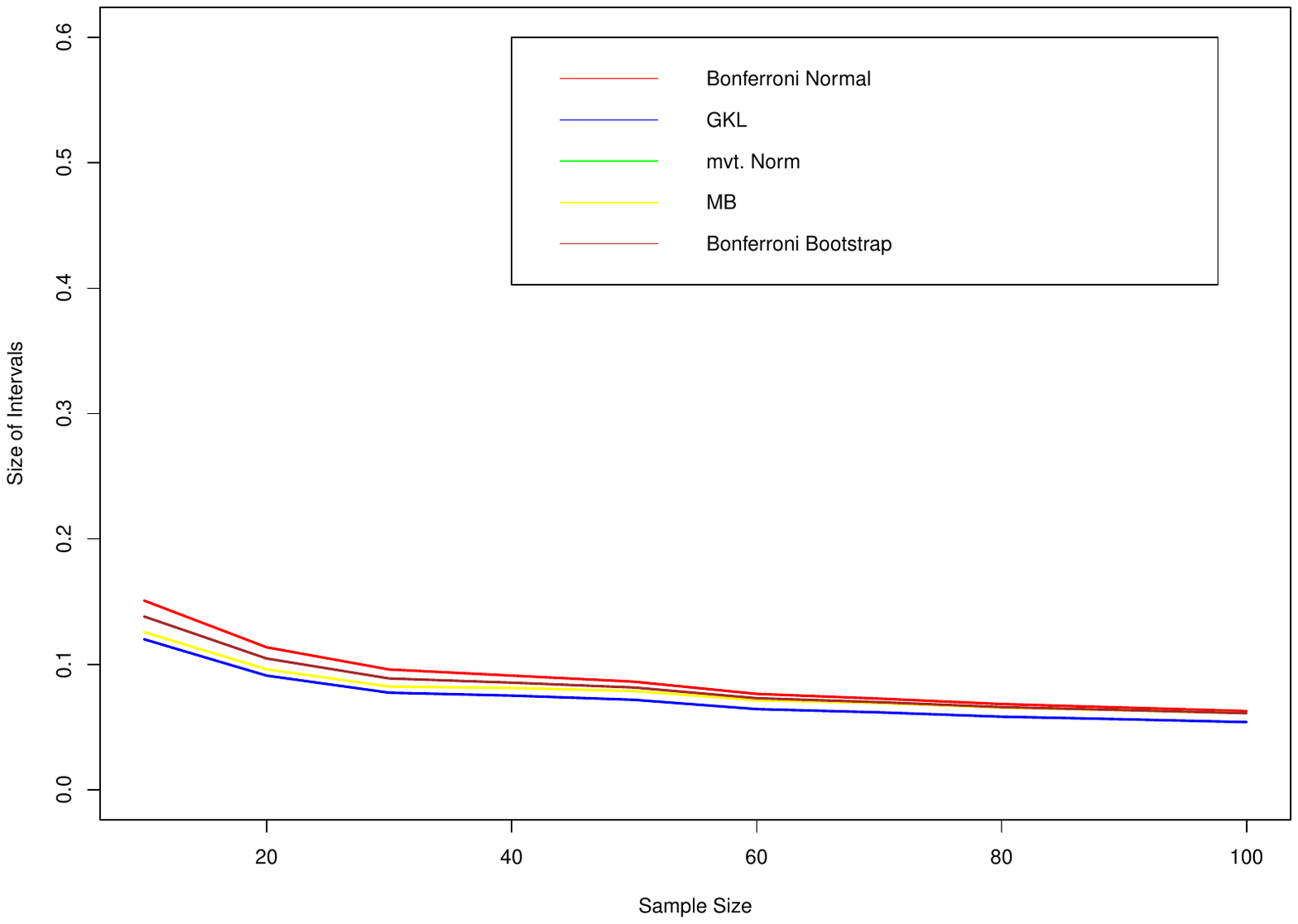}
\caption{Interval length for $N(2, 1)$ vs. $U[-0.5,0.5]$} \label{fig:d2}
\end{subfigure}

\medskip
\begin{subfigure}{0.48\textwidth}
\includegraphics[width=\linewidth]{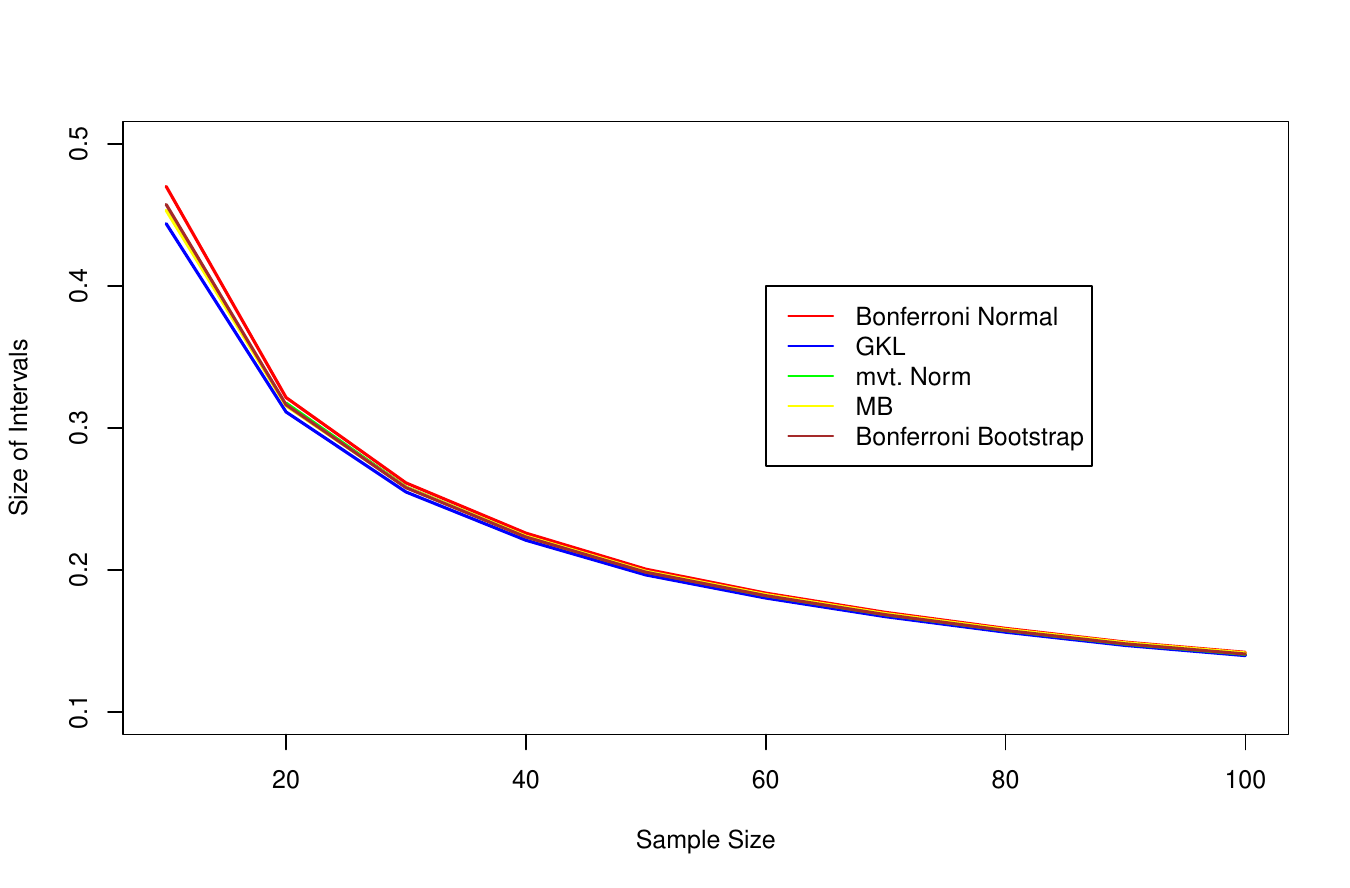}
\caption{Interval length for $N(1, 1)$ vs. Exp$(1)$} \label{fig:e2}
\end{subfigure}\hspace*{\fill}
\begin{subfigure}{0.48\textwidth}
\includegraphics[width=\linewidth]{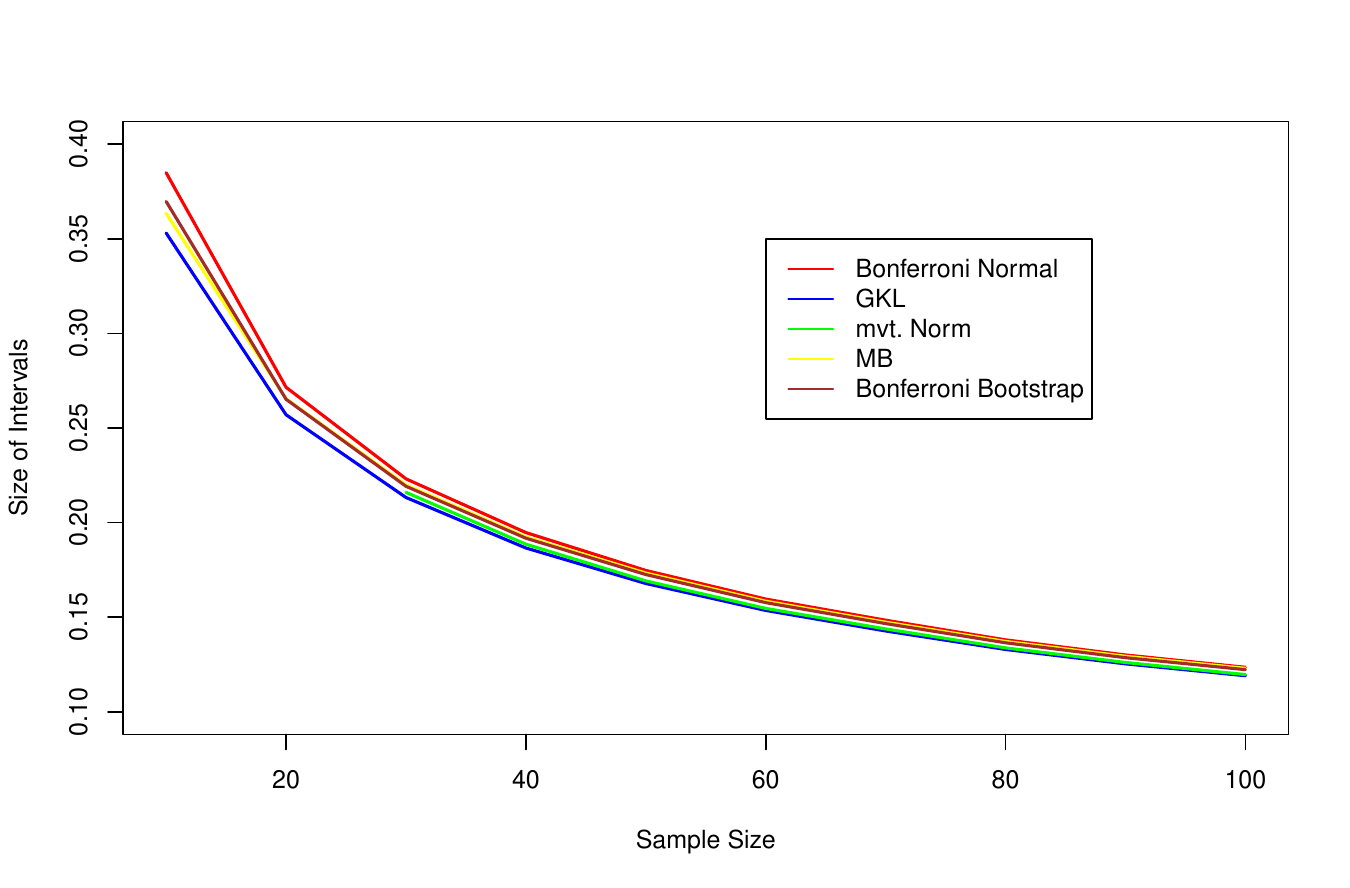}
\caption{Interval length for $N(2, 1)$ vs. Exp$(1)$} \label{fig:f2}
\end{subfigure}
\caption{Empirical interval lengths for different combinations of underlying distributions}
\end{figure}

All methods performed quite similarly. 
Only for small data sets, the method by \cite{gao_21} yielded smaller confidence intervals. 
In return, the method was in some cases also slightly liberal for small data sets. \\
In Example \ref{fig:d}, the exact values are rather extreme (close to 0 or 1), and in this situation, all methods produced a low simulated coverage probability for small samples. 
Therefore, we would recommend to avoid the usage of simultaneous confidence intervals in the case of extreme parameter values and very small samples.

\section*{Acknowledgments}
The authors gratefully acknowledge the support of the WISS project ``IDA-Lab Salzburg'' (20204-WISS/225/197-2019, 20102-F1901166-KZP).

\newpage
\bibliographystyle{apalike}
\bibliography{comb_tendency_overlap_preprint}  
\end{document}